\documentclass[11pt,fleqn]{amsart}

\RequirePackage{amsthm,amsmath,amsfonts,amssymb}
\RequirePackage[authoryear, round]{natbib}
\RequirePackage{graphicx}
\RequirePackage{mathrsfs, dsfont}
\usepackage{xcolor}
\usepackage{enumerate}
\usepackage{abbreviationsGeneral-ext}

\numberwithin{equation}{section}

\theoremstyle{remark}
\newtheorem{assumption}{Assumption}[section] 
\newtheorem{definition}[assumption]{Definition}
\newtheorem{example}[assumption]{Example}
\newtheorem{remark}[assumption]{Remark}

\theoremstyle{plain}
\newtheorem{corollary}[assumption]{Corollary}
\newtheorem{lemma}[assumption]{Lemma}
\newtheorem{proposition}[assumption]{Proposition}
\newtheorem{theorem}[assumption]{Theorem}

\newcommand{\ds}{\displaystyle}
\newcommand{\FF}{\emph{first-vs-full}}
\newcommand{\FL}{\emph{first-vs-last}}

\begin{document}


\title[U-statistic change-point tests under alternatives]{First versus full or first versus last: U-statistic change-point tests under fixed and local alternatives}
\author{Herold Dehling}
\author{Daniel Vogel}
\author{Martin Wendler}
\date{February 15, 2026}

\address{
Fakult\"at f\"ur Mathematik, Ruhr-Universit\"at Bo\-chum, 
44780 Bochum, Germany}
\email{herold.dehling@rub.de}

\address{
MEDICE, 
58239 Iserlohn, Germany}
\email{daniel.vogel@tu-dortmund.de}

\address{
Otto-von-Guericke-Universit{\"a}t Magdeburg, 
39106 Magdeburg, Germany}
\email{martin.wendler@ovgu.de}




\begin{abstract}
The use of U-statistics in the change-point context has received considerable attention in the literature. We compare two approaches of constructing CUSUM-type change-point tests, which we call the \FF\ and \FL\ approach. Both have been pursued by different authors. The question naturally arises if the two tests substantially differ and, if so, which of them is better in which data situation. In large samples, both tests are similar: they are asymptotically equivalent under the null hypothesis and under sequences of local alternatives. In small samples, there may be quite noticeable differences, which is in line with a different asymptotic behavior under fixed alternatives. We derive a simple criterion for deciding which test is more powerful. 
We examine the examples Gini's mean difference, the sample variance, and Kendall's tau in detail. 
Particularly, when testing for changes in scale by Gini's mean difference, we show that the \FF\ approach has a higher power if and only if the scale changes from a smaller to a larger value -- regardless of the population distribution or the location of the change. The asymptotic derivations are under weak dependence. The results are illustrated by numerical simulations and data examples.
\end{abstract}



\maketitle

\tableofcontents


\section{Introduction}

The classical approach for testing the constancy of the mean of a time series $X_1, \ldots, X_n$ is the CUSUM test based on the test statistic
\be \label{eq:cusum1}
	T_n = \max_{1 \le k \le n} \frac{k}{\sqrt{n}} \left\vert \bar{X}_{1:k} - \bar{X}_{1:n} \right\vert,
\ee
where $\bar{X}_{i:j}$ denotes the sample mean computed from $X_i, \ldots, X_j$ for any two integers $i,j$ such that $1 \le i < j \le n$. The test statistic may equally be written as
\be
	 \label{eq:cusum2}
	T_n = \max_{1 \le k \le n} \frac{k(n-k)}{n^{3/2}} \left\vert \bar{X}_{1:k} - \bar{X}_{(k+1):n} \right\vert.
\ee
These two representations (\ref{eq:cusum1}) and (\ref{eq:cusum2}) suggest two different views of the CUSUM test statistic: one may either view it, for each $k$, as a comparison of the sample mean of the first part of the sample $X_1, \ldots, X_k$ to the sample mean of the whole sample or as a comparison of the mean of the first part to the mean of the remaining part $X_{k+1}, \ldots, X_n$. 
When the CUSUM approach is extended to a general setting, i.e., when we want to test the constancy of some parameter $\theta \in \R$ of the marginal distribution of the observed process, for which an estimator $\hat\theta_n$ is available, we may consider either of the test statistics 
	$T_1(\hat\theta_n) = \max_{1 \le k \le n} k n^{-1/2} \vert \hat\theta_{1:k} - \hat\theta_{1:n} \vert$
and
	 \label{eq:cusum:gen2}
	$T_2(\hat\theta_n) = \max_{1 \le k \le n} k(n-k)n^{-3/2} \vert \hat\theta_{1:k} - \hat\theta_{(k+1):n}\vert$
with $\hat\theta_{i:j}$ being defined analogously to $\bar{X}_{i:j}$. 
If $\hat\theta_n$ is a \emph{linear} statistic, i.e., if
\be \label{eq:linear}
	\hat\theta_n = \frac{1}{n} \sum_{i=1}^n \xi (X_i)
\ee
for some function $\xi$, both test statistics are identical, but in general they are not. 
Both routes to generalizing the CUSUM test statistic, \emph{first-vs-full} and \emph{first-vs-last}, have been taken by different authors. 
In the present paper, we analyze the difference between $T_1(\hat\theta_n)$ and $T_2(\hat\theta_n)$ in the case $\hat\theta_n$ is a U-statistic. 

Let $(X_i)_{i\in\Z}$ be a $p$-dimensional, stationary stochastic process. A one-sample U-statistic or order 2 is defined as
\be \label{eq:u-statistic}
  U_n = \frac{1}{\binom{n}{2}} \sum_{1\leq i<j\leq n} h(X_i,X_j), 
\ee
where $h: \R^p \times \R^p \to \R$ is a kernel function satisfying $h(x,y) = h(y,x)$. 
In the case of i.i.d.\ data, $U_n$ is an unbiased estimator of the parameter $\theta=E(h(X,Y))$, where $X,Y$ are two independent random variables with the same marginal distribution as $X_1$.  If the underlying data are weakly dependent, $U_n$ is still a consistent estimator of $\theta$ as long as the conditions of the U-statistics ergodic theorem are satisfied, see, e.g., \citet{aaronson:burton:dehling:gilat:hill:weiss:1996}. 
Thus $U_n$ can be used to test for stationarity of a process against the alternative of a change in the parameter $\theta$. Thus change-point tests based on U-statistics have been studied by many others, e.g., \citet{sen:1983}, \citet{hawkins:1989}, \citet{gombay:horvath:1995}, \citet{gombay:2001}, \citet{gombay:horvath:2002}, \citet{chen:qin:2010}, and \citet{matteson:james:2014}. Tests for weakly dependent series have been considered by \citet{buecher:kojadinovic:2016} and \citet{dehling:vogel:wendler:wied:2017}.
Recently, \citet{liu:zhou:zhang:liu:2020}, \citet{wang:zhu:volgushev:shao:2022}, \citet{boniece:horvath:jacobs:2024}, and \citet{zhao:zhu:tan:2024} consider the U-statistic-based change-point tests in the high-dimensional settings. Other lines of research are devoted to sequential (or \emph{online}) change-point tests \citep[e.g.][]{gombay:2000, kirch:stoehr:2022} and the use of anti-symmetric (or \emph{two-sample}) U-statistics \citep[e.g.][]{yu:chen:2022, dehling:vuk:wendler:2022, wegner:wendler:2024}. However, the latter is not pursued here. We restrict our attention to symmetric kernels in the offline (or \emph{retrospective}) setting. 

%
We define the \FF\ and the \FL\ U-statistic change-point tests based on the kernel $h$ as 
\be \label{eq:cusum:ustat1}
	T_1(h) = \max_{1 \le k \le n} \frac{k}{\sqrt{n}} \left\vert U_{1:k} - U_{1:n} \right\vert, 
	\qquad
	T_2(h) = \max_{1 \le k \le n} \frac{k(n-k)}{n^{3/2}} \left\vert U_{1:k} - U_{(k+1):n} \right\vert,
\ee
respectively, where 
\[
	U_{k:l} = \frac{1}{\binom{l-k+1}{2}} \sum_{k\leq i<j\leq l} h(X_i,X_j),
\]
for any $1 \le k < l \le n$. The former test statistic has been considered, e.g., by \citet{dehling:vogel:wendler:wied:2017} and the latter by \citet{buecher:kojadinovic:2016}.

A prominent example for a U-statistic is Gini's mean difference with the kernel $h(x,y) = |x-y|$. The corresponding \FF\ CUSUM test has been considered by \citet{gerstenberger:vogel:wendler:2020} and was found to be a quite competitive change-point test for scale. Another popular U-statistic is Kendall's tau. Its use in the change-point context has been studied in the \FL\ version by \citet{quessy:said:favre:2013} and in the \FF\ version by \citet{dehling:vogel:wendler:wied:2017}. 
In fact, a discussion emerging during the review process of the latter paper has, to a large degree, prompted the research of the present paper. It has been argued that the \FL\ version is presumably uniformly more powerful against one-sudden-change alternatives. However, this is not the case. 

To illustrate the situation and to motivate all following derivations, we shall quote simulated powers for Gini's mean difference (GMD) in a simple example setting. Consider a series $X_1, \ldots, X_n$ of univariate, independent and centered normal variables. At the beginning of the sequence, the observations have standard deviation 1 and at time $[n/3]$, the standard deviation changes to $\sigma > 0$. We apply the studentized versions of the \FF\ and the \FL\ GMD test at the 5\% significance level using the asymptotic null distribution and the variance estimator
\[
	\hat\sigma_{\rm GMD}^2 = \frac{2}{n} \sum_{i=1}^n \Big( \frac{1}{n}\sum_{j=1}^n |X_i- X_j| - g_n \Big)^2, 
\]
where $g_n$ is the sample GMD of $X_1, \ldots, X_n$. 
We observe the following:
\begin{enumerate}[(S1)]
\item \label{list:I} For sample size $n=4000$ and $\sigma = 1.08$, the \FF\ test has a power of 0.79 and the \FL\ test has a power of 0.79. Their relative difference in power is less than 0.1\%. Under the null hypotheses, i.e., $\sigma = 1$, both test have a rejection frequency of 0.05.
\item \label{list:II}
For sample size $n=60$ and $\sigma = 2$ the \FF\ test has a power of 0.70 and the \FL\ test has a power of 0.61. If the variance changes from 1 to $0.5$, the two tests have a power of 0.65 and 0.71, respectively. In both cases, the relative difference is about 10\%.
Under the null hypothesis, the size of both tests is between 0.03 and 0.04.
\end{enumerate}

The rejection frequencies stated above are each based on 10,000 runs and rounded to two decimals. These numbers may serve as illustration for general principles which will be shown in this paper: 
\begin{enumerate}[(A)]
\item Both test statistics, $T_1(h)$ and $T_2(h)$, have the same asymptotic distribution under the null hypothesis.
\item \label{list:B} Both test statistics have, for any non-degenerate U-statistic kernel $h$, the same asymptotic distribution under local alternatives. 
\item \label{list:C} For sequences of fixed alternatives, the test statistics generally behave differently. It does depend on the specific alternative which test has a better power. Particularly, the ranking of the tests is determined by only three numbers: 
\be \label{eq:thetas}
	\theta_F = E(h(X,X')), \qquad
	\theta_G = E(h(Y,Y')), \quad \mbox{ and } \quad
	\theta_{FG} = E(h(X,Y')), 
\ee
where $X, X', Y, Y'$ are independent with $X, X' \sim F$, $Y, Y' \sim G$ and $F$ and $G$ being the data distributions before and after the change, respectively.  
\item \label{list:D} Specifically for Gini's mean difference, the \FF\ test is always more powerful 
if $\theta_F < \theta_G$, and it is less powerful if $\theta_F > \theta_G$.
\item \label{list:E} Specifically for the sample variance, both tests are equally powerful unless there is a simultaneous change in mean. If there is a change in mean, the same relation as for Gini's mean difference holds.
%
\end{enumerate}

Particularly, (\ref{list:B}) implies that in large samples the tests have a similar power, which is illustrated by (S\ref{list:I}) above. 
Likewise, (\ref{list:C}) suggests that for smaller samples there may be substantial differences, as is illustrated by (S\ref{list:II}).

The organization of the paper follows the outline above: 
In Section \ref{sec:null}, we study both test statistics under the null hypothesis, in Section \ref{sec:local} under sequences of local alternatives. 
Section \ref{sec:fixed} is devoted to fixed alternatives and the investigation of the differences of the tests. 
Within the fixed-alternative scenario, we further derive results about the consistency of the change-point tests (Section 
\ref{sec:consistency}), the change-point location estimation (Section \ref{sec:location}), and the asymptotic normality of the test statistics (Section \ref{sec:asym.dist}).
In Section \ref{sec:examples}, we examine several popular U-statistics: Gini's mean difference, the sample variance, the sample covariance, and Kendall's tau. We show in particular (\ref{list:D}) and (\ref{list:E}) and provide further simulation results.
Section \ref{sec:data} contains data examples, Section \ref{sec:conclusion} concludes. All proofs are deferred to the appendix in supplementary material.

%
%
%
%
%
%
%
%
%
%
\section{Asymptotic results}
\label{sec:asym}

We study the asymptotic behavior of both test statistics under the null hypothesis, under sequences of local alternatives, under sequences of fixed alternatives, and the asymptotic behavior of the corresponding maximizing points, which serve as estimators for the location of the change under one-change-point alternatives.
Throughout, $(X_i)_{i\in\Z}$ is a $p$-dimensional, strongly stationary sequence with marginal distribution $F$ and 
$h: \R^p \times \R^p \to \R$ a symmetric kernel function. We restrict our attention to real-valued kernels, but, as in the case of Kendall's tau, the observations may be multivariate. 
Define, for $n \in \N$, the short-hand notation 
\be \label{eq:process.FvsF}
	D^F_n(k) =  k 
	\left(  U_{1:k} - U_{1:n} \right), 
	\qquad 2 \le k \le n, 
\ee
and
\be \label{eq:process.FvsL}
	D^L_n(k) =  \frac{k(n-k)}{n} \left(  U_{1:k} - U_{(k+1):n} \right),
		\qquad 2 \le k \le n, 
\ee
for the weighted \emph{D}ifferences processes of the first-vs-\emph{F}ull and first-vs-\emph{L}ast approach, respectively, so that the corresponding test statistics (\ref{eq:cusum:ustat1}) can be written as
	$T_1(h) = n^{-1/2} \max_{1\le k \le n} \vert D^F_n(k) \vert$ and
	$T_2(h) = n^{-1/2} \max_{1\le k \le n} \vert D^L_n(k) \vert$.

%
%
%
%
%
%
%
%
%
%

\subsection{Null hypothesis}
\label{sec:null}

The main tool for treating U-statistics asymptotically is the Hoeffding decomposition. For the stationary sequence $(X_i)_{i \in \Z}$, we write $U_{1:n}$ as  
\begin{equation*}
	U_{1:n} = 
	\theta + 
	\frac{2}{n}\sum_{i=1}^{n}h_{1}\left(X_{i}\right)
	+\frac{2}{n\left(n-1\right)}\sum_{1\leq i<j\leq n}h_{2}\left(X_{i},X_{j}\right),
\end{equation*}
where
\be \label{eq:hoeffding}
\theta  = E h(X,Y), \quad
h_1(x) = E h(x,Y)-\theta, 
\ee
\[
h_2(x,y) =h(x,y) - h_1(x) -h_1(y) - \theta,
\] 
and $X$, $Y$ are independent copies of $X_0$.
We assume the following condition to hold.
\begin{assumption} \label{ass:nh} \mbox{ \\ } 
\begin{enumerate}[(1)]
	\item \label{assu:weakcon}
	$\displaystyle
		\left(\frac{1}{\sqrt{n}}\sum_{i=1}^{[nt]}h_1(X_i)\right)_{t\in[0,1]} \to \sigma _h W$ \ \ 
   in distribution in the Skorokhod space $D[0,1]$   \\[0.5ex]
	
	as $n \rightarrow \infty$, where $W$ is a standard Brownian motion and 
	\be \label{eq:sigma2}
		\sigma^2_h = 4 \sum_{k=-\infty}^\infty \Cov(h_1(X_0),h_1(X_k)). 
	\ee
	\item \label{assu:h2}
	$\displaystyle
		E\bigg[\bigg(\max_{l\leq n}\Big|\sum_{1\leq i<j\leq l}h_{2}(X_{i},X_{j})\Big|\bigg)^2\bigg], \quad
		E\bigg[\bigg(\max_{l\leq n}\Big|\sum_{l\leq i<j\leq n}h_{2}(X_{i},X_{j})\Big|\bigg)^2\bigg]$ \\[0.5ex]
	are both $o(n^{3})$ as $n \to \infty$.
\end{enumerate}
\end{assumption}

Then we have the following result.

\begin{theorem} \label{th:nh}
If Assumption \ref{ass:nh} is satisfied, then 
\begin{enumerate}[(1)]
\item \label{theo:hypo1}
$\left(n^{-1/2} D^F_n([nt]) \right)_{0 \le t \le 1} \to \, \sigma_h B$ 
in distribution in the Skorokhod space $D[0,1]$ as $n \to \infty$, where $B$ is a Brownian bridge on $[0,1]$ and $\sigma_h$ defined in (\ref{eq:sigma2}).
\item \label{theo:hypo2}
$ \frac{1}{\sqrt{n}} \sup_{0\le t \le 1} \left| D_n^F([nt]) - D_n^L([nt]) \right| \to 0$ in probability as $n \to \infty$. 
\end{enumerate}
\end{theorem}

We have the following remarks. 
\begin{remark} \mbox{ \\ } \rm \label{rem:null}
\begin{enumerate}[(1)]
\item Assumption \ref{ass:nh} is a \emph{high-level} condition, which is \emph{usually} fulfilled, but where it is not necessarily obvious how it translates to explicit conditions on 
$(X_i)_{i \in \Z}$ and $h$. One such set of conditions is that $(X_i)_{i \in \Z}$ are independent and identically distributed and 
$E h^2(X_1,X_2) < \infty$. For dependent sequences, more technical conditions are necessary. One weak set of conditions, allowing a broad class of short-range dependent processes as detailed in Section \ref{app:pned} of the Appendix, is given by Lemma \ref{lem:dependent_data}. 

\item \label{2}
The theorem implies that both $\sigma_h^{-1} T_1(h)$ and $\sigma_h^{-1} T_2(h)$ converge to $\sup_{0\le \lambda\le 1} | B(\lambda)|$ for a Brownian bridge $B$, i.e., the corresponding studentized tests have the same asymptotic behavior under the null and hence the same critical values. The estimation of $\sigma^2_h$ for short-range dependent sequences has been treated by \citet{dehling:vogel:wendler:wied:2017}. Alternatively, a bootstrap scenario for U-statistic processes has been considered by \citet{buecher:kojadinovic:2016}.

\end{enumerate}
\end{remark}

%
%
%
%
%
%
%
%
%
%

\subsection{Local alternatives}
\label{sec:local}

We study local alternatives in the following set-up: As before, $(X_i)_{i \in \Z}$ is a $p$-dimensional, strongly stationary process with marginal distribution $F$. Further consider an array of $p$-dimensional random variables $(Y_{i,n})_{i\in\Z, n\in \N}$, where each row $(Y_{i,n})_{i\in\Z}$ for any fixed $n$ is a stationary process with marginal distribution $G_n$. 
We then construct an array $(X_{i,n})_{i\in\Z, n\in \N}$ as follows:
There is a fixed number $\tau^\star \in (0,1)$ such that
\[
   X_{i,n} = 
	\begin{cases}
			  X_i      &  \ \text{ for } \, i \le [n\tau^\star], \\
				Y_{i,n}  &  \ \text{ for } \, i \ge [n\tau^\star]+1.
	\end{cases}
\]

The observed data in the local-alternative setting is then the triangular array $(X_{i,n})_{1\le i \le n, n \in \N}$ obtained as a cut-out from $(X_{i,n})_{i\in\Z, n\in \N}$. This means, e.g., the processes $D_n^F$ and $D_n^L$, defined in (\ref{eq:process.FvsF}) and (\ref{eq:process.FvsL}), respectively, are henceforth, for each $n \in \N$, computed on $X_{1,n},\ldots,X_{n,n}$ instead of $X_1,\ldots, X_n$.
We adopt the convention to omit the commas separating several indices if there is no ambiguity. The reason for this slight misuse of notation will become apparent in Assumption \ref{ass:local2} below. For instance, $X_{in}$ is short for $X_{i,n}$ and $h_{2ijn}$ is short for $h_{2,i,j,n}$.
Let furthermore
\[
	\theta_F = E(h(X,X')), \quad
	\theta_{G_n} = E(h(Y,Y')), \quad 
	\theta_{FG_n} = E(h(X,Y')), 
\]
and
\[
	\rho_{FG_n} = \theta_{FG_n}-(\theta_F+\theta_{G_n})/2,
\]
where $X, X', Y, Y'$ are independent with $X, X' \sim F$, $Y, Y' \sim G_n$.
The idea of the \emph{local} alternative is that $G_n$ approaches $F$ is some suitable sense as $n \to \infty$.

\begin{assumption} \label{ass:local1}
Let $F$ and $(G_n)_{n\in\N}$ be such that, for $n \to \infty$,
\begin{enumerate}[(1)] 
\item $\sqrt{n} (\theta_F - \theta_{G_n}) \to \Delta$ for some $\Delta \in \R$ and
\item $\sqrt{n} \rho_{FG_n} \to 0$.
\end{enumerate}
\end{assumption}

We further require technical conditions on the kernel $h$. Let $(\tilde{X}_{in})_{1\leq i\leq n, n\in\N}$ be an independent copy of the triangular scheme 
$(X_{in})_{1\leq i\leq n, n\in\N}$ and define further
\begin{align} \label{theta_ijn}
		\theta_{ijn} & =E\left[h(X_{in},\tilde{X}_{jn})\right],\nonumber \\
		h_{1ijn}(x)  & =E\left[h(x,X_{jn})\right]-\theta_{ijn}, \\
		h_{2ijn}(x,y)& =h(x,y)-E\left[h(X_{in},y)\right]-E\left[h(x,X_{jn})\right]+\theta_{ijn}. \nonumber
\end{align}
These terms constitute the Hoeffding decomposition when the underlying data generating process is a triangular array. 

\begin{assumption} \label{ass:local2}
Let $(X_{in})_{1\leq i\leq n, n\in\N}$ and $h$ be such that
\begin{enumerate}[(1)] 
\item 
$\ds E\bigg[\bigg(  \max_{2\leq l \leq n}\Big|\sum_{1\leq i<j\leq l}h_{2ijn}(X_{in},X_{jn})\Big|   \bigg)^2\bigg], \quad
     E\bigg[\bigg(  \max_{2\leq l \leq n}\Big|\sum_{l\leq i<j\leq n}h_{2ijn}(X_{in},X_{jn})\Big|  \bigg)^2\bigg]$ \\
are $o(n^3)$ for $n \to \infty$, and 

\vspace{0.5ex}
\item the four terms
\[
 \frac{1}{n} E\bigg( \max_{k\leq n}\Big|\sum_{i=1}^k \{ h_{1i1n}(X_{in})-h_{1}(X_i)\} \Big| \bigg)^2, \quad
 \frac{1}{n} E\bigg( \max_{k\leq n}\Big|\sum_{i=1}^k \{ h_{1inn}(X_{in})-h_{1}(X_i)\} \Big| \bigg)^2,
\]
\[
 \frac{1}{n^3} E\bigg( \max_{k\leq n}\Big|\sum_{i=1}^k i \{ h_{1i1n}(X_{in})-h_{1}(X_i) \} \Big| \bigg)^2, \
 \frac{1}{n^3} E\bigg( \max_{k\leq n}\Big|\sum_{i=1}^k i \{ h_{1inn}(X_{in})-h_{1}(X_i) \} \Big| \bigg)^2
\]
all converge to zero as $n \to \infty$.
\end{enumerate}
\end{assumption}

We then have the following process convergence result in the local-alternative set-up.

\begin{theorem}\label{th:local} Let Assumptions \ref{ass:nh} (1), \ref{ass:local1}, and \ref{ass:local2} hold. Then
\begin{enumerate}[(1)]
\item 
$\ds \left( \frac{1}{\sqrt{n}} D^F_n([n t]) - \phi(t) \right)_{0 \le t \le 1}  \to \, \sigma_h B$ 
in distribution in the Skorokhod space $D[0,1]$ as $n \to \infty$, where $B$ is a Brownian bridge on $[0,1]$, $\sigma_h$ 
as in Theorem \ref{th:nh}, and
\[
		\phi(t) = 
		\begin{cases}t(1-\tau^\star)\Delta &\text{ for } t<\tau^\star\\
		(1-t)\tau^\star\Delta &\text{ for } t\geq\tau^\star\end{cases}
\]
\item
$\ds \frac{1}{\sqrt{n}} \sup_{0\le t \le 1} \left| D_n^F([nt]) - D_n^L([nt]) \right| \to 0$ in probability as $n \to \infty$. 
\end{enumerate}
\end{theorem}
Assumption \ref{ass:local1} captures the common notion of a series of local alternatives to approach the null hypothesis at a 
$(1/\sqrt{n})$-rate. 
Assumption \ref{ass:local2} is a technical regularity condition. Easier-to-verify conditions for independent and dependent data sequences
are given by Lemmas \ref{lem:local2} and \ref{lem:local3}, respectively, of Section \ref{app:assumptions_local} in the supplementary material. 

The crucial assumption here is \ref{ass:local1} (2). In many standard situations, Assumption \ref{ass:local1} (2) is fulfilled if 
\ref{ass:local1} (1) holds. In Proposition \ref{prop:densities}, we treat the case that $F$ and $G_n$ have densities $f$ and $g_n$, respectively, belonging to a one-parameter family $p_\lambda(x)$, $\lambda\in \Lambda$, where $\Lambda \subset \R$ is an interval. 
We assume that the function $\lambda\mapsto p_\lambda(x)$ is differentiable for almost all $x$, and we denote its partial derivative by $\frac{\partial}{\partial \lambda} p_\lambda (x)$. Let $\lambda_0$ 
be an interior point of $\Lambda$, let $a\in \R$, and consider the densities
\[
 f(x) =p_{\lambda_0}(x), \qquad 
g_n(x) = p_{\lambda_0+a/\sqrt{n}} (x).
\]
Then, under some technical conditions, Assumption~\ref{ass:local1} is satisfied.

\begin{proposition} \label{prop:densities}
Suppose there is a measurable function $v(x)$ such that for some $\epsilon >0$
\begin{equation}\label{eq:dom-conv-1}
   \Big| p_\lambda (x) \Big| +     \Big| \frac{\partial}{\partial \lambda} p_\lambda (x) \Big| \leq v(x)
	 \ \mbox{ for all } \ \lambda \in (\lambda_0-\epsilon, \lambda_0+\epsilon)
\end{equation}
and such that
\begin{equation}
\label{eq:dom-conv-2}
     \iint |h(x,y)| v(x)\, v(y) \, dx\, dy  <\infty 
\end{equation}
Then Assumption~\ref{ass:local1} is satisfied for the distributions $F$ and $G_n$ with densities $f(x)$ and $g_n(x)$ as defined above, and with
\[
  \Delta=-2\, a\, \iint h(x,y)\,  \frac{\partial}{\partial \lambda} p_{\lambda_0}(x) \, p_{\lambda_0}(x) \, dx. 
\]
\end{proposition}

This covers in particular the Gaussian scale family, which we use extensively in the simulations in Section \ref{sec:examples}, see Example \ref{ex:scale_family} in the Appendix.



Theorems \ref{th:nh} and \ref{th:local} together state the two tests, the \FF\ and the \FL, are \emph{asymptotically equivalent}, i.e., for any array $(X_{in})$ of random variables falling either under the null hypothesis or a local alternative, the probability of a differing test decision tends to zero as the sample size tends to infinity.

%
%
%
%
%
%
%
%
%
%

\subsection{Fixed alternatives}
\label{sec:fixed}

So far we have noted that the two U-statistic change-point test statistics are asymptotically equivalent. Nevertheless, noteworthy differences in power are observed in simulations. The question remains if these differences can be backed by theoretical results and if conditions can be given to identify situations where one test is better than the other. 
Toward this end, we consider sequences of fixed alternatives. 
Let $(X_{i})_{i\in\Z}$ and $(Y_{i})_{i\in\Z}$ be two $p$-dimensional, strongly stationary sequences with marginal distributions $F$ and $G$, respectively.
We further require the $2p$-dimensional process $(X_i, Y_i)_{i\in\Z}$ to be strongly stationary.
Then the data array $(X_{in})_{i\in\Z, n\in \N}$ is constructed analogously to the local-alternative set-up:
\[
   X_{in} = 
	\begin{cases}
			  X_i      &  \ \text{ for } \, i \le [n\tau^\star], \\
				Y_{i}  &  \ \text{ for } \, i \ge [n\tau^\star]+1
	\end{cases}
\]
for some fixed $\tau^\star \in (0,1)$. The observed data is, as before, the triangular cut-out $(X_{in})_{1\leq i\leq n, n\in\N}$.
%
%
%
%
%
%
Let $(\tilde{X}_{in})_{1\leq i\leq n, n\in\N}$ be an independent copy of $(X_{in})_{i\in\Z, n\in \N}$. 
Recall the definitions of $\theta_F$, $\theta_G$, and $\theta_{FG}$, cf.~(\ref{eq:thetas}), and let
\be \label{eq:offset}
	\rho_{FG} = \theta_{FG} - (\theta_F+\theta_G)/2.
\ee
We call $\rho_{FG}$ the \emph{eccentricity} of $\theta_{FG}$ as it describes the deviation of the mixed parameter 
$\theta_{FG}$ from the midpoint between $\theta_F$ and $\theta_G$. This offset will play a crucial role in the following: it determines the power ranking of the test.

Below we employ again the notation $\theta_{ijn}$, $h_{1ijn}(x)$, and $h_{2ijn}(x,y)$, as defined in (\ref{theta_ijn}), but keep in mind that the arrays $(X_{in})$ and $(\tilde{X}_{in})$ are now defined slightly different than in Section \ref{sec:local}.

\begin{assumption} \label{ass:fixed}
Let $(X_{in})_{1\leq i\leq n, n\in\N}$ and $h$ be such that
\begin{enumerate}[(1)] 
\item the terms 
\label{enum:fixed1}
\[
E\bigg[\bigg(\max_{l\leq n}\Big|\sum_{1\leq i<j\leq l}h_{2ijn}(X_{in},X_{jn})\Big|\bigg)^2\bigg], \quad
E\bigg[\bigg(\max_{l\leq n}\Big|\sum_{l\leq i<j\leq n}h_{2ijn}(X_{in},X_{jn})\Big|\bigg)^2\bigg]
\]
are $o(n^4)$ for $n \to \infty$, and 	

\vspace{0.5ex}
\item the eight terms
\label{enum:fixed2}
\be
  \label{eq:fixed3}
	E\bigg[\bigg(\max_{l\leq n}\Big|\sum_{i=1}^lh_{1i1n}(X_{in})\Big|\bigg)^2\bigg], \quad
	E\bigg[\bigg(\max_{l\leq n}\Big|\frac{1}{l}\sum_{i=1}^lih_{1i1n}(X_{in})\Big|\bigg)^2\bigg],
\ee\be 
  \label{eq:fixed4}
	E\bigg[\bigg(\max_{l\leq n}\Big|\sum_{i=1}^lh_{1inn}(X_{in})\Big|\bigg)^2\bigg], \quad 
	E\bigg[\bigg(\max_{l\leq n}\Big|\frac{1}{l}\sum_{i=1}^lih_{1inn}(X_{in})\Big|\bigg)^2\bigg],
\ee
\be \label{eq:fixed5} 
	E\bigg[\bigg(\max_{l\leq n}\Big|\sum_{i=l+1}^nh_{1i1n}(X_{in})\Big|\bigg)^2\bigg], \quad
	E\bigg[\bigg(\max_{l\leq n}\Big|\frac{1}{n-l}\sum_{i=l+1}^n(n-i+1)h_{1i1n}(X_{in})\Big|\bigg)^2\bigg], 
\ee
\be \label{eq:fixed6}
	E\bigg[\bigg(\max_{l\leq n}\Big|\sum_{i=l+1}^nh_{1inn}(X_{in})\Big|\bigg)^2\bigg], \quad 
	E\bigg[\bigg(\max_{l\leq n}\Big|\frac{1}{n-l}\sum_{i=l+1}^n(n-i+1)h_{1inn}(X_{in})\Big|\bigg)^2\bigg]
\ee
are all $o(n^2)$ for $n \to \infty$.
\end{enumerate}
\end{assumption}

Similarly to Assumption \ref{ass:local2} in the previous section, Assumption \ref{ass:fixed} is a general technical condition. For independent data, it greatly simplifies, see Lemma \ref{lem:fixed:independent}. Analogous conditions for dependent data are given in Lemma 
\ref{lem:fixed:pned}.
Then we have the following result concerning the behavior of the change-point tests under fixed alternatives.
\begin{theorem}\label{th:fixed} 
Under the Assumption \ref{ass:fixed} we have
\begin{enumerate}[(1)] 
\item $\displaystyle
\sup_{t\in[0,1]}\left|\frac{1}{n}D_{n}^F([nt])-\Psi_1(t)\right|\to 0$ 
in probability as $n \to \infty$, where \\[0.5ex]
\begin{equation*}
\Psi_1(t) = 
\begin{cases}
	t(1-\tau^\star)(\theta_F-\theta_G)-2t\tau^\star(1-\tau^\star)\rho_{FG} & \text{ for } t<\tau^\star, \\
	(1-t)\tau^\star(\theta_F-\theta_G)+2\frac{t-\tau^\star}{t}\tau^\star\rho_{FG}-2t\tau^\star(1-\tau^\star)\rho_{FG} & \text{ for } t\geq\tau^\star.
\end{cases}
\end{equation*}
	
\item $\displaystyle
\sup_{t\in[0,1]}\left|\frac{1}{n}D_n^L([nt])-\Psi_2(t)\right| \to 0$
in probability as $n \to \infty$, where \\[0.5ex]
\begin{equation*}
\Psi_2(t) =
\begin{cases}
t(1-\tau^\star)(\theta_F-\theta_G)-2\frac{t(\tau^\star-t)}{1-t}(1-\tau^\star)\rho_{FG} &\text{ for } t<\tau^\star,\\
(1-t)\tau^\star(\theta_F-\theta_G)+2\frac{(t-\tau^\star)}{t}\tau^\star\rho_{FG}-2(t-\tau^\star)\tau^\star\rho_{FG} &\text{ for } t\geq\tau^\star. 
\end{cases}
\end{equation*}
\end{enumerate}
\end{theorem}
As the limit distribution under the null-hypothesis is the same, the different limits under fixed alternatives give information about the differences in power for finite samples. A short calculation yields
\begin{equation*}
\Psi_2(t)-\Psi_1(t) =
\begin{cases}
  2 \displaystyle \frac{t^2 (1-\tau^\star)^2}{1-t} \rho_{FG} \ \  &\text{ for } t<\tau^\star, \\[1.5ex] 
  2 (1-t)(\tau^\star)^2 	\rho_{FG} &\text{ for } t\geq\tau^\star. \end{cases}
\end{equation*}
As $\Psi_2(t)-\Psi_1(t)$ always equals $\rho_{FG}$ times a positive factor, the sign of the eccentricity $\rho_{FG}$ is essential for the power comparison. We arrive at the following criterion:
\begin{enumerate}[(1)]
\item
If $\rho_{FG}$ has the same sign as $\theta_G - \theta_F$, the absolute value of $\Psi_1$ is larger than the absolute value of  $\Psi_2$,  
and we expect the \FF\ test to be more powerful. 
\item
If $\rho_{FG}$ and $\theta_G - \theta_F$ have opposite signs, the absolute value of $\Psi_1$ is \emph{smaller} than the absolute value of $\Psi_2$, and we expect the \FL\ test to be more powerful.
\item
If $\rho_{FG} = 0$, the two functions $\Psi_1$ and $\Psi_2$ are equal. We expect both tests to be equally powerful.
\end{enumerate}

Thus the sign of $\rho_{FG}$ relative to $\theta_G - \theta_F$ determines the power ranking of the tests. Consequently, there is no kernel $h$ for which one of the tests is strictly better than the other for all $F$ and $G$, since $\rho_{FG}$ is symmetric in $F$ and $G$ and the ranking of the tests can hence be reversed by switching $F$ and $G$. Moreover, both tests are equally powerful for all $F$ and $G$ if and only if $U_n$ is a strictly linear statistic. This is the content of the next lemma.

\begin{lemma} \label{lemma:linear}
Let $h: \R^p \times \R^p \to \R$ be a symmetric, measurable kernel function.
If $\rho_{FG} = 0$ for all $p$-dimensional distribution functions $F$ and $G$ for which $\theta_F$, $\theta_G$ and $\theta_{FG}$ exist, then 
 the U-statistic $U_n$ defined by (\ref{eq:u-statistic}) is a linear statistic in the sense of (\ref{eq:linear}) with 
$\xi(x) = h(x,x)$ for $x \in \R^p$ and hence $D^F_n = D^L_n$.
\end{lemma}

%
%
%
%
%
%
%
%
%
%

\subsection{Consistency}
\label{sec:consistency}

Based on Theorem \ref{th:fixed} we show that the change-point tests based on $D_n^F$ and $D_n^L$ are consistent. Under the null hypothesis, we have that
\[
	T_1(h) = \frac{1}{\sqrt{n}} \max_{1\le k \le n} \vert D^F_n(k) \vert = O_P(1), \qquad 
	T_2(h) = \frac{1}{\sqrt{n}} \max_{1\le k \le n} \vert D^L_n(k) \vert = O_P(1)
\] 
by Theorem \ref{th:nh}. 
If, under a fixed alternative, $\Psi_1(t)$ and $\Psi_2(t)$ are not 0 for a least one $t \in[0,1]$, the corresponding test statistics will converge to infinity, so that the tests are consistent. This is summarized in the following corollary.
\begin{corollary} \label{cor:consistency}
If the assumptions of Theorem \ref{th:fixed} and one of the two conditions
\begin{enumerate}[(1)]
\item
$\theta_F \neq \theta_G$,
\item
$\rho_{FG} \neq 0$, 
\end{enumerate}
hold, then
\[
	T_1(h) = \frac{1}{\sqrt{n}} \max_{1\le k \le n} \vert D^F_n(k) \vert \to \infty, \qquad
	T_2(h) = \frac{1}{\sqrt{n}} \max_{1\le k \le n} \vert D^L_n(k) \vert\to \infty
\]
in probability as $n \to \infty$.
\end{corollary}
Corollary \ref{cor:consistency} is an affirmative result of a common property of the tests: Both are constructed for detecting changes in the parameter $\theta$, and, if this parameter differs before and after the change-point, both will detect that change asymptotically with power 1. However, Corollary \ref{cor:consistency} contains another important information: if a change in the series occurs, but the change is such that the parameter $\theta$ remains constant, both tests will detect this change asymptotically with power 1 \emph{if and only if} the eccentricity $\rho_{FG}$ is non-zero. Some implications of this result are illustrated in Section \ref{sec:var}.

%
%
%
%
%
%
%
%
%
%

\subsection{Estimation of change location}
\label{sec:location}

If the hypothesis is rejected, it is natural to estimate the change-point by
\be \label{eq:location_estimator}
\hat{\tau}_{1,n} =\operatorname*{argmax}_{t\in[0,1]}\left|D_n^F([nt])\right|
\quad \mbox{ or } \quad
\hat{\tau}_{2,n} =\operatorname*{argmax}_{t\in[0,1]}\left|[D_n^L([nt])\right|.
\ee
If $\Psi_1$ respectively $\Psi_2$ have their maximal absolute value at $\tau^\star$ and this maximal point is unique, 
consistency of these estimators follows by the argmax theorem \citep[e.g.][Corollary 3.2.3]{vandervaart:wellner:1996}.
The following proposition states technical conditions for $D_n^F$ and $D_n^L$.

\begin{proposition}\label{prop:max} 
\mbox{ \\}
\begin{enumerate}[(1)] 
\item If $\theta_F>\theta_G$ and
\begin{equation*}
\frac{\theta_G-\theta_F}{2}<\rho_{FG}<\frac{\theta_F-\theta_G}{2(\tau^\star-1+1/\tau^\star)},
\end{equation*}
then $|\Psi_1(t)|<|\Psi_1(\tau^\star)|$ for all  $t\in[0,1]$ with $t\neq \tau^\star$.
The same is true if $\theta_F<\theta_G$ and
\begin{equation*}
\frac{\theta_F-\theta_G}{2(\tau^\star-1+1/\tau^\star)}<\rho_{FG}<\frac{\theta_G-\theta_F}{2}.
\end{equation*}
\item If $\theta_F\neq\theta_G$ and  $|\rho_{FG}|\leq \frac{1}{2}\min\{\frac{\tau^\star}{1-\tau^\star},\frac{1-\tau^\star}{\tau^\star}\}|\theta_F-\theta_G|$, then for all  $t\in[0,1]$ with $t\neq \tau^\star$
\begin{equation*}
|\Psi_2(t)|<|\Psi_2(\tau^\star)|.
\end{equation*}
\end{enumerate}
\end{proposition}

Similar to the previous two results, the implications of Proposition \ref{prop:max} are threefold:
First, in a wide range of realistic application cases, both tests have the desired property and behave alike (consistency of the argmax location estimator). Second, there is no general superiority of one test over the other: neither set of conditions in parts (1) and (2) of Proposition \ref{prop:max} implies the other. And lastly, the decisive criterion is again the eccentricity $\rho_{FG}$ (it must not be too large in absolute value), and one can hence easily construct alternatives where the argmax estimator is consistent and alternatives where it is not, see Section \ref{sec:var} below and Example \ref{ex:gmd:location} in the supplementary material.

%
%
%
%
%
%
%
%
%
%

\subsection{Asymptotic distribution under fixed alternatives}
\label{sec:asym.dist}

From Theorem \ref{th:fixed} and Proposition \ref{prop:max} we know that, under fixed alternatives, the $(1/\sqrt{n})$-scaled test statistics
$n^{-1/2} T_1(h)$ and $n^{-1/2} T_2(h)$ converge in probability to different values $\Psi_1(\tau^\star)$ and $\Psi_2(\tau^\star)$, respectively.
To complete the asymptotic analysis of the \FF\ and the \FL\ approach, we study the asymptotic distribution of the two test statistics within the fixed-alternative set-up. We demonstrate that they differ not only in mean, but we also obtain asymptotic variances that are generally different. This becomes apparent for univariate and independent data sequences and we restrict our derivations to this scenario in the following.

%
Using the notation introduced at the beginning of Section \ref{sec:fixed} define further
\[
	h_F(x) =\int h(x,y)\, dF(y), \qquad h_G(x) = \int h(x,y)\,  dG(y).
\] 
%
%
%

\begin{theorem}\label{th:limitfixed} 
Let the triangular array $(X_{in})_{1\leq i\leq n, n\in\N}$ be row-wise independent, and let 
the assumptions of Theorem \ref{th:fixed} and Proposition \ref{prop:max} be satisfied. Let further the kernel $h:\R^2\rightarrow \R$ be such that $\theta_F\neq \theta_G$ and the integrals
\[
	\iint h^2(x,y)dF(x)dF(Y), \quad
	\iint h^2(x,y)dF(x)dG(Y), \quad
	\iint h^2(x,y)dG(x)dG(Y)
\]
are all finite.
Then
\[
	\big(T_1(h)-\sqrt{n}|\Psi_1(\tau^\star)|\big) \rightarrow Z_1, \quad 
	\big(T_2(h)-\sqrt{n}|\Psi_2(\tau^\star)|\big) \rightarrow Z_2
\]
in distribution as $n \to \infty$, where $Z_1$, $Z_2$ are centered normal random variables with variances
\begin{align*}
\Var[Z_1]&=4\tau^\star\Var[h_F(X)-\tau^\star h_G(X)]+4\tau^{\star2}(1-\tau^\star)\Var[h_G(Y)],\\
\Var[Z_2]&=4\tau^\star(1-\tau^\star)^2\Var[h_F(X)]+4\tau^{\star2}(1-\tau^\star)\Var[h_G(Y)]
\end{align*}
for a random  variable $X$ with distribution function $F$ and a random variable $Y$ with distribution function $G$.
\end{theorem}

If we have a strictly linear U-statistic, i.e., $h(x,y)=\frac{1}{2}(h(x,x)+h(y,y))$, both test statistics are equal, and should have the same limit variance. 
This is indeed the case: In this situation, $h_F(x)=\frac{1}{2}h(x,x)-E[\frac{1}{2}h(X,X)]$ and $h_G(x)=\frac{1}{2}h(x,x)-E[\frac{1}{2}h(Y,Y)]$, so $\Var[h_F(X)-\tau^\star h_G(X)]=\Var[(1-\tau^\star)h_F(X)]=(1-\tau^\star)^2\Var[h_F(X)]$ and consequently $\Var[Z_1]=\Var[Z_2]$. 
This argument does not apply if the U-statistic is not linear, and the asymptotic variances are generally different.

%
%
%
%
%
%
%
%
%
%

\section{Examples}
\label{sec:examples}

We examine four well-known U-statistics in detail: Gini's mean difference, the sample variance, the sample covariance, and Kendall's tau. 
The former two can be used to detect changes in the scale of univariate series. So in Sections \ref{sec:gmd} and \ref{sec:var} we assume the data series to be univariate with distributions $F$ and $G$ before and after the potential change-point. In case $F$ and $G$ possess Lebesgue densities $f$ and $g$, respectively, we have the following representation of the eccentricity $\rho_{FG}$, which will be employed in Proposition \ref{prop:gmd}.
\be \label{eq:rhoFG:univariate}
  \rho_{FG} = - \frac{1}{2} \int\int h(x,y) \{f(x)-g(x)\}\{f(y)-g(y)\} dx dy.
\ee

%
%
%
%
%
%
%
%
%
%

\subsection{Gini's mean difference}
\label{sec:gmd}

We consider the kernel $h(x,y) = |x - y|$, leading to Gini's mean difference
\[
 U_{1:n}=\frac{2}{n(n-1)}\sum_{1\leq i < j\leq n}|X_i-X_j|
\] 
with
	$\theta_F = E|X-X'|$, 
	$\theta_G = E|Y-Y'|$,  and $\theta_{FG} = E|X-Y|$, 
for $X, X', Y, Y'$ being independent with $X, X' \sim F$, $Y, Y' \sim G$. For the special case of normal distributions, specifically 
$F = N(0,1)$ and $G = N(0,\sigma^2)$ for some $\sigma > 0$, a quick calculation yields
\[
	\theta_F = \frac{2}{\sqrt{\pi}},\quad
	\theta_G = \frac{2}{\sqrt{\pi}}\sigma, \quad 
	\theta_{FG} = \frac{2}{\sqrt{\pi}}\sqrt{\frac{1}{2}(1+\sigma^2)}
\]
and hence $\rho_{FG} = \theta_{FG} - (\theta_F+\theta_G)/2 \ge 0$. We find the eccentricity is always positive for $\sigma \neq 1$. This holds indeed generally.

\begin{proposition}\label{prop:gmd}
Let the univariate distributions $F$ and $G$ with Lebesgue densities $f$ and $g$, respectively, have finite first moments.
Then, for the kernel $h(x,y) = |x - y|$, we have $\rho_{FG} \ge 0$. 
\end{proposition}

Consequently, the power ranking of both change-point tests is solely determined by the order of $\theta_F$ and $\theta_G$. 
Following our criterion from Section \ref{sec:fixed}, the \FF\ test is more powerful than the \FL\ test if $\theta_G - \theta_F > 0$, i.e., if the scale of the series changes from a smaller value to a larger one. The \FF\ is less powerful than the \FL\ if the scale decreases after the change. This is regardless of where the change occurs and of any other aspect of $F$ and $G$ such as, e.g., a simultaneous change in location.

Also, we find in the normal examples that $\theta_{FG}$ is always between $\theta_F$ and $\theta_G$. This is also generally true for median-centered sequences. 
\begin{proposition} \label{prop:gmd:location}
Let $h(x,y) = |x-y|$ and $X \sim F$ with $E|X|< \infty$ and $F(0) = 1/2$, i.e, $X$ is median-centered. Furthermore, for $c>0$, we denote by $G = G_c$ the distribution function of $cX$. Then $\theta_{FG}\in[\min\{\theta_F,\theta_G\},\max\{\theta_F,\theta_G\}]$
\end{proposition}

\paragraph*{Visualization.} Figure \ref{fig:gini} illustrates Theorem \ref{th:fixed} for Gini's mean difference. The set-up is similar to the simulation example in the Introduction. We consider a series of independent observations stemming from $F = N(0,1)$ before the change at $\tau^\star = 1/3$ and $G = N(0,4)$ afterwards. 
\begin{figure}[t]%
\centering
\includegraphics[width=\columnwidth]{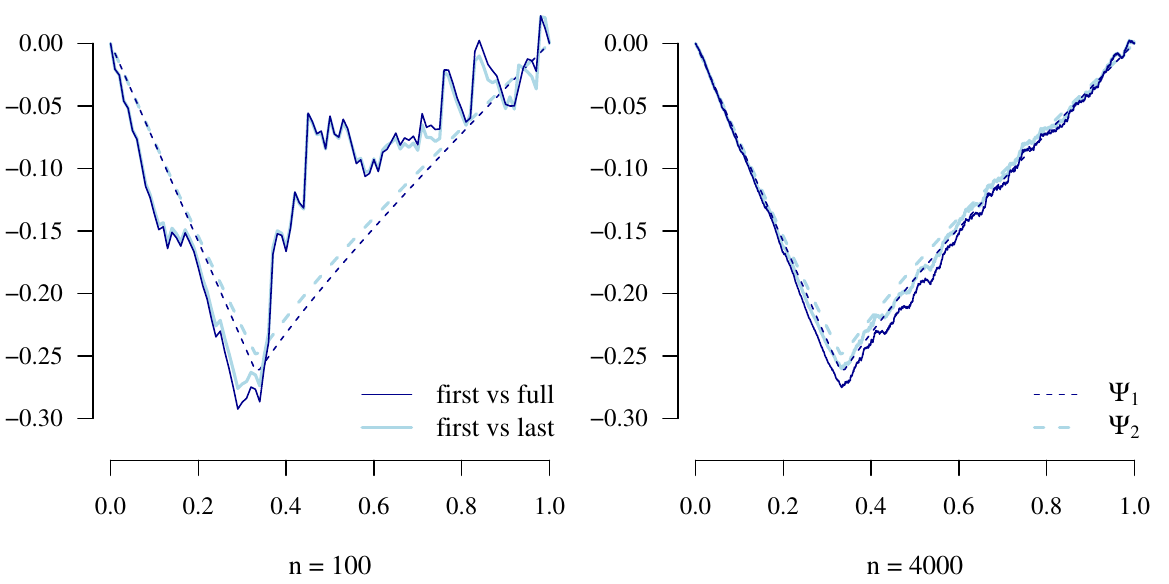}
\caption{Visualization of Theorem \ref{th:fixed} for Gini's mean difference: A change occurs at $\tau^\star = 1/3$ from 
$F = N(0,1)$ to $G = N(0,4)$. 
The limit functions $\Psi_1$ (\FF\ test; dark blue) and $\Psi_2$ (\FL\ test; light blue) as given in Theorem \ref{th:fixed} along with corresponding trajectories of
$\left(\frac{1}{n}D_{n}^F(k)\right)_{1\le k \le n}$ and 
$\left(\frac{1}{n}D_{n}^L(k)\right)_{1\le k \le n}$ for a sequence of $n$ independent observations with $n = 100$ (left panel) and
$n = 4000$ (right panel).
}
\label{fig:gini}%
\end{figure}
The dashed lines in both panels depict the limit functions $\Psi_1$ (\FF\ test; dark blue) and $\Psi_2$ (\FL\ test; light blue). The jagged lines show trajectories of the corresponding processes
$\left(\frac{1}{n}D_{n}^F(k)\right)_{1\le k \le n}$ and 
$\left(\frac{1}{n}D_{n}^L(k)\right)_{1\le k \le n}$ based on a sequence of $n$ independent observations with $n = 100$ (left panel) and
$n = 4000$ (right panel).
Figure \ref{fig:gini} visualizes the converges of the processes 
to their respective deterministic limit as stated by Theorem \ref{th:fixed}. 
The difference between the trajectories can be seen to be essentially captured by the difference between the respective limit processes.
%
%
%
%
%
%
\paragraph*{Simulation results.} We consider independent Gaussian processes only. The sample sizes are $n = 63, 250, 1000, 4000$ with potential change-point locations at $\tau^\star =  1/4, 1/2$ and $3/4$. The height of the change is sample-size dependent.
The three scenarios are as follows:
\begin{enumerate}[(1)]
\item Null hypothesis: No change occurs. The data are standard normal iid sequences.
\item Alternative 1: Before the change-point, the observations have standard deviation 1 and after the change $\sigma_n = 1 + 3/\sqrt{n}$, e.g., $\sigma_{63} \approx 1.378$, $\sigma_{4000} \approx 1.047$. 
\item Alternative 2: Before the change-point, the observations have standard deviation $\sigma_n$ and afterwards 1.
\end{enumerate}

\begin{table}[t]
	\caption{\small Gini's mean difference. Empirical rejection frequencies (\%) at the 5\% significance level of the \FF\ (FvsF) and \FL\ (FvsL) change-point test. Studentized with long-run variance estimation (bandwidth $b_n = n^{1/3}$). 2000 runs. 
		\label{tab:gmd2}
	}
	\renewcommand{\arraystretch}{1.1}
	\begin{center}
		\begin{tabular}{rrr|rrrr}  
		             &                  &                  & \multicolumn{4}{c}{Sample size $n$} \\
        Scenario &  change location &   test statistic &               63 &              250 &             1000 &             4000 \\ 
				\hline 
 Null hypothesis &                  &             FvsF &              2.9 &              3.6 &              4.0 &              3.8 \\ 
                 &                  &             FvsL &              2.1 &              3.3 &              4.0 &              3.9 \\[1.0ex] 
				\hline 
   Alternative 1 &             0.25 &             FvsF &             15.8 &             28.9 &             36.8 &             43.6 \\ 
                 &             0.25 &             FvsL &              9.2 &             25.1 &             35.2 &             42.6 \\[1.0ex] 
                 &              0.5 &             FvsF &             39.5 &             55.8 &             63.2 &             68.8 \\ 
                 &              0.5 &             FvsL &             25.8 &             51.7 &             61.5 &             68.0 \\[1.0ex] 
                 &             0.75 &             FvsF &             26.0 &             36.5 &             41.8 &             43.7 \\ 
                 &             0.75 &             FvsL &             12.7 &             29.0 &             39.2 &             42.1 \\[1.0ex] 
				\hline 
   Alternative 2 &             0.25 &             FvsF &              5.9 &             25.8 &             38.1 &             43.7 \\ 
                 &             0.25 &             FvsL &              9.5 &             28.6 &             39.6 &             44.5 \\[1.0ex] 
                 &              0.5 &             FvsF &             16.8 &             46.9 &             60.8 &             69.8 \\ 
                 &              0.5 &             FvsL &             28.7 &             51.5 &             63.1 &             70.3 \\[1.0ex] 
                 &             0.75 &             FvsF &              4.0 &             19.6 &             33.5 &             41.8 \\ 
                 &             0.75 &             FvsL &             10.8 &             26.2 &             36.2 &             43.7 \\[1.0ex] 
	\hline 
		\end{tabular}
	\end{center}
\end{table}
\color{black}
Table \ref{tab:gmd2} shows empirical rejection frequencies of the \FF\ test and the \FL\ test for Gini's mean difference at the asymptotic 5\% level. Both test statistics are divided by the square root of a long-run variance estimate using the Bartlett kernel and a sample-size-dependent bandwidth $b_n = n^{1/3}$, see Appendix \ref{app:data}. The thus obtained studentized test statistics are compared to the 95\% quantile of the Kolmogorov distribution, see Remark \ref{rem:null} (\ref{2}).

We observe the following:
As expected, a change in the middle of the data ($\tau^\star = 1/2$) is easier to detect than an early or late change. 
Although the height of the change was adjusted for the sample size, we see an increase in power with growing sample size $n$.
Concerning the comparison of the two approaches, we find our expectations based on the theoretical results confirmed:
The \FF\ test has a higher power when the variance is increasing (Alternative 1), while the \FL\ version has higher power for a decreasing variance (Alternative 2). The power ranking is independent of the location of $\tau^\star$. The difference in power is most pronounced for the smallest sample size $n=63$. For $n=4000$, the power difference is very small, which is not surprising, as the two tests are asymptotically equivalent (Theorem \ref{th:local}). 

An aspect not pursued in this paper is why one would use Gini's mean difference as scale estimator in the first place. There are good reasons to do so: The estimator Gini's mean difference is -- for practical purposes -- as efficient as the standard deviation at the normal model \citep{nair:1936}, which we find reflected in the powers of the corresponding tests when comparing Table \ref{tab:gmd2} to Table \ref{tab:var}. Furthermore, Gini's mean difference is more efficient than the standard deviation at heavier-tailed distributions \citep{gerstenberger:vogel:2015}, which translates to better powers of the respective change-point tests \citep{gerstenberger:vogel:wendler:2020}.

\subsection{The sample variance}
\label{sec:var}

The sample variance is obtained by the kernel $h(x,y) = (x - y)^2/2$, leading to the U-statistic 
\[
 U_{1:n} =\frac{1}{n(n-1)}\sum_{1\leq i<j\leq n}(X_i-X_j)^2 = \frac{1}{n-1}\sum_{i = 1}^n(X_i-\bar{X}_{1:n})^2
\] 
Also, in this case, the eccentricity is strictly non-negative. Specifically, we have the following result.
\begin{proposition}\label{prop:var}
Let $F$ and $G$ have finite second moment. Then for $h(x,y) = (x - y)^2/2$ we have
\begin{equation*}
	\rho_{FG} = \frac{1}{2}\big(EX-EY\big)^2.
\end{equation*}
\end{proposition}

Despite being straightforward to show, Proposition \ref{prop:var} is a very remarkable result and has -- in light of the derivation in Section \ref{sec:asym} -- a number of interesting implications. 

The change-in-variance problem has received some considerable attention. It has been studied in varying degrees of generality, e.g., by
\citet{inclan:tiao:1994}, \citet{gombay:horvath:huskova:1996}, \citet{wied:arnold:bissantz:ziggel:2012} and \citet{gerstenberger:vogel:wendler:2020}. A robustified version based on trimming has been proposed by \citet{lee:park:2001}, and the multivariate extension is due to \citet{aue:hoermann:horvath:reimherr:2009}. 
None of these authors consider the sample variance explicitly in the U-statistic context, and none of the authors discuss or draw any connection between the two different construction principles, presumably due to the sample variance being an \emph{essentially linear} statistic. It would be \emph{exactly linear} if the sample mean was replaced by the population mean. 
Indeed, Proposition \ref{prop:var} suggests that, for any alternative, as long as the mean does not change from $F$ to $G$, both tests are equally powerful. 

Much more interesting is what happens if there is a change in mean at the same time as the change in variance.
It is a general consensus that it is advisable to apply changes-for-variance tests (or generally scale) to centered sequences, i.e., applying a de-trending step beforehand. This having said, it may seem intuitive that the \FL\ test is less effected by a change in mean.
%
However, the behavior of both tests under a change in mean is qualitatively very similar. Since, under a change in mean, the eccentricity $\rho_{FG}$ is always positive, we have the same picture as for Gini's mean difference. Either test can be more powerful than the other, and the ranking is solely determined by the order of $\theta_F$ and $\theta_G$. If $\theta_F < \theta_G$, i.e., the variance increases from $F$ to $G$, the \FF\ test is more powerful. If $\theta_F > \theta_G$, the \FL\ test is more powerful. This holds regardless of the direction of the change in mean or where the change occurs.

By Corollary \ref{cor:consistency}, we have that, as long as there is a change in variance, both tests are consistent. 
Corollary \ref{cor:consistency} tells further that, if there is \emph{no} change in variance, a change in mean renders both tests consistent as well. The same holds true for Gini's mean difference, see Example \ref{ex:gmd:change-in-mean}.

Furthermore,  by Proposition \ref{prop:max}, one can easily construct fixed alternatives where both argmax change location estimators are consistent and fixed alternatives where both are \emph{not} consistent, simply by making the difference of the means sufficiently small or large compared to the difference of the variances. 
Incidentally, this is also similarly true for Gini's mean difference, see Example \ref{ex:gmd:location}.

Last but not least, one can use Proposition \ref{prop:var} to construct interesting sequences of \emph{local} alternatives. For instance, consider the sequence of local alternatives
\[
	F = N(0,1)
	\quad \mbox{ and } \quad
	G_n = N(n^{-1/4} \mu, \, 1 + n^{-1/2} \Delta)
\]
for any real numbers $\mu, \Delta \neq 0$.
This sequence satisfies Assumption \ref{ass:local1} (1), i.e., $\theta_F$ and $\theta_G$ approach each other at the $1/\sqrt{n}$ rate, but not Assumption \ref{ass:local1} (2). The processes 
$( n^{-1/2} D^F_n([n t]))_{0 \le t \le 1}$
and
$( n^{-1/2} D^L_n([n t]))_{0 \le t \le 1}$
converge to different limit processes in this case, which can be seen from the proof of Theorem \ref{th:local}, and hence both tests have different asymptotic power.

\begin{table}[t]
	\caption{\small Sample Variance. 
	 Empirical rejection frequencies (\%) at the 5\% significance level of the \FF\ (FvsF) and \FL\ (FvsL) change-point test. Studentized with long-run variance estimation (bandwidth $b_n = n^{1/3}$). 2000 runs. 
		\label{tab:var}
	}
	\renewcommand{\arraystretch}{1.1}
	\begin{center}
		\begin{tabular}{rrr|rrrr}  
		             &                  &                  & \multicolumn{4}{c}{Sample size $n$} \\
        Scenario &  change location $\tau^\star$ &   test statistic &               63 &              250 &             1000 &             4000 \\ 
				\hline 
 Null hypothesis  &                  &             FvsF &              2.9 &              3.2 &              4.3 &              4.4 \\ 
                  &                  &             FvsL &              3.3 &              3.2 &              4.3 &              4.4 \\[1.0ex] 								
				\hline 
    Alternative 1 &             0.25 &             FvsF &              6.0 &             18.6 &             33.8 &             42.0 \\ 
                  &             0.25 &             FvsL &              6.4 &             18.3 &             33.9 &             41.8 \\[1.0ex] 
                  &              0.5 &             FvsF &             24.5 &             50.6 &             63.9 &             68.7 \\ 
                  &              0.5 &             FvsL &             26.0 &             50.5 &             63.6 &             68.6 \\[1.0ex] 
                  &             0.75 &             FvsF &             17.6 &             35.4 &             43.2 &             45.6 \\ 
                  &             0.75 &             FvsL &             22.1 &             35.9 &             43.4 &             45.6 \\[1.0ex]									
				\hline 
    Alternative 2 &             0.25 &             FvsF &             22.5 &             34.4 &             42.4 &             45.8 \\ 
                  &             0.25 &             FvsL &             20.4 &             34.1 &             42.4 &             45.8 \\[1.0ex] 
                  &              0.5 &             FvsF &             28.7 &             51.3 &             63.5 &             68.2 \\ 
                  &              0.5 &             FvsL &             26.7 &             51.2 &             63.4 &             68.2 \\[1.0ex] 
                  &             0.75 &             FvsF &              8.6 &             20.1 &             33.1 &             40.5 \\ 
                  &             0.75 &             FvsL &              7.6 &             20.0 &             33.1 &             40.6 \\[1.0ex]
	\hline 
		\end{tabular}
	\end{center}
\end{table}
\color{black}

\paragraph*{Simulation results.} We consider the exact same set-up as in Section \ref{sec:gmd}. Again, we apply studentization with a long-run variance estimate using the Bartlett kernel and the bandwidth $b_n = n^{1/3}$. 
%
The results are given in Table~\ref{tab:var}.
For $n = 250$ and larger, the power of both tests is essentially the same. For $n = 63$, there are some noticeable differences. Except for $\tau^\star = 0.75$ under Alternative 1, the difference stays within a 2\%-points margin. It should be noted that the differences for $n = 63$ are less pronounced if a smaller bandwidth in the long-run variances estimation is used.

%
%
%
%
%
%
%
%
%
%

\subsection{Sample covariance}
The sample covariance 
\[
 U_{1:n} =\frac{1}{n(n-1)}\sum_{1\leq i<j\leq n}(X_i-X_j)(Y_i-Y_j) = \frac{1}{n-1}\sum_{i = 1}^n (X_i-\bar{X}_{1:n}) (Y_i-\bar{Y}_{1:n})
\] 
of a sequence of bivariate random variables $(X_1,Y_1), \ldots, (X_n,Y_n)$ is a U-statistic with kernel
$h((x_1,y_1),(x_2,y_2)) = (x_2-x_1)(y_2-y_1)/2$. Letting $F$ and $G$ be the bivariate marginal distributions before and after the change-point, respectively, we have similarly to Proposition \ref{prop:var}
\[
	\rho_{FG} =  \frac{1}{2} (E_F X-E_G X)(E_F Y-E_G Y)
\]
if $F$ and $G$ have finite second moments. The calculation can be found in Appendix \ref{app:examples} along with the proof of Proposition \ref{prop:var}.
Hence much of what has been said about the sample variance is analogously true for the sample covariance. We turn our attention to another measure of association in the following.

%
%
%
%
%
%
%
%
%
%

\subsection{Kendall's tau}
\label{sec:tau}

Now $F$ and $G$ are bivariate distribution functions. 
\be \label{eq:kendall}
	h((x_1,y_1),(x_2,y_2)) = \Ind{(x_2-x_1)(y_2-y_1) > 0} - \Ind{(x_2-x_1)(y_2-y_1) < 0}. 
\ee
%
Contrary to Gini's mean difference and the sample variance, for this kernel $h$, the eccentricity $\rho_{FG}$ may be positive, zero, or negative. The ranking of the tests is not solely determined by the sign of $\theta_G - \theta_F$. The next proposition states that under certain symmetry conditions on $F$ and $G$, the eccentricity is zero. 
\begin{proposition} \label{prop:kendall:symmetry}
Let $(X,Y) \sim F$ 
and $G$ be such that $(X,-Y) \sim G$, then, for the kernel $h$ given by (\ref{eq:kendall}), we have $\rho_{FG} = 0$. 
\end{proposition}

\begin{figure}[t]%
\centering
\includegraphics[width=0.7\columnwidth]{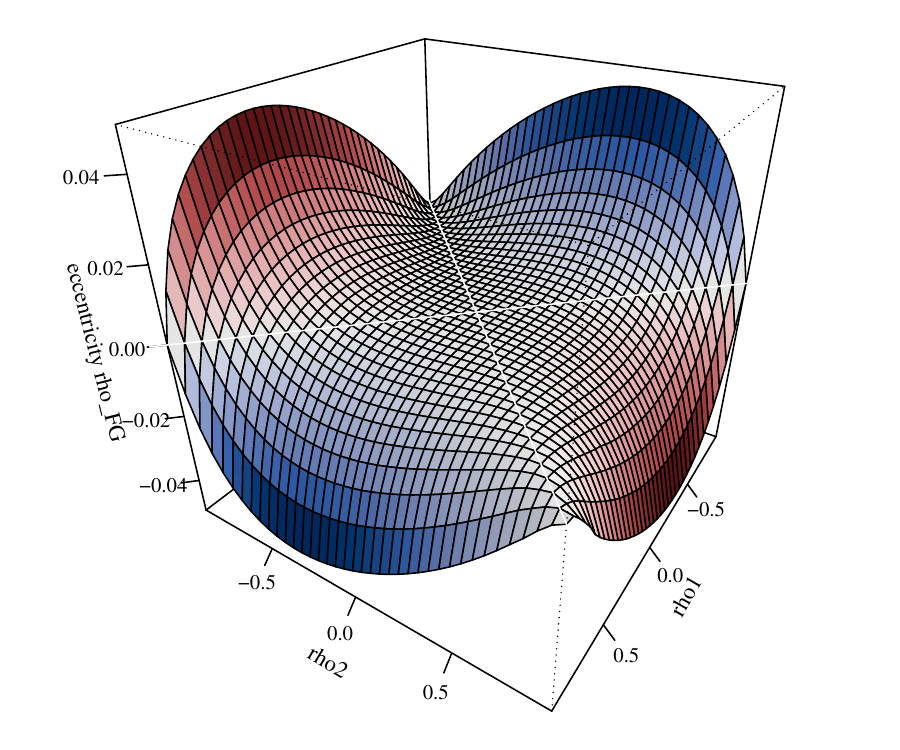}%
\caption{The plot of Example \ref{ex:kendall:normal}: It shows the eccentricity $\rho_{FG}$ on the z-axis as a function of $\rho_1$ and $\rho_2$, i.e., the correlation of the bivariate centered Gaussian time series before and after the change, respectively. In the blue region, the \FF\ test is more efficient.
}%
\label{fig:kendall}%
\end{figure}
\begin{example}[Bivariate normal distribution] \label{ex:kendall:normal}
\rm 
Consider the bivariate distribution functions $F$ and $G$ with
\[	F = N_2\left(\begin{pmatrix} 0 \\ 0\end{pmatrix},\ \begin{pmatrix} 1 & \rho_1 \\ \rho_1 & 1 \end{pmatrix} \right) 
	\quad \mbox{and} \quad
		G = N_2\left(\begin{pmatrix} 0 \\ 0\end{pmatrix},\ \begin{pmatrix} 1 & \rho_2 \\ \rho_2 & 1 \end{pmatrix} \right) 
\]
for correlation parameters $-1 < \rho_1, \rho_2 < 1$. The corresponding values of Kendall's tau are
$\theta_F = (2/\pi) \arcsin(\rho_1)$ and $\theta_G = (2/\pi) \arcsin(\rho_2)$. 
Figure \ref{fig:kendall} depicts the eccentricity $\rho_{FG}$ as function of $\rho_1$ and $\rho_2$ and was created by means of numerical integration. 
The surface is shaded in blue where the \FF\ test is better, i.e., where $\rho_{FG}$ has the same sign as $\theta_G - \theta_F$. It is shaded in red where \FL\ is better. On the diagonals, i.e., for $\rho_1 = \rho_2$ (no change) and $\rho_1 = - \rho_2$ (the situation of Proposition \ref{prop:kendall:symmetry}), the eccentricity is zero. 
\end{example}

\begin{table}[t]
	\caption{\small Kendall's $\tau$. 
	Empirical rejection frequencies (\%) at the 5\% significance level of the \FF\ (FvsF) and \FL\ (FvsL) change-point test.
	2000 runs. 
	  With long-run variance estimation (bandwidth $b_n = n^{1/3}$).
		\label{tab:tau}
	}
	\renewcommand{\arraystretch}{1.1}
	\begin{center}
		\begin{tabular}{lc@{\qquad }c@{\quad }|@{\quad }rrrr}  
		             &  change                 	&    test       & \multicolumn{4}{c}{sample size $n$} \\
        scenario &  location $\tau^\star$  	&    statistic  &           63 &              250 &             1000 &             4000 \\ 
				\hline 
              NH &                  				&          FvsF &          5.3 &                4.0 &            4.4 &              4.8 \\ 
$\rho_1 = \rho_2 = -3/\sqrt{n}$&    				&          FvsL &          3.8 &                4.0 &            4.0 &              4.8 \\[2.0ex] 
				\hline 
   Alternative 1 &              0.5 				&          FvsF &         55.2 &             68.0 &             71.9 &             73.4 \\ 
$\rho_1 = -3/\sqrt{n}$, $\rho_2 = 3/\sqrt{n}$ &     &  FvsL &         53.8 &             67.6 &             71.9 &             73.7 \\[2.0ex] 
   Alternative 2 &              0.5 				&          FvsF &         62.1 &             68.2 &             72.2 &             72.9 \\ 
$\rho_1 = 0$, $\rho_2 = 6/\sqrt{n}$ &     	&          FvsL &         70.1 &             69.2 &             71.5 &             73.2 \\[2.0ex] 
   Alternative 3 &              0.5 				&          FvsF &         76.6 &             70.8 &             71.0 &             71.7 \\ 
$\rho_1 = -6/\sqrt{n}$, $\rho_2 = 0$&     	&          FvsL &         69.7 &             69.4 &             71.4 &             71.6 \\[2.0ex] 
	\hline 
		\end{tabular}
	\end{center}
\end{table}

\paragraph*{Simulation results}
Table \ref{tab:tau} contains rejection frequencies of the studentized \FF\ and the \FL\ change-point tests for Kendall's tau at the asymptotic 5\% significance level based on 2000 runs for each case. The long-run variance estimator uses a Bartlett kernel and bandwidth $b_n = n^{1/3}$, and is explicitly stated in Appendix \ref{app:data}. The data sequence are independent bivariate normal variates as in Example \ref{ex:kendall:normal}. Four scenarios are considered:
\begin{enumerate}[(1)]
\item Null hypothesis: $\rho_1 = \rho_2 = -3/\sqrt{n}$,
\item Alternative 1: $\rho_1 = -3/\sqrt{n}$, $\rho_2 = 3/\sqrt{n}$,
\item Alternative 2: $\rho_1 = 0$, $\rho_2 = 6/\sqrt{n}$,
\item Alternative 3: $\rho_1 = -6/\sqrt{n}$, $\rho_2 = 0$.
\een
The change always occurs in the middle of the series at $\tau^\star = 1/2$. 
With the sample-size dependent change height, we mimic the local-alternative setting, and obtain non-trivial powers for all sample sizes. We find:
Both tests keep the nominal 5\% level for sample sizes of at least $n = 250$. For large $n$ and changes of small magnitude, both tests have similar power. For $n = 4000$, all rejection frequencies are within a 0.5\%-point margin of each other. This is in line with Theorem \ref{th:local}.
For small $n$ and large changes, we observe substantial differences in power, which are in line with Theorem \ref{th:fixed} and its consequences:
Under all three alternatives, we have $\theta_G - \theta_F > 0$. By Example \ref{ex:kendall:normal} we know that, under Alternative 2, the eccentricity $\rho_{FG}$ is negative, and hence by the criterion from Section \ref{sec:fixed} the \FL\ test is more powerful. Indeed, for $n = 63$, we observe an 8\%-points higher power of the \FL\ test. Under Alternative 3, the picture is reversed: The eccentricity is positive, has the same sign as $\theta_G - \theta_F$. Accordingly, we find the \FF\ test to be more powerful by the same magnitude. Under Alternative 1, both tests have similar power with a difference of about 1.5\%-points for $n=63$ and less for higher sample sizes. This illustrates Proposition \ref{prop:kendall:symmetry}.

\section{Data example}
\label{sec:data}

\begin{figure}[t]%
\centering
\includegraphics[width=0.9\columnwidth]{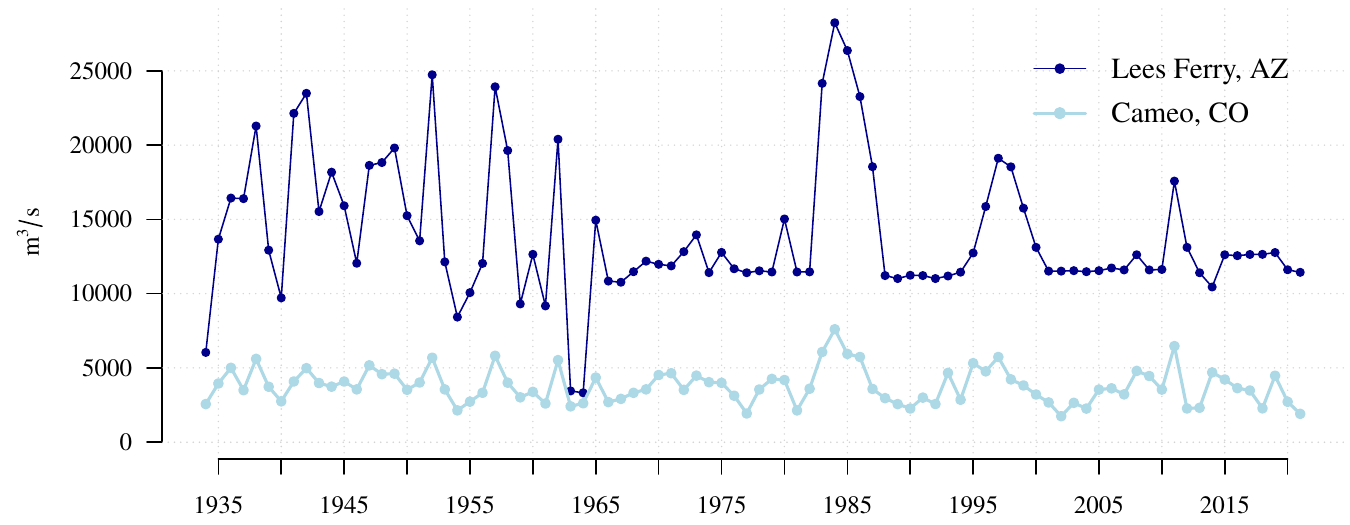}%
\caption{Annual mean discharge ($m^3/s$) of the Colorado river at Lees Ferry, AZ (darkblue) and Cameo, CO (lightblue) from 1934 to 2021.}%
\label{fig:1}%
\end{figure}
Figure \ref{fig:1} shows the annual mean discharge ($m^3/s$) of the Colorado river at Lees Ferry, AZ, and Cameo, CO, in the 88 years from 1934 to 2021. 
The data are taken from the webpages of the U.S.\ Geological Service.\footnote{https://waterdata.usgs.gov/nwis} 
\begin{figure}[t]%
\centering
\includegraphics[width=0.9\columnwidth]{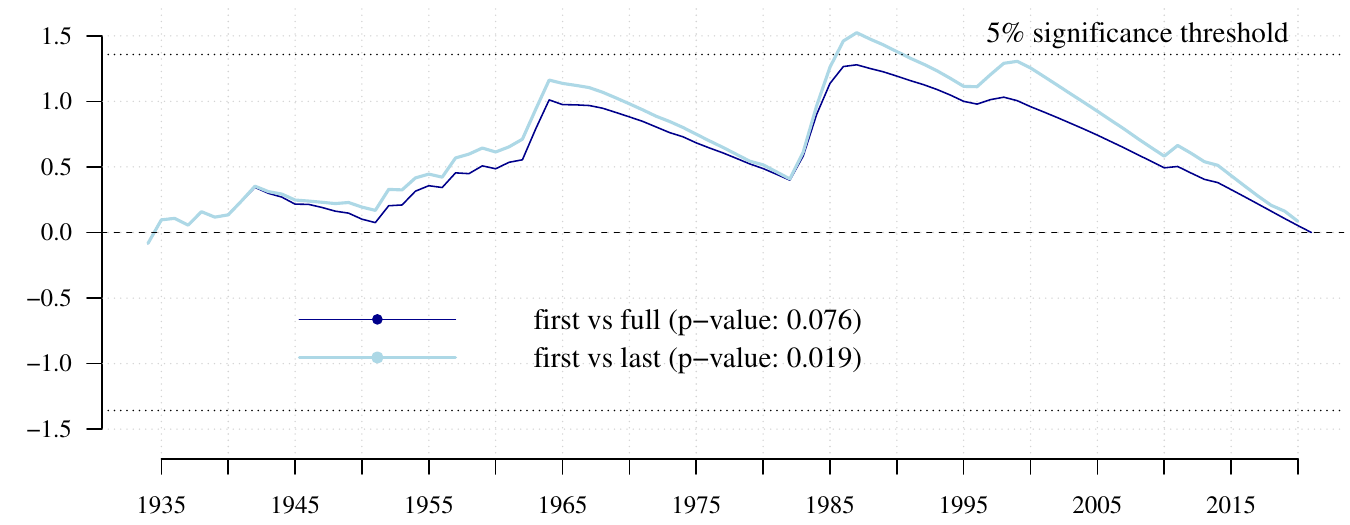}%
\caption{Gini's mean difference: studentized change-point processes \FF\ and \FL\ applied to the Lees Ferry series.}%
\label{fig:2}%
\end{figure}
\begin{figure}[t]%
\centering
\includegraphics[width=0.9\columnwidth]{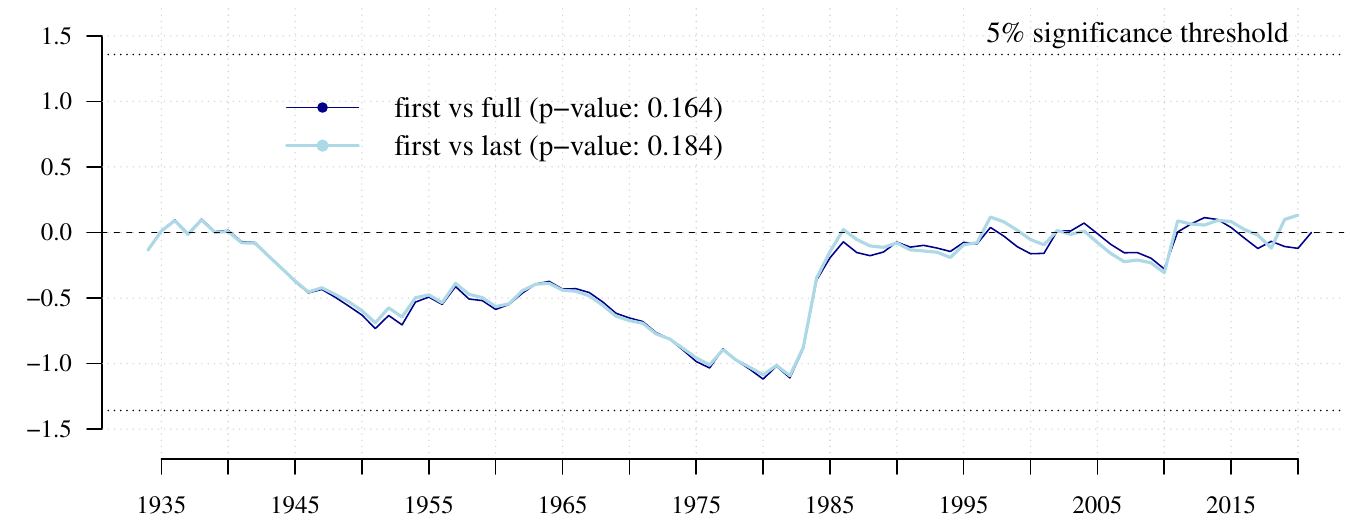}%
\caption{Gini's mean difference: studentized change-point processes \FF\ and \FL\ applied to the Cameo series.}%
\label{fig:3}%
\end{figure}
Lees Ferry, AZ, is downstream from Glen Canyon Dam, which was built from 1956 to 1964 and creates Lake Powell. The resulting structural change in the annual mean discharge series is clearly visible, affecting particularly the variability of the series. 
Incidentally, the southwestern North American megadrought in recent years \citep[e.g.][]{overpeck:2020, williams:2022} does not become apparent at the series: The discharge downstream from Lake Powell remained largely constant (while its water level has been dramatically decreasing).

Figure \ref{fig:2} shows the studentized change-point processes \FF\ and \FL, cf.~(\ref{eq:process.FvsF}) and (\ref{eq:process.FvsL}), for Gini's mean difference applied to the Lees Ferry sequence. All tests in this section use a long-run variance estimate with Bartlett kernel and 
bandwidth $b = 2$, see Appendix \ref{app:data} in the supplementary material for details.
From the results of Section \ref{sec:gmd} we know that under alternatives where the scale \emph{decreases}, as it is the case here, the \FL\ change-point test is more powerful than the \FF\ test. 
Indeed, the \FL\ change-point process deviates further from zero, clearly exceeding the asymptotic 5\% significance threshold of $\pm 1.3581$, yielding a p-value of 0.019. The \FF\ change-point process narrowly stays within the 5\% threshold, and the test returns a \emph{mildly significant} p-value of 0.076.

Cameo, CO, is located upstream from Lake Powell and is not affected by any major hydrological constructions. A pronounced change in scale is neither visible nor plausible for this series. Both Gini's-mean-difference change-point processes, \FF\ and \FL, are very close to each other, returning p-values of 0.164 and 0.184, respectively, cf.~Figure~\ref{fig:3}.

Another expectable consequence of the construction of the dam is a lower correlation between discharges upstream and downstream from the reservoir. Figure~\ref{fig:4} shows both change-point processes for Kendall's tau rank correlation coefficient between the two discharge series.
In this case, the \FF\ process deviates further than the \FL. The former yields a p-value of 0.006, the latter of 0.02. This fits in with the observations of Section \ref{sec:tau}, where we noted that -- at least for the normal model -- the \FF\ Kendall's-tau test is more powerful than \FL\ under alternatives where the \emph{absolute value} of the correlation coefficient decreases. Here, Kendall's tau essentially decreases from 0.77 in the period 1934--1962 to 0.43 in the period 1965--2021. 
So here we find a similar picture as in Figure \ref{fig:2} with roles of the two tests interchanged. 


\begin{figure}%
\centering
\includegraphics[width=0.9\columnwidth]{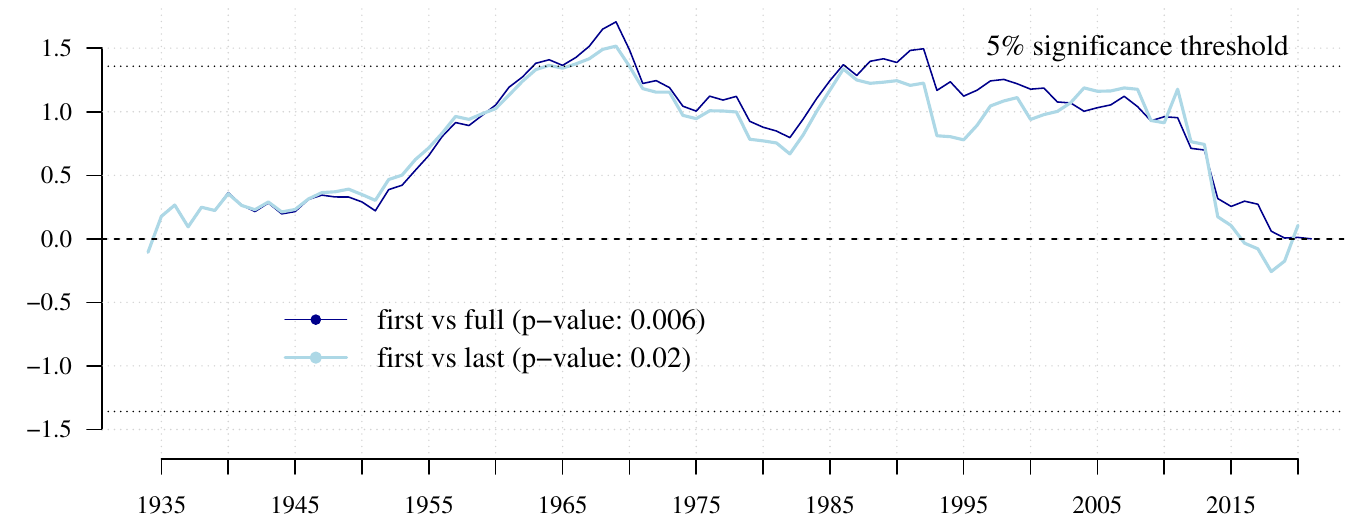}%
\caption{Kendall's $\tau$ between Lees Ferry and Cameo: studentized change-point processes \FF\ and \FL.}%
\label{fig:4}%
\end{figure}

\section{Conclusion}
\label{sec:conclusion}

We studied two construction principles of generalized CUSUM test statistics: For each potential change-point $k$, the \FF\ approach compares the estimate $\hat\theta_{1:k}$ 
evaluated at the first $k$ observations to the estimate $\hat\theta_{1:n}$ evaluated at the whole sample, whereas the \FL\ approach compares $\hat\theta_{1:k}$ to the estimate $\hat\theta_{(k+1):n}$ evaluated at the remaining $n-k$ observations. 

Within the huge body of literature devoted the change-point problem, authors usually consider one the of two approaches without any mention of the respective other or a motivation of their choice. An exception is \citet{sen:1983}, who remarks that they are asymptotically equivalent.
When discussing the question with colleagues, we encountered two opinion camps: Either it is assumed that both approaches are generally very similar (as they are identical for the original CUSUM test with the sample mean) or it is expected that the \FL\ is generally more powerful.

In this paper, we give a thorough comparison of the two approaches for the class of U-statistics. It turns out that the answer is not as unanimous as repeatedly believed:
First, the two test approaches have common properties: Under the null hypothesis and local alternatives, where the change-point processes $(D^F_n([n t]))_{0\le t \le 1}$ and $(D^L_n([n t]))_{0\le t \le 1}$ are scaled with $1/\sqrt{n}$, they have the same limit. This can be summarized by saying, they are asymptotically equivalent. 

However, when scaling with $1/n$ in the fixed-alternatives setting, different limits appear, which is reflected by power differences in small to moderate sample sizes.
We identify a simple criterion to determine which test is more powerful in which situation: 
the sign of the eccentricity $\rho_{FG}$ relative to the sign of $\theta_G - \theta_F$.
The eccentricity $\rho_{FG}$ of the kernel $h$ at distributions $F$ and $G$ describes the offset of the mixed parameter $\theta_{FG}$ from the center of the interval between $\theta_F$ and $\theta_G$. 
 Consequently, since the eccentricity is symmetric in $F$ and $G$, we have that, for any alternative where \FF\ is more powerful, an equally relevant alternative can be given where \FL\ is more powerful by the same order of magnitude -- by simply exchanging $F$ and $G$. 

The ensuing question for future research is if similar criteria can be given for other classes of statistics, e.g., quantile-based estimators (e.g.~median, interquartile range), rank-based statistics (e.g.~Spearman's rho), or the large class of differentiable functionals of statistics.


%
%

\section*{Acknowledgments}
We thank Svenja Fischer for pointing us to the data source. All figures and simulations were done in R 4.1.3 \citep{R} using packages 
mvtnorm \citep{mvtnorm}, pcaPP \citep{pcaPP}, and cubature \citep{cubature}.

\appendix

\section*{Appendices}

There are in total four appendices: Appendix \ref{app:pned} presents the concept of near epoch dependence in probability ($P$-NED), where we essentially follow \citet{dehling:vogel:wendler:wied:2017}. This class of weakly dependent processes extends the $L_2$-NED class to essentially arbitrarily heavy-tailed marginal distributions. Within this framework we state regularity assumptions on the data process and on the U-statistics that guarantee the high-level conditions we impose in the theorems throughout the paper.
Appendix \ref{app:asym} contains the proofs of the results of Section \ref{sec:asym}. Appendix \ref{app:examples} contains proofs of Section \ref{sec:examples}.
Appendix \ref{app:data} explicitly states the long-run variance estimators used in the simulations and the data examples.


\section{A general, moment-free framework for dependent series}
\label{app:pned}

\begin{definition}[Absolute regularity] \rm 
Let $(\Omega,\F,P)$ be a probability space. 
\label{def:absolute.regularity}
\begin{enumerate}[(1)]
\item
For two sub-$\sigma$-fields $\calA$, $\calB$ of $\F$, we define the absolute regularity coefficient
\[     
	 	\beta(\calA,\calB) = E \left[ \esssup\left\{ | P(A|\calB) - P(A)| \, : \, A \in \calA  \right\} \right].
\]
The absolute regularity coefficient is a measure of dependence between the $\sigma$-fields $\calA$ and $\calB$, it lies between $0$ and $1$, and equals 0 if $\calA$ and $\calB$ are independent.
\item
Let $(Z_n)_{n \in \Z}$ be an $r$-variate, strongly stationary stochastic processes on $(\Omega,\F,P)$. For $k \le n$, let $\F_k^n = \sigma(Z_k,\ldots,Z_n)$, where also $k = -\infty$ and $n = \infty$ are permitted. The process $(Z_n)_{n \in \Z}$ is called \emph{absolutely regular} if the absolute regularity coefficients
\[
		\beta_k = \beta(\F^0_{\!\!-\infty},\F_k^{\infty}), \qquad k \ge 1, 
\]
converge to zero as $k \to \infty$. 
\end{enumerate}
\end{definition}
Absolute regularity is also referred to as $\beta$-mixing and goes back to \citet{volkonskii:rozanov:1959}. It is a mixing condition stronger than $\alpha$-mixing and weaker than $\phi$-mixing. However, we do not consider absolutely regular processes themselves, but approximating functionals of such absolutely regular processes. Specifically, we use the concept  of near epoch dependence in probability ($P$-NED). It was introduced by \citet{dehling:vogel:wendler:wied:2017} and generalizes the $L_2$-NED concept. It is as general as $L_2$-NED concerning the serial dependence, but imposes no moment conditions. This is useful for bounded U-statistics like Kendall's tau. Their asymptotics are valid without any moment assumptions, and $P$-NED provides a framework for weakly dependent and very heavy-tailed sequences.

\begin{definition}\rm  \label{def:pned}
Let $(Z_{in})_{i \in \Z, n \in \N}$ be an array of $r$-variate stochastic processes ($n$ denoting the row index).
The array $(X_{in})_{i \in \Z, n \in \N}$ of $p$-variate stochastic processes is called \emph{near epoch dependent in probability} or short \emph{$P$-near epoch dependent} ($P$-NED) on the array $(Z_{in})_{n \in\Z}$ if there is a sequence of approximating constants $(a_k)_{k \in \N}$ with $a_{k} \to 0$ as $k \to \infty$, functions $f_{n,i,k}:\R^{r \times (2k+1)} \to \R^p$, $k \in \N$, and a non-increasing function $\Phi: (0,\infty) \to (0,\infty)$ such that 
\be \label{eq:l_0}
	P \left( \left|  X_{in}-f_{n,i,k}(Z_{i-k,n},\ldots,Z_{i+k,n})  \right|_1 > \varepsilon\right) \ \le \ a_k \Phi(\varepsilon) \\
\ee
for all $k, n \in \N$, $i \in \Z$ and $\varepsilon > 0$. Here $|x|_1$ denotes the $L_1$-norm of the vector $x \in \R^p$.
\end{definition}

For the most part of Section \ref{sec:asym} of the paper we consider triangular arrays of processes. Hence we define $P$-NED for arrays, noting that this includes $P$-NED sequences as special case by simply taking all rows identical.

For the various functional limit theorems for the processes $D^F_n$ and $D^L_n$, cf.~(\ref{eq:process.FvsF}) and (\ref{eq:process.FvsL}), respectively, in the main paper, we require the following technical Assumptions \ref{ass:weak_dependence}, \ref{ass:uniform_moments}, and \ref{ass:variation_condition} on the data array $(X_{in})_{i \in \Z, n \in \N}$ and the kernel $h$.

\begin{assumption}[Weak dependence]
\label{ass:weak_dependence} 
The array $(X_{in})_{i \in \Z, n \in \N}$ is $P$-NED on an absolutely regular array $(Z_{in})_{i \in \Z, n \in \N}$, and there is a $\delta > 0$ such that 
\[
	a_k\Phi(k^{-6}) = O (k^{-6(2+\delta)/\delta}), \qquad
	\sum_{k=1}^\infty k\beta_k^{\delta/(2+\delta)} < \infty,
\]
where $\beta_k$, $k \in \N$ are the absolute regularity coefficients (Definition \ref{def:absolute.regularity}) and $a_k$, $k \in \N$, and $\Phi$ are as in Definition \ref{def:pned}.
\end{assumption}
Furthermore, a uniform moment condition on $h(X_i,X_j)$ is required. We do not impose any moment conditions on the data sequence itself. 
\begin{assumption}[Uniform moments]
\label{ass:uniform_moments} There is a constant $M > 0$ such that for all $k,n\in\N$ and $i,j\in\{1,\ldots,n\}$
\[
		E\left|h(X_{in},X_{jn})\right|^{2+\delta}\leq M,
\]\[
		E\left|h(f_{n,i,k}(Z_{i-k,n},\ldots,Z_{i+k,n}),f_{n,j,k}(Z_{j-k,n},\ldots,Z_{j+k,n}))\right|^{2+\delta}\leq M.
\]
\end{assumption}
Assumptions \ref{ass:weak_dependence} and \ref{ass:uniform_moments} are linked via $\delta$. Weaker moment conditions have to be paid for by a faster decay of the short-range dependence coefficients and vice versa. The last assumption is also known as the variation condition and was introduced by \citet{denker:keller:1986}. It is a form of Lipschitz continuity of the kernel $h$ with respect to $F$.


\begin{assumption}[Variation condition]
\label{ass:variation_condition}
Let $(\tilde{X}_{in})_{i \in \Z, n \in \N}$ be an independent copy of the array $(X_{in})_{i \in \Z, n \in \N}$. 
There are constants $L,\epsilon_0>0$ such that for all $\epsilon\in(0,\epsilon_0)$, all $n\in\N$, $i,j\in\{1,\ldots,n\}$
\[ 
 E\left(\sup_{|x-X_{in}|\leq \epsilon,|y-\tilde{X}_{jn}|\leq \epsilon}\left|h\left(x,y\right)-h\left(X_{in},\tilde{X}_{jn}\right)\right|\right)^2\leq L\epsilon.
\]  
\end{assumption}


\section{Proofs of Section \ref{sec:asym} and further results}
\label{app:asym}

Proposition \ref{prop:local:degenerate} below will be employed in both Subsections \ref{app:null} \emph{Proofs of Section \ref{sec:null} (Null hypothesis)} and
\ref{app:local} \emph{Proofs of Section \ref{sec:local} (Local alternatives) and further results}. It is formulated in terms of the array 
$(X_{in})_{i \in \Z, n \in \N}$ as constructed at the beginning of Section \ref{sec:local} (Local alternatives) in the main paper. When applying it to the null-hypothesis case, it is understood that the array consists of identical rows $(X_i)_{i\in\Z}$. 
As in the main paper, we henceforth usually omit the commas separating several indices as long as there is no ambiguity.
Let $(\tilde{X}_{in})_{i \in \Z, n\in\N}$ be an independent copy of $(X_{in})_{i \in \Z, n\in\N}$ and define
\begin{align} \label{theta_ijn.repeat}
		\theta_{ijn} & =E\left[h(X_{in},\tilde{X}_{jn})\right],\nonumber \\
		h_{1ijn}(x)  & =E\left[h(x,X_{jn})\right]-\theta_{ijn}, \\
		h_{2ijn}(x,y)& =h(x,y)-E\left[h(X_{in},y)\right]-E\left[h(x,X_{jn})\right]+\theta_{ijn}. \nonumber
\end{align}
This is the same as (\ref{theta_ijn}) in the main paper. Define further 
\begin{align*}
U_{n,1}^{(2)}(l)& =\frac{2}{l(l-1)}\sum_{1\leq i<j\leq l}h_{2ijn}(X_{in},X_{jn})-\frac{2}{n(n-1)}\sum_{1\leq i<j\leq n}h_{2ijn}(X_{in},X_{jn}),\\
U_{n,2}^{(2)}(l)& =\frac{2}{l(l-1)}\sum_{1\leq i<j\leq l}h_{2ijn}(X_{in},X_{jn})-\frac{2}{(n-l)(n-l-1)}\sum_{l+1\leq i<j\leq n}h_{2ijn}(X_{in},X_{jn}).
\end{align*}
for $1 \le l \le n$. These terms can be understood as the degenerate parts of the U-statistic processes $D_n^F$ and $D_n^L$, respectively.
The following proposition presents moment conditions under which these degenerate parts vanish asymptotically. It is used in the proofs in Theorems \ref{th:nh}, \ref{th:local}, and \ref{th:fixed}.

\begin{proposition}\label{prop:local:degenerate} \mbox{\\} 
\begin{enumerate}[(1)]
\item \label{num:1} If
\begin{equation}\label{condition1}
			E\bigg[\bigg(\max_{l\leq n}\Big|\sum_{1\leq i<j\leq l}h_{2ijn}(X_{in},X_{jn})\Big|\bigg)^2\bigg]=o( n^{3})
\end{equation}
then
\begin{equation}
			\frac{[\lambda n]}{\sqrt{n}}U_{n,1}^{(2)}([\lambda n]) \to  0
\end{equation}
in probability uniformly in $0 \le \lambda \le 1$. 
\item \label{num:2}
If additionally
\begin{equation}\label{condition2}
	E\bigg[\bigg(\max_{l\leq n}\Big|\sum_{l\leq i<j\leq n}h_{2ijn}(X_{in},X_{jn})\Big|\bigg)^2\bigg]=o( n^{3}),
\end{equation}
then
\begin{equation}
\frac{[\lambda n](n-[\lambda n])}{n^{3/2}}U_{n,2}^{(2)}([\lambda n]) \to 0
\end{equation}
in probability uniformly in $0 \le \lambda \le 1$.
\end{enumerate}
\end{proposition}

\begin{proof}
For part (\ref{num:1}):
\begin{multline*}
E\bigg[\Big(\max_{\lambda\in[0,1]}\big|\frac{[\lambda n]}{\sqrt{n}}U_{n,1}^{(2)}([\lambda n])\big|\Big)^2\bigg]\\
\leq 2E\bigg[\bigg(\max_{l\leq n}\frac{2}{\sqrt{n}(l-1)}\Big|\sum_{1\leq i<j\leq l}h_{2ijn}(X_{in},X_{jn})\Big|\bigg)^2\bigg]\\
+2E\bigg[\bigg(\frac{2}{\sqrt{n}(n-1)}\sum_{1\leq i<j\leq n}h_{2ijn}(X_{in},X_{jn})\bigg)^2\bigg]
\end{multline*}
For the second summand, we have that
\begin{multline*}
2E\bigg[\bigg(\frac{2}{\sqrt{n}(n-1)}\sum_{1\leq i<j\leq n}h_{2ijn}(X_{in},X_{jn})\bigg)^2\bigg]\\
\leq \frac{32}{n^{3}}E\bigg[\bigg(\max_{l\leq n}\Big|\sum_{1\leq i<j\leq l}h_{2ijn}(X_{in},X_{jn})\Big|\bigg)^2\bigg]\rightarrow 0.
\end{multline*}
For the first summand, we obtain
\begin{multline*}
2E\bigg[\bigg(\max_{l\leq n}\frac{2}{\sqrt{n}(l-1)}\Big|\sum_{1\leq i<j\leq l}h_{2ijn}(X_{in},X_{jn})\Big|\bigg)^2\bigg]\\
\leq 2 E\bigg[\max_{1\leq k \leq \lceil \log_2(n) \rceil}\bigg(\max_{l\leq 2^k}\frac{2}{\sqrt{n}(2^{k-1}-1)}\Big|\sum_{1\leq i<j\leq l}h_{2ijn}(X_{in},X_{jn})\Big|\bigg)^2\bigg]\\
\leq \frac{128}{n}\sum_{k=1}^{\lceil \log_2(n) \rceil} E\bigg[\bigg(\max_{l\leq 2^k}\frac{1}{2^{k}}\Big|\sum_{1\leq i<j\leq l}h_{2ijn}(X_{in},X_{jn})\Big|\bigg)^2\bigg]\\
=\frac{1}{n}\sum_{k=1}^{\lceil \log_2(n) \rceil} o(2^k)=\frac{1}{n}o\bigg(\sum_{k=1}^{\lceil \log_2(n) \rceil} 2^k\bigg)=\frac{1}{n}o(2^{\log_2(n)+2})\rightarrow 0.
\end{multline*}
Hence part (\ref{num:1}) of the proposition follows by Chebyshev's inequality. For part (\ref{num:2}), note that $\frac{[\lambda n](n-[\lambda n])}{n^{3/2}}\leq \frac{[\lambda n]}{\sqrt{n}}$ and $\frac{[\lambda n](n-[\lambda n])}{n^{3/2}}\leq \frac{n-[\lambda n]}{\sqrt{n}}$. Using this, the statement can be proved in the same way.
\end{proof}

%
%
%
%
%
%

\subsection{Proofs of Section \ref{sec:null} (Null hypothesis)}
\label{app:null}

In Section \ref{sec:null}, we assume the data to be a stationary process, i.e., all rows of the array $(X_{in})$ are identical and $(X_{i1})_{i\in \Z}$ is a strongly stationary process. We write short $X_i$ for $X_{i1}$. Then the above notation simplifies to
\begin{align} \label{eq:hoeffding.repeat}
		\theta_{ijn} & = \theta =  E\left[h(X_0,\tilde{X}_0)\right],\nonumber \\
		h_{1ijn}(x)  & = h_1(x) =  E\left[h(x,X_0)\right] - \theta, \\
		h_{2ijn}(x,y)& = h_2(x,y) = h(x,y) - E\left[h(X_0,y)\right]-E\left[h(x,X_0)\right] + \theta. \nonumber
\end{align}
This is the same as (\ref{eq:hoeffding}) in the main paper. The processes $U_{n,1}^{(2)}$ and $U_{n,2}^{(2)}$ defined above can then be written shorter as
\begin{align*}
U_{n,1}^{(2)}(l)& =\frac{2}{l(l-1)}\sum_{1\leq i<j\leq l}h_{2}(X_{i},X_{j})-\frac{2}{n(n-1)}\sum_{1\leq i<j\leq n}h_{2}(X_{i},X_{j}),\\
U_{n,2}^{(2)}(l)& =\frac{2}{l(l-1)}\sum_{1\leq i<j\leq l}h_{2}(X_{i},X_{j})-\frac{2}{(n-l)(n-l-1)}\sum_{l+1\leq i<j\leq n}h_{2}(X_{i},X_{j})
\end{align*}
for $1 \le l \le n$.

We are now ready to prove Theorem \ref{th:nh}.

\begin{proof}[Proof of Theorem \ref{th:nh}]
Part (\ref{theo:hypo1}): Using the Hoeffding decomposition (\ref{eq:hoeffding.repeat}), we obtain
\[
\frac{1}{\sqrt{n}} D_n^F([nt])=\frac{1}{\sqrt{n}}\Big(\sum_{i=1}^{[nt]}h_1(X_i)-\frac{[nt]}{n}\sum_{i=1}^nh_1(X_i)\Big)+\frac{[nt]}{\sqrt{n}}U_{n,1}^{(2)}([nt]).
\]
The weak convergence of the first summand follows from Assumption \ref{ass:nh} (\ref{assu:weakcon}) by means of the continuous mapping theorem. The second summand converges uniformly to 0, see Proposition \ref{prop:local:degenerate}.

To show part (\ref{theo:hypo2}), we use the representation
\begin{equation*}
\frac{1}{\sqrt{n}}D_n^L([nt])=\frac{1}{\sqrt{n}}\Big(\sum_{i=1}^{[nt]}h_1(X_i)-\frac{[nt]}{n}\sum_{i=1}^nh_1(X_i)\Big)+\frac{[nt](n-[nt])}{n^{3/2}}U_{n,2}^{(2)}([nt]),
\end{equation*}
so we can bound the difference by 
\[
		\frac{1}{\sqrt{n}}\left|D_n^F([nt])-D_n^L([nt])\right| 
		\leq |\frac{[nt]}{\sqrt{n}}U_{n,1}^{(2)}([nt])|+|\frac{[nt](n-[nt])}{n^{3/2}}U_{n,2}^{(2)}([nt])|.
\]
Another application of Proposition \ref{prop:local:degenerate} completes the proof.
\end{proof}

%
%
%
%

Lemma \ref{lem:dependent_data} gives conditions for Theorem \ref{th:nh} to hold within the $P$-NED framework of Section \ref{app:pned}.

\begin{lemma} \label{lem:dependent_data}
Let $(X_i)_{i \in \Z}$ be a strongly stationary process which is $P$-NED on a process $(Z_i)_{i\in\Z}$ (which is cast into Definition \ref{def:pned} by considering arrays with identical rows). If Assumptions \ref{ass:weak_dependence}, \ref{ass:uniform_moments}, and \ref{ass:variation_condition} are fulfilled, then Assumption \ref{ass:nh} is satisfied.
\end{lemma}
\begin{proof}
Assumption \ref{ass:nh} (\ref{assu:weakcon}) follows by Theorem 2.5 of \citet{dehling:vogel:wendler:wied:2017}, which uses an invariance principle by \citet{wooldridge:white:1988}. 
Assumption \ref{ass:nh} (\ref{assu:h2}) follows along the lines of the proof of Lemma D.6 (i) in \citet[][supplementary material]{dehling:vogel:wendler:wied:2017}. 
%
\end{proof}

%
%
%
%
%
%

\subsection{Proofs of Section \ref{sec:local} (Local alternatives) and further results}
\label{app:local}

%
%

\begin{proof}[Proof of Theorem \ref{th:local}]
The first part of the theorem follows from the proof of Theorem B.3 and Corollary B.4 of \citet{dehling:vogel:wendler:wied:2017}. 

For the second part of the theorem, we apply the Hoeffding decomposition to both statistics:
\begin{multline*}
	\frac{1}{\sqrt{n}} D_n^F([nt])
		= \frac{[nt]}{\sqrt{n}}\frac{2}{[nt]([nt]-1)}\sum_{1\leq i<j\leq [nt]}\theta_{ijn}-\frac{[nt]}{\sqrt{n}}\frac{2}{n(n-1)}\sum_{1\leq i<j\leq n}\theta_{ijn}\\
	+\frac{[nt]}{\sqrt{n}}\frac{2}{[nt]([nt]-1)}\sum_{1\leq i<j\leq [nt]}h_{1ijn}(X_{in})-\frac{[nt]}{\sqrt{n}}\frac{2}{n(n-1)}\sum_{1\leq i<j\leq n}h_{1ijn}(X_{in})\\
	+\frac{[nt]}{\sqrt{n}}\frac{2}{[nt]([nt]-1)}\sum_{1\leq i<j\leq [nt]}h_{1jin}(X_{jn})-\frac{[nt]}{\sqrt{n}}\frac{2}{n(n-1)}\sum_{1\leq i<j\leq n}h_{1jin}(X_{jn})\\
	+\frac{[nt]}{\sqrt{n}}\frac{2}{[nt]([nt]-1)}\sum_{1\leq i<j\leq [nt]}h_{2ijn}(X_{in},X_{jn})-\frac{[nt]}{\sqrt{n}}\frac{2}{n(n-1)}\sum_{1\leq i<j\leq n}h_{2ijn}(X_{in},X_{jn})
\end{multline*}
and 
\begin{multline*}
	\frac{1}{\sqrt{n}} D_n^L([nt]) 
	= \frac{[nt](n-[nt])}{n \sqrt{n}}\frac{2}{[nt]([nt]-1)}\sum_{1\leq i<j\leq [nt]}\theta_{ijn}\\
	-\frac{[nt](n-[nt])}{n \sqrt{n}}\frac{2}{(n-[nt])(n-[nt]-1)}\sum_{[nt]< i<j\leq n}\theta_{ijn}\\
	+\frac{[nt](n-[nt])}{n \sqrt{n}}\frac{2}{[nt]([nt]-1)}\sum_{1\leq i<j\leq [nt]}h_{1ijn}(X_{in})\\
	-\frac{[nt](n-[nt])}{n \sqrt{n}}\frac{2}{(n-[nt])(n-[nt]-1)}\sum_{[nt]< i<j\leq n}h_{1ijn}(X_{in})\\
	+\frac{[nt](n-[nt])}{n \sqrt{n}}\frac{2}{[nt]([nt]-1)}\sum_{1\leq i<j\leq [nt]}h_{1jin}(X_{jn})\\
	-\frac{[nt](n-[nt])}{n \sqrt{n}}\frac{2}{(n-[nt])(n-[nt]-1)}\sum_{[nt]< i<j\leq n}h_{1jin}(X_{jn})\\
	+\frac{[nt](n-[nt])}{n \sqrt{n}}\frac{2}{[nt]([nt]-1)}\sum_{1\leq i<j\leq [nt]}h_{2ijn}(X_{in},X_{jn})\\
	-\frac{[nt](n-[nt])}{n \sqrt{n}}\frac{2}{(n-[nt])(n-[nt]-1)}\sum_{[nt]< i<j\leq n}h_{2ijn}(X_{in},X_{jn})
\end{multline*}
For the first two summands of the decompositions, we can use the fact that $\theta_{ijn}=\theta_F$ for $i,j\leq [n\tau^\star]$, $\theta_{ijn}\approx\theta_F +\Delta/(2\sqrt{n})$ for $i\leq [n\tau^\star]<j$ 
and  $\theta_{ijn}\approx\theta_F +\Delta/\sqrt{n}$ for $i,j \le [n\tau^\star]$ or $[n\tau^\star] < i,j$. Taking the limit $n\rightarrow\infty$ gives the function $\phi$ in Theorem \ref{th:local}.

The last two summands are negligible by Assumption \ref{ass:local2} (1) and Proposition \ref{prop:local:degenerate}.

In the middle four summands, we can replace $h_{1ijn}(X_{in})$ by $h_{1}(X_{i})$ making use of Assumption \ref{ass:local2} (2).
For $D_n^F$, this gives
\begin{multline*}
\frac{[nt]}{\sqrt{n}}\frac{2}{[nt]([nt]-1)}\sum_{1\leq i<j\leq [nt]}h_{1ijn}(X_{in})-\frac{[nt]}{\sqrt{n}}\frac{2}{n(n-1)}\sum_{1\leq i<j\leq n}h_{1ijn}(X_{in})\\
+\frac{[nt]}{\sqrt{n}}\frac{2}{[nt]([nt]-1)}\sum_{1\leq i<j\leq [nt]}h_{1jin}(X_{jn})-\frac{[nt]}{\sqrt{n}}\frac{2}{n(n-1)}\sum_{1\leq i<j\leq n}h_{1jin}(X_{jn})\\
\approx\frac{[nt]}{\sqrt{n}}\frac{2}{[nt]([nt]-1)}\sum_{1\leq i<j\leq [nt]}h_{1}(X_{i})-\frac{[nt]}{\sqrt{n}}\frac{2}{n(n-1)}\sum_{1\leq i<j\leq n}h_{1}(X_{i})\\
+\frac{[nt]}{\sqrt{n}}\frac{2}{[nt]([nt]-1)}\sum_{1\leq i<j\leq [nt]}h_{1}(X_{j})-\frac{[nt]}{\sqrt{n}}\frac{2}{n(n-1)}\sum_{1\leq i<j\leq n}h_{1}(X_{j})\\
=\frac{2}{\sqrt{n}}\sum_{i=1}^{[nt]}h_1(X_i)-\frac{2[nt]}{n\sqrt{n}}\sum_{i=1}^{n}h_1(X_i)=\frac{2(n-[nt])}{n\sqrt{n}}\sum_{i=1}^{[nt]}h_{1}(X_{i})-\frac{2[nt]}{n\sqrt{n}}\sum_{i=[nt]+1}^{n}h_1(X_i).
\end{multline*}
A similar calculation for $D_n^F$ leads to
\begin{multline*}
\frac{[nt](n-[nt])}{n \sqrt{n}}\frac{2}{[nt]([nt]-1)}\sum_{1\leq i<j\leq [nt]}h_{1ijn}(X_{in})\\
-\frac{[nt](n-[nt])}{n \sqrt{n}}\frac{2}{(n-[nt])(n-[nt]-1)}\sum_{[nt]< i<j\leq n}h_{1ijn}(X_{in})\\
+\frac{[nt](n-[nt])}{n \sqrt{n}}\frac{2}{[nt]([nt]-1)}\sum_{1\leq i<j\leq [nt]}h_{1jin}(X_{jn})\\
-\frac{[nt](n-[nt])}{n \sqrt{n}}\frac{2}{(n-[nt])(n-[nt]-1)}\sum_{[nt]< i<j\leq n}h_{1jin}(X_{jn})\displaybreak[0]\\
\shoveleft\approx
\frac{[nt](n-[nt])}{n \sqrt{n}}\frac{2}{[nt]([nt]-1)}\sum_{1\leq i<j\leq [nt]}h_1(X_i)\\
-\frac{[nt](n-[nt])}{n \sqrt{n}}\frac{2}{(n-[nt])(n-[nt]-1)}\sum_{[nt]< i<j\leq n}h_1(X_i)\\
+\frac{[nt](n-[nt])}{n \sqrt{n}}\frac{2}{[nt]([nt]-1)}\sum_{1\leq i<j\leq [nt]}h_1(X_j)\\
-\frac{[nt](n-[nt])}{n \sqrt{n}}\frac{2}{(n-[nt])(n-[nt]-1)}\sum_{[nt]< i<j\leq n}h_1(X_j)\\
=\frac{2(n-[nt])}{n\sqrt{n}}\sum_{i=1}^{[nt]}h_{1}(X_{i})-\frac{2[nt]}{n\sqrt{n}}\sum_{i=[nt]+1}^{n}h_1(X_i)
\end{multline*}
So the middle four summands can be approximated by the same expression, hence 
\begin{equation}
	\sup_{t\in[0,1]}\left| 
	   \frac{1}{\sqrt{n}} D_n^F([nt]) - \frac{1}{\sqrt{n}} D_n^L([nt])
	\right|
	\xrightarrow{n\rightarrow\infty}0,
\end{equation}
which completes the proof of Theorem \ref{th:local}.
\end{proof}

%
%

Next we prove Proposition \ref{prop:densities}, showing that Assumption \ref{ass:local1} is fulfilled if $F$ and $G_n$ have densities $f$ and $g_n$, respectively, and the latter approaches the former at the $(1/\sqrt{n})$ rate in some suitable sense.


\begin{proof}[Proof of Proposition \ref{prop:densities}]
For part (1) of Assumption \ref{ass:local1}, note that 
\begin{align}
\sqrt{n}(\lambda_F-\lambda_{G_n})
&=\iint h(x,y) \sqrt{n}\big(p_{\lambda_0}(x) p_{\lambda_0}(y) - p_{\lambda_0+a/\sqrt{n}}(x) p_{\lambda_0+a/\sqrt{n}}(y) \big) \, dx\, dy  \label{eq:dom-conv-3} \nonumber \\
&=  \iint h(x,y) \sqrt{n} \, \big(p_{\lambda_0}(x) -p_{\lambda_0+a/\sqrt{n}}(x)   \big) p_{\lambda_0}(y) \, dx \, dy \\
&\quad +  \iint h(x,y) \sqrt{n} \big(p_{\lambda_0}(y) -p_{\lambda_0+a/\sqrt{n}}(y)   \big) p_{\lambda_0}(x) \, dx \, dy \nonumber \\
&\quad +  \iint h(x,y) \sqrt{n} \big(p_{\lambda_0}(x) -p_{\lambda_0+a/\sqrt{n}}(x)   \big) \big( p_{\lambda_0+a/\sqrt{n}}(y)- p_{\lambda_0}(y) \big)\, dx \, dy.  \nonumber
\end{align}
For the first two summands on the right, we consider the function 
\[
  F(\lambda) = \iint h(x,y) p_\lambda(x) p_{\lambda_0}(y) \, dx \, dy.
\]
By assumptions \eqref{eq:dom-conv-1} and \eqref{eq:dom-conv-2}, we can interchange derivative and integral and obtain 
\[
  F^\prime(\lambda) = \iint h(x,y) \, \frac{\partial}{\partial \lambda} p_\lambda(x) \, p_{\lambda_0}(y) \, dx\, dy.
\]
Thus, we find for the first summand on the right hand side of \eqref{eq:dom-conv-3}
\[
   \iint h(x,y) \sqrt{n}\, \big(p_{\lambda_0}(x) -p_{\lambda_0+a/\sqrt{n}}(x)   \big) p_{\lambda_0}(y) \, dx \, dy 
\]\[
   = \sqrt{n} \big(F(\lambda_0) -F(\lambda_0+a/\sqrt{n}) \big)  \longrightarrow - a\, F^\prime(\lambda_0).
	\]
The same holds for the second summand. Concerning the third summand on the right hand side of \eqref{eq:dom-conv-3}, note that the integrand 
$h(x,y) \sqrt{n}\,  (p_{\lambda_0}(x) -p_{\lambda_0+a/\sqrt{n}}(x)) ( p_{\lambda_0+a/\sqrt{n}}(y)- p_{\lambda_0}(y) )$ converges to zero, and that it is dominated by $|h(x,y)|\, v(x) \, v(y)$, which is integrable by \eqref{eq:dom-conv-2}. Thus, the integral converges to zero, and hence part (1) of Assumption~2.4 is satisfied.
Regarding part (2) of Assumption~2.4, note that 
\[
  \sqrt{n} \rho_{FG_n} = - \iint h(x,y) \sqrt{n} (p_{\lambda_0} (x)- p_{\lambda_0+a/\sqrt{n}}(x)) (p_{\lambda_0} (y)- p_{\lambda_0+a/\sqrt{n}}(y)) dx\, dy,
\]
which equals the third integral on the right hand side of \eqref{eq:dom-conv-3}, and thus converges to zero. 
This completes the proof of Proposition \ref{prop:densities}.
\end{proof}

Example \ref{ex:scale_family} below demonstrates that the Gaussian scale family, which we consider in the simulations of Section \ref{sec:examples}, fulfills Proposition \ref{prop:densities}.

\begin{example} \label{ex:scale_family}
We consider the Gaussian scale family of densities
\[
  p_\lambda(x) = \frac{1}{\sqrt{2\pi\, \lambda^2}} e^{-\frac{x^2}{2\lambda^2}}, 
\]
where, for once, the Gaussian scale parameter is denoted by $\lambda$ for notational consistency with Proposition \ref{prop:densities}. 
The partial derivative with respect to the parameter $\lambda$ is given by
\begin{align*}
  \frac{\partial}{\partial \lambda} p_\lambda(x) &= -\frac{1}{\lambda^2 \, \sqrt{2\pi}} e^{-\frac{x^2}{2\lambda^2}} + \frac{1}{ \lambda\, \sqrt{2\pi} }e^{-\frac{x^2}{2\lambda^2}} \, \frac{x^2}{\lambda^3} 
  =  \frac{1}{\lambda^2 \, \sqrt{2\pi}} \Big[ \frac{x^2}{\lambda^2}-1  \Big] e^{-\frac{x^2}{2\lambda^2}}.
\end{align*} 
For $\lambda \in [a,b] \subset (0,\infty) $  we thus obtain the upper bound 
\[
 \Big|  \frac{\partial}{\partial \lambda} p_\lambda(x)    \Big| \leq \frac{1}{a^2\, \sqrt{2\pi}} \Big[ \frac{x^2}{a^2}+1   \Big]
 e^{-\frac{x^2}{2b^2}} \leq C\Big\{ x^2 \frac{1}{\sqrt{2\pi\, b^2}} e^{-\frac{x^2}{ b^2}} +  \frac{1}{2\pi \,\sqrt{b^2}} e^{-\frac{x^2}{2b^2}} \Big\},
\]
for some constant $C$. The same upper bound, possibly with a larger constant $C$, holds for the densities $p_\lambda(x)$,
$a\leq \lambda \leq b$. Hence, we can take $v(x)= C\big[x^2 +1\big] \frac{1}{\sqrt{2\pi\, b^2}} e^{-\frac{x^2}{ b^2}}$, and thus the integrability condition \eqref{eq:dom-conv-2} becomes 
\[
  \iint |h(x,y)| (x^2+1)(y^2+1) \frac{1}{2\pi b^2} e^{-\frac{x^2+y^2}{2b^2}} dx\, dy <\infty. 
\]
Note that the integral on the left hand side equals 
\[
   E \big\{ |h(X,Y)|  (X^2+1)(Y^2+1) \big\},
\]
where $X,Y$ are two  independent $N(0,b^2)$-distributed random variables. Thus, the integrability condition  \eqref{eq:dom-conv-2}  translates into the requirement that there exists $b>\lambda_0$ such that 
\[
   E \big\{ |h(X,Y)|  (X^2+1)(Y^2+1) \big\} <\infty,
\]
where  $X,Y$ are two independent $N(0,b^2)$-distributed random variables. 
\end{example}

%
%
%
%
%
%

\subsection{Conditions for Assumptions \ref{ass:local2} to hold (Local alternatives)}
\label{app:assumptions_local}

In this section we give specific conditions on the array $(X_{in})_{i \in \Z, n \in \N}$ and the U-statistic kernel $h$ which ensure that
Assumption \ref{ass:local2} is satisfied. We do so for two situation: for independent observations and within the $P$-NED framework of Section \ref{app:pned} for dependent observations.

We first state a technical lemma which is used in both situations. Recall the construction of the ``local-alternative array'' $(X_{in})_{i \in \Z, n \in \N}$ from the stationary process $(X_i)_{i\in\Z}$ and the row-wise stationary array $(Y_{in})_{i \in \Z, n \in \N}$ as outlined at the beginning of Section \ref{sec:local}.

\begin{lemma}
\label{lem:local1} 
If for some $\epsilon>0$ and $C>0$ and all $k<m\leq n$
\begin{align*}
E\left[\Big(\sum_{i=k+1}^m\big( h_{1i1n}(X_{in})-h_{1}(X_i)\Big)^2\right]&\leq Cn^{-\epsilon}(m-k),\\
E\left[\Big(\sum_{i=k+1}^m\big( h_{1inn}(X_{in})-h_{1}(X_i)\big)\Big)^2\right]&\leq Cn^{-\epsilon}(m-k),\\
E\left[\Big(\sum_{i=k+1}^mi\big( h_{1i1n}(X_{in})-h_{1}(X_i)\big)\Big)^2\right]&\leq Cn^{2-\epsilon}(m-k),\\
E\left[\Big(\sum_{i=k+1}^mi\big( h_{1inn}(X_{in})-h_{1}(X_i)\big)\Big)^2\right]&\leq Cn^{2-\epsilon}(m-k),
\end{align*}
then Assumption \ref{ass:local2} (2) holds.
\end{lemma}

\begin{proof} By Theorem 3 of \citet{Moricz1976}, there exists a constant $C_1$ such that 

\begin{align*}
E\left[\max_{1\leq k\leq n}\left|\frac{1}{\sqrt{n}}\sum_{i=1}^k\big( h_{1i1n}(X_{in})-h_{1}(X_i) \big)\right|^2\right]&\leq C_1n^{-\epsilon}\log^2(n)n\frac{1}{\sqrt{n}^2}\xrightarrow{n\rightarrow\infty}0,\\
E\left[\max_{1\leq k\leq n}\left|\frac{1}{\sqrt{n}}\sum_{i=1}^k\big( h_{1inn}(X_{in})-h_{1}(X_i) \big)\right|^2\right]&\leq C_1n^{-\epsilon}\log^2(n)n\frac{1}{\sqrt{n}^2}\xrightarrow{n\rightarrow\infty}0,\\
E\left[\max_{1\leq k\leq n}\left|\frac{1}{n^{3/2}}\sum_{i=1}^ki\big( h_{1i1n}(X_{in})-h_{1}(X_i) \big)\right|^2\right]&\leq C_1n^{2-\epsilon}\log^2(n)n\frac{1}{n^3}\xrightarrow{n\rightarrow\infty}0,\\
E\left[\max_{1\leq k\leq n}\left|\frac{1}{n^{3/2}}\sum_{i=1}^ki\big( h_{1inn}(X_{in})-h_{1}(X_i) \big)\right|^2\right]&\leq C_1n^{2-\epsilon}\log^2(n)n\frac{1}{n^3}\xrightarrow{n\rightarrow\infty}0.
\end{align*}
The Chebyshev inequality completes the proof of Lemma \ref{lem:local1}.
\end{proof}

\subsection*{Independent sequences}

Lemma \ref{lem:local2} below states conditions on $(X_{in})_{i \in \Z, n \in \N}$ and $h$ that ensure Assumption \ref{ass:local2} if the elements of $(X_{in})_{i \in \Z, n \in \N}$ are independent. 

\begin{lemma}
\label{lem:local2} 
Let the array $(X_i,X_{in})_{i \in \Z, n \in \N}$ as constructed at the beginning of Section \ref{sec:local} consist of independent random vectors. Further assume there are constants $\epsilon > 0$, $M >0$ such that for all $n \in \N$, $i,j \in \Z$
\be \label{eq:lem:local2:1}
	E\left[\big(h(X_1,X_i)-h(X_1,X_{in})\big)^2\right] \leq M n^{-\epsilon},
\ee 
\be \label{eq:lem:local2:2}
	E\left[\big(h(X_1,X_i)-h(X_{1n},X_{in})\big)^2\right] \leq M n^{-\epsilon}, 
\ee	
\be \label{eq:lem:local2:3}
  E\left[h^2(X_{in},X_{jn})\right]\leq M.
\ee
Then Assumption \ref{ass:local2} is satisfied.
\end{lemma}

\begin{proof} 
We show first that (\ref{eq:lem:local2:1}) and (\ref{eq:lem:local2:2}) together imply the conditions of Lemma \ref{lem:local1} and hence Assumption \ref{ass:local2} (2) (part 1). We then show that (\ref{eq:lem:local2:3}) implies Assumption \ref{ass:local2} (1) (part 2). 

Part 1: Recall that $h_{1i1n}(x)  =  E\left[h(x,X_1)\right] - E\left[h(X_{in},X_1)\right]$ and that $h_{1}(x)  =  E\left[h(x,X_1)\right] - E\left[h(X_2,X_1)\right]$, so by construction, $E[h_{1}(X_i)]=E[h_{1i1n}(X_{in})]=0$. Using the Jensen inequality for conditional expectation and the law of iterated expectations, we get
\begin{multline*}
\Var\left[h_{1}(X_i)-h_{1i1n}(X_{in})\right] =\Var\Big[E[h(X_i,X_1)|X_i,X_{in}]-E[h(X_{in},X_1)|X_i,X_{in}]\Big]\\
\leq E\left[\big(E[h(X_i,X_1)|X_i,X_{in}]-E[h(X_{in},X_1)|X_i,X_{in}]\big)^2\right]\\
\leq E\left[E\Big[\big(h(X_i,X_1)-h(X_{in},X_1)\big)^2\big|X_i,X_{in}\Big]\right]  \\
=E\left[\big(h(X_i,X_1)-h(X_{in},X_1)\big)^2\right]\leq Mn^{-\epsilon}.
\end{multline*}
By the independence of the random variables, we can conclude that
\begin{multline*}
E\left[\Big(\sum_{i=k+1}^m\big( h_{1i1n}(X_{in})-h_{1}(X_i)\big)\Big)^2\right]=\Var\left[\sum_{i=k+1}^m\big( h_{1i1n}(X_{in})-h_{1}(X_i)\big)\right]\\
=\sum_{i=k+1}^m\Var\left[ h_{1i1n}(X_{in})-h_{1}(X_i)\right]\leq Mn^{-\epsilon}(m-k)
\end{multline*}
 the first assumption of Lemma \ref{lem:local1} holds and the second assumption follows in the same way.  For the third assumption, 
\begin{multline*}
E\left[\Big(\sum_{i=k+1}^mi\big( h_{1i1n}(X_{in})-h_{1}(X_i)\big)\Big)^2\right] = \Var\left[\sum_{i=k+1}^m i\big( h_{1i1n}(X_{in})-h_{1}(X_i)\big)\right]\\
=\sum_{i=k+1}^m i^2 \Var\left[ h_{1i1n}(X_{in})-h_{1}(X_i)\right]\leq n^2\sum_{i=k+1}^m  \Var\left[ h_{1i1n}(X_{in})-h_{1}(X_i)\right]\\
\leq Mn^{2-\epsilon}(m-k)
\end{multline*}
and we can apply the same arguments for the fourth assumption to complete part 1.

Part 2: Assumption (\ref{eq:lem:local2:3}) implies that for all $n\in\N$, $i,j\in\{1,\ldots,n\}$
\begin{equation*}
E\left[h^2_{2ijn}(X_{in},X_{jn})\right]\leq 9 M.
\end{equation*}
Furthermore, for all $i',j'\in\{1,\ldots,n\}$ with $(i,j) \neq (i',j') $ we have
\begin{equation*}
\Cov(h_{2ijn}(X_{in},X_{jn}),h_{2i'j'n}(X_{i'n},X_{j'n}))=0,
\end{equation*}
see \citet[][p.~30]{lee:1990}.
By counting the summands, we get that for all $1\leq n_1\leq n_2\leq n$
\begin{align*}
E\left[\Big(\sum_{\substack{1\leq i<j\\n_1<j\leq n_2}}h_{2ijn}(X_{in},X_{jn})\Big)^2\right]&\leq C(n_2-n_1)n, \\
E\left[\Big(\sum_{\substack{ i<j\leq n\\n_1<i\leq n_2}}h_{2ijn}(X_{in},X_{jn})\Big)^2\right]&\leq C(n_2-n_1)n.
\end{align*}
Setting $Y_{n_1 n}=\sum_{1\leq i< n_1}h_{2in_1n}(X_{in},X_{n_1 n})$ and $Z_{n_1 n}=\sum_{n_1< j\leq n}h_{2 n_1 j n}(X_{n_1 n},X_{jn})$ and applying Theorem 3 of \cite{Moricz1976}, we obtain
\begin{align*}
E\bigg[\bigg(\max_{l\leq n}\Big|\sum_{1\leq i<j\leq l}h_{2ijn}(X_{in},X_{jn})\Big|\bigg)^2\bigg]&=E\bigg[\bigg(\max_{l\leq n}\Big|\sum_{n_1=1}^lY_{n_1n}\Big|\bigg)^2\bigg]\\
&\leq Cn^2(\log_2(n)+1)^2=o(n^3),\\
E\bigg[\bigg(\max_{l\leq n}\Big|\sum_{l\leq i<j\leq n}h_{2ijn}(X_{in},X_{jn})\Big|\bigg)^2\bigg]&=E\bigg[\bigg(\max_{l\leq n}\Big|\sum_{n_1=l}^nZ_{n_1n}\Big|\bigg)^2\bigg]\\
&\leq Cn^2(\log_2(n)+1)^2=o(n^3).
\end{align*}
Thus Assumption \ref{ass:local2} (1) is satisfied, which completes the proof of Lemma \ref{lem:local2}.
\end{proof}

%

\subsection*{Dependent sequences}

Analogous to Lemma \ref{lem:local2}, we next state conditions on $(X_{in})_{i \in \Z, n \in \N}$ and $h$ that ensure Assumption \ref{ass:local2} if $(X_{in})_{i \in \Z, n \in \N}$ is near epoch dependent as described in Section \ref{app:pned}. 

\begin{lemma} 
\label{lem:local3} 
Let the array $(X_{in})_{i \in \Z, n \in \N}$ from Section \ref{sec:local} and the kernel $h$ satisfy Assumptions \ref{ass:weak_dependence}, \ref{ass:uniform_moments}, and \ref{ass:variation_condition}. 
Let $(\tilde{X}_i)_{i\in\Z}$ and $(\tilde{X}_{in})_{i\in\Z, n \in \N}$ be independent copies of $(X_i)_{i\in\Z}$ and $(X_{in})_{i\in\Z, n \in \N}$, respectively. 
Further assume there are constants $\epsilon > 0$, $M >0$ such that for all $n \in \N$
\begin{align}
\label{eq:lem:local3:1}
E\left[\big(h(X_1,X'_i)-h(X_1,X'_{in})\big)^{2+\delta}\right]&\leq M n^{-\epsilon},\\
\label{eq:lem:local3:2}
E\left[\big(h(X_1,X'_i)-h(X_{1n},X'_{in})\big)^{2+\delta}\right]&\leq M n^{-\epsilon}.
\end{align}
Then Assumption \ref{ass:local2} is met.
\end{lemma}

\begin{proof} 
We show first that Assumptions \ref{ass:weak_dependence} and \ref{ass:variation_condition} together with (\ref{eq:lem:local3:1}) and (\ref{eq:lem:local3:2}) imply the conditions of Lemma \ref{lem:local1} and hence Assumption \ref{ass:local2} (2) (part 1). 
We then show that Assumptions \ref{ass:weak_dependence}, \ref{ass:uniform_moments}, and \ref{ass:variation_condition} imply Assumption \ref{ass:local2} (1) (part 2). 

Part 1: Following the proof of Lemma \ref{lem:local2}, it can be shown that
\begin{align*}
E\left[\big( h_1(X_i)-h_{1i1n}(X_{in})\big)^2\right]& \leq M n^{-\epsilon},\\
E\left[\big( h_1(X_i)-h_{1inn}(X_{in})\big)^2\right]& \leq M n^{-\epsilon}.
\end{align*}
By Assumption \ref{ass:weak_dependence} and Lemma A.2 of \citet{dehling:vogel:wendler:wied:2017}, the arrays
\[
  (h_{1}(X_i)-h_{1inn}(X_i))_{i\in\Z, n \in \N}, \ \ 
  (h_{1}(X_i)-h_{1}(X_{in}))_{i\in\Z, n \in \N}, \ \ 
	(h_{1}(X_i)-h_{1inn}(X_{in}))_{i\in\Z, n \in \N}
\] 
are row-wise $L_2$-NED and thus $L_1$-NED with approximating constants $a_{k}=O(k^{-3\delta/(1+\delta)})$. From the proof of Lemma 2.23 of \citet{borovkova:burton:dehling:2001}, it follows that
\begin{multline*}
E\left[\Big(\sum_{i=k+1}^m\big( h_1(X_i)-h_{1inn}(X_i)\big)\Big)^{2+\delta}\right]\\
\leq C(m-k)\Big(E\big[\big|h_1(X_i)-h_{1inn}(X_i)\big|^{2+\delta}\big]^{\frac{2}{2+\delta}}+E\big[\big|h_1(X_i)-h_{1inn}(X_i)\big|^{2+\delta}\big]^{\frac{1}{1+\delta}}\Big)\\ \leq Cn^{-\epsilon}(m-k)
\end{multline*}
The other parts of the conditions of Lemma \ref{lem:local1} are shown likewise.

Part 2: 
This can be shown along the lines of the proof of Lemma D.6 (i) in \citet[][supplementary material]{dehling:vogel:wendler:wied:2017}. 
The proof of Lemma \ref{lem:local3} is complete.
\end{proof}

%


\subsection{Proofs of Section \ref{sec:fixed} (Fixed alternatives)}

\begin{proof}[Proof of Theorem \ref{th:fixed}] We prove part (1) concerning $D_{n}^F$. Part (2) for $D_{n}^L$ is proved analogously.
Using the Hoeffding decomposition, we obtain
\begin{multline} 
\label{eq:fixed1}
\frac{1}{n} D_{n}^F([nt]) =
 \frac{[nt]}{n}\frac{2}{[nt]([nt]-1)}\sum_{1\leq i<j\leq [nt]}\theta_{ijn} -     
						\frac{[nt]}{n}\frac{2}{n(n-1)}\sum_{1\leq i<j\leq n}\theta_{ijn}\\
+\frac{[nt]}{n}\frac{2}{[nt]([nt]-1)}\sum_{1\leq i<j\leq [nt]}h_{1ijn}(X_{in}) -
						\frac{[nt]}{n} \frac{2}{n(n-1)}\sum_{1\leq i<j\leq n}h_{1ijn}(X_{in})\\
+\frac{[nt]}{n}\frac{2}{[nt]([nt]-1)}\sum_{1\leq i<j\leq [nt]}h_{1jin}(X_{jn}) - 
						\frac{[nt]}{n}\frac{2}{n(n-1)}\sum_{1\leq i<j\leq n}h_{1jin}(X_{jn})\\
+\frac{[nt]}{n}\frac{2}{[nt]([nt]-1)}\sum_{1\leq i<j\leq [nt]}h_{2ijn}(X_{in},X_{jn}) - 
            \frac{[nt]}{n}\frac{2}{n(n-1)}\sum_{1\leq i<j\leq n}h_{2ijn}(X_{in},X_{jn})
\end{multline}
For the first row of (\ref{eq:fixed1}), note that
\begin{equation*}
\theta_{ijn}=\begin{cases}
\theta_F \ \ &\text{ if } i,j\leq [n\tau^\star], \\
\theta_{FG}=\frac{\theta_F+\theta_G}{2}+\rho_{FG} \ \ &\text{ if } i\leq [n\tau^\star]<j, \\
\theta_G \ \ &\text{ if } i,j> [n\tau^\star].
\end{cases}
\end{equation*}
So for $t\leq \tau^\star$ we have
\begin{multline*}
\frac{[nt]}{n}\frac{2}{[nt]([nt]-1)}\sum_{1\leq i<j\leq [nt]}\theta_{ijn}-\frac{[nt]}{n}\frac{2}{n(n-1)}\sum_{1\leq i<j\leq n}\theta_{ijn}\\
=\frac{[nt]}{n}\theta_F -\frac{[nt]}{n}\frac{2}{n(n-1)}\bigg(\frac{[n\tau^\star]([n\tau^\star]-1)}{2}\theta_F\\
+[n\tau^\star](n-[n\tau^\star])(\theta_F/2+\theta_G/2+\rho_{FG})+\frac{(n-[n\tau^\star])(n-[n\tau^\star]-1)}{2}\theta_G\bigg),
\end{multline*}
which, for $n \to \infty$, converges to
\begin{multline*}
	t\theta_F-\tau^{\star 2}\theta_F
	-t\tau^\star(1-\tau^\star)(\theta_F+\theta_G)-2t\tau^\star(1-\tau^\star)\rho_{FG}-t(1-\tau^\star)^2\theta_G\\
	=t(1-\tau^\star)(\theta_F-\theta_G)-2t\tau^\star(1-\tau^\star)\rho_{FG}=\Psi_1(t).
\end{multline*}
For $t>\tau^\star$, similar calculation also lead to the limit $\Psi_1(t)$. 
The last row of (\ref{eq:fixed1}) converges to 0 uniformly by our Assumption \ref{ass:fixed} (\ref{enum:fixed1}). 
For the two middle rows of (\ref{eq:fixed1}), note that $h_{1ijn}=h_{1i1n}$ for $j<[n\tau^\star]$ and $h_{1ijn}=h_{1inn}$ for $j\geq[n\tau^\star]$. For $[nt]\geq[n\tau^\star]$. 
This leads to
\begin{multline*}
\sum_{1\leq i<j\leq [nt]}h_{1ijn}(X_{in})=\sum_{i=1}^{[n\tau^\star]-1}\left(([n\tau^\star]-i)h_{1i1n}(X_{in})+([nt]-[n\tau^\star])h_{1inn}(X_{in})\right)\\
+\sum_{i=[n\tau^\star]}^{[nt]}\left(([nt]-i)h_{1inn}(X_{in})\right).
\end{multline*}
By the triangle inequality, we have
\begin{multline*}
	\left|\frac{[nt]}{n}\frac{2}{[nt]([nt]-1)}\sum_{1\leq i<j\leq [nt]}h_{1ijn}(X_{in})\right|\\
	\leq \frac{4}{n}\left|\sum_{i=1}^{[n\tau^\star]-1}h_{1i1n}(X_{in})\right|+\frac{4}{n[nt]}\left|\sum_{i=1}^{[n\tau^\star]-1}ih_{1i1n}(X_{in})\right|+\frac{4}{n}\left|\sum_{i=1}^{[n\tau^\star]-1}h_{1inn}(X_{in})\right|\\
	+\frac{4}{n}\left|\sum_{i=1}^{[nt]}h_{1inn}(X_{in})\right| + \frac{4}{n}\left|\sum_{i=1}^{[n\tau^\star]-1}h_{1inn}(X_{in})\right|\\
	+\frac{4}{n[nt]}\left|\sum_{i=1}^{[nt]}ih_{1inn}(X_{in})\right|+\frac{4}{n[nt]}\left|\sum_{i=1}^{[n\tau^\star]-1}ih_{1inn}(X_{in})\right|.
\end{multline*}
We conclude that this summand converges to 0 uniformly in probability, since its second moment converges to 0 
by Assumptions \ref{ass:fixed} (\ref{enum:fixed2}). Treating the other summands in the same way, the statement of the theorem follows.

The proof of part (2) of \ref{th:fixed} is similar and omitted here.
\end{proof}

\begin{proof}[Proof of Lemma \ref{lemma:linear}]
Since, $\rho_{FG} = 0$ particularly for all point mass distributions $F$ and $G$, we have
	$h(x,y) = \frac{1}{2}\left\{ h(x,x) + h(y,y) \right\}$
for all $x, y \in \R^p$, and hence for all $x_1, \ldots, x_n \in \R^p$, 
\[
	U_n = \frac{1}{\binom{n}{2}} \sum_{1\leq i<j\leq n} h(x_i,x_j) 
	= \frac{1}{\binom{n}{2}} \sum_{1\leq i<j\leq n} \frac{1}{2} \left\{ h(x_i,x_i) + h(x_j,x_j) \right\}
\]\[
	= \frac{1}{n(n-1)} \left( \sum_{i=1}^n (n-i) h(x_i,x_i) + \sum_{j=1}^n (j-1) h(x_j,x_j) \right)
	= \frac{1}{n} \sum_{i=1}^n h(x_i, x_i),
\]
which completes the proof.
\end{proof}

%
%
%
%
%
%

\subsection{Conditions for Assumptions \ref{ass:fixed} to hold (Fixed alternatives)}
\label{app:assumptions_fixed}

Analogous to Section \ref{app:assumptions_local} we give specific conditions on the array $(X_{in})_{i \in \Z, n \in \N}$ and the U-statistic kernel $h$ which ensure that Assumption \ref{ass:fixed} is satisfied. Again, we consider two scenarios: independent observations (Lemma \ref{lem:fixed:independent}) and dependent observations (Lemma \ref{lem:fixed:pned}).

Recall the construction of the ``fixed-alternative array'' $(X_{in})_{i \in \Z, n \in \N}$ from the stationary processes $(X_i)_{i\in\Z}$ 
and $(Y_i)_{i \in \Z}$ from Section \ref{sec:fixed}.

\subsection*{Independent sequences}

\begin{lemma} \label{lem:fixed:independent}
Assume that $(X_{in})_{n\in\N, 1\leq i\leq n}$ is row-wise independent and that there exists a constant $M<\infty$ such that for all $n\in\N$, $i,j\in\{1,\ldots,n\}$
\begin{equation*}
E\left[h^2(X_{in},X_{jn})\right]\leq M.
\end{equation*}
Then Assumption \ref{ass:fixed} is satisfied.
\end{lemma}

\begin{proof} 
Assumption \ref{ass:fixed} (\ref{enum:fixed1}) is an immediate consequence of Lemma D.6 (i) in the supplementary material of \citet{dehling:vogel:wendler:wied:2017}.
It remains to show that the altogether eight terms, labelled (\ref{eq:fixed3})--(\ref{eq:fixed6}), of Assumption \ref{ass:fixed} (\ref{enum:fixed2}), are all $o(n^2)$ for $n \to \infty$.
By standard arguments for all $i,j,n$ we have 
$E\big[h^2_{1ijn}(X_{in})\big]\leq M$
and hence by independence
\begin{equation*}
E\left[\Big(\sum_{i=a}^b h_{1ijn}(X_{in})\Big)^2\right]\leq M(b-a).
\end{equation*}
By Theorem 2.4.1 of \citet{stout:1974}, we have that
\[
	E\bigg[\bigg(\max_{l\leq n}\Big|\sum_{i=1}^lh_{1i1n}(X_{in})\Big|\bigg)^2\bigg]=O(n\log^2(n))=o( n^{2}), 
\]
and thus the first term of (\ref{eq:fixed3}) is $o(n^2)$. 
For the second term of (\ref{eq:fixed3}), we define
\[
	S_a=\frac{1}{a}\sum_{i=1}^a ih_{1i1n}(X_{in}).
\]
If $a<b$, then
\begin{multline*}
E\left[\big(S_b-S_a\big)^2\right]\leq 2E\left[\big(\frac{1}{b}\sum_{i=a+1}^bih_{1ijn}(X_{in})\big)^2\right]+2E\left[\big((\frac{1}{a}-\frac{1}{b})\sum_{i=1}^aih_{1ijn}(X_{in})\big)^2\right]\\
\leq 2\frac{1}{b^2}\sum_{i=a+1}^b i^2M+2\Big(\frac{b-a}{ab}\Big)^2\sum_{i=1}^ai^2M\leq 2\frac{1}{b^2}(b-a)b^2M+2\Big(\frac{b-a}{ab}\Big)^2a^3M\leq 4M(b-a).
\end{multline*}
So we can apply Theorem 2.4.1 of \citet{stout:1974} again to obtain that the second term of (\ref{eq:fixed3}) is also $o(n^2)$. The terms in 
(\ref{eq:fixed4}), (\ref{eq:fixed5}), and (\ref{eq:fixed6}) can be treated in a similar way. 
\end{proof}

\subsection*{Dependent sequences}

\begin{lemma} \label{lem:fixed:pned}
If Assumptions \ref{ass:weak_dependence}, \ref{ass:uniform_moments}, \and \ref{ass:variation_condition} hold, 
then Assumption \ref{ass:fixed} is satisfied.
\end{lemma}

\begin{proof} Assumption \ref{ass:fixed} (\ref{enum:fixed1}) is an immediate consequence of Lemma D.6 (ii) in \citet[][supplementary material]{dehling:vogel:wendler:wied:2017}.  

By Assumption \ref{ass:weak_dependence} and Lemma A.2 of \citet{dehling:vogel:wendler:wied:2017}, for all $j$, the triangular array $(h_{1ijn}(X_{in}))_{i=1,\ldots,n}$ is rowwise $L_2$-NED with and thus $L_1$-NED with approximation constants $a_{k}=O(k^{-3\delta/(1+\delta)})$. 
Thus we can apply Lemma 2.23 of \citet{borovkova:burton:dehling:2001} to obtain
\begin{equation*}
		E\left[\Big(\sum_{i=a}^b h_{1ijn}(X_{in})\Big)^2\right]\leq C(b-a)
\end{equation*}
and (\ref{eq:fixed3}) and (\ref{eq:fixed4}) follow in the same way as in Lemma \ref{lem:fixed:independent} above. 

A careful reading of the proof of  Lemma 2.23 of \citet{borovkova:burton:dehling:2001} leads to the stronger version
\begin{equation*}
		E\left[\Big(\sum_{i=a}^b c_i h_{1ijn}(X_{in})\Big)^2\right]\leq C\max\{c_a^2,c_{a+1}^2,\ldots,c_b^2\}(b-a)
\end{equation*}
so that we can apply the same arguments as in Lemma \ref{lem:fixed:independent} above to show (\ref{eq:fixed5}) and (\ref{eq:fixed6}).
\end{proof}


\subsection{Proofs of Section \ref{sec:consistency} (Consistency)}

\begin{proof}[Proof of Corollary \ref{cor:consistency}]
In the case $\rho_{FG}=0$, it is obvious that $\Psi_1$ and $\Psi_2$ take their maximum at $\tau^\star$ with
\begin{equation*}
|\Psi_1(\tau^\star)|=|\Psi_2(\tau^\star)|=\tau^\star(1-\tau^\star)|\theta_F-\theta_G|
\end{equation*}
which is not 0 if $\theta_F\neq\theta_G$. In the case $\rho_{FG}\neq0$, note that for $t>\tau^\star$
\begin{equation*}
\Psi'_1(t)=-\tau^\star(\theta_F-\theta_G)-2\tau^\star(1-\tau^\star)\rho_{FG}+2\frac{1}{t^2}(\tau^\star)^2\rho_{FG}
\end{equation*}
which is not constant. Thus, $\Psi_1$ can not be equal to 0 on the whole interval $[0,1]$. Similarly, for $t>\tau^\star$
\begin{equation*}
\Psi'_2(t)=-\tau^\star(\theta_F-\theta_G)-2\tau^\star\rho_{FG}+2\frac{1}{t^2}(\tau^\star)^2\rho_{FG}
\end{equation*}
is not constant. 

\end{proof}


\subsection{Proofs of Section \ref{sec:location} (Estimation of change location)}

\begin{proof}[Proof of Proposition \ref{prop:max}] We first prove part (1). We only treat the case $\theta_F>\theta_G$. For $t<\tau^\star$, it follows that
\begin{align*}
\Psi_1(\tau^\star)-\Psi_1(t)=&\tau^\star(1-\tau^\star)(\theta_F-\theta_G)-2\tau^\star\tau^\star(1-\tau^\star)\rho_{FG}\\
&-t(1-\tau^\star)(\theta_F-\theta_G)+2t\tau^\star(1-\tau^\star)\rho_{FG}\displaybreak[0]\\
=&(\tau^\star-t)(1-\tau^\star)(\theta_F-\theta_G)-2(\tau^\star-t)\tau^\star(1-\tau^\star)\rho_{FG}\displaybreak[0]\\
>&(\tau^\star-t)(1-\tau^\star)(\theta_F-\theta_G)-(\tau^\star-t)(1-\tau^\star)(\theta_F-\theta_G)=0
\end{align*}
as $\rho_{FG}<\rho_{FG}<\frac{\theta_F-\theta_G}{2(\tau^\star-1+1/\tau^\star)}<(\theta_F-\theta_G)/(2\tau^\star)$. On the other hand
\begin{align*}
\Psi_1(\tau^\star)+\Psi_1(t)=&\tau^\star(1-\tau^\star)(\theta_F-\theta_G)-2\tau^\star\tau^\star(1-\tau^\star)\rho_{FG}\\
&+t(1-\tau^\star)(\theta_F-\theta_G)-2t\tau^\star(1-\tau^\star)\rho_{FG}\displaybreak[0]\\\
=&(\tau^\star+t)(1-\tau^\star)(\theta_F-\theta_G)-2(\tau^\star+t)\tau^\star(1-\tau^\star)\rho_{FG}\displaybreak[0]\\
>&(\tau^\star+t)(1-\tau^\star)(\theta_F-\theta_G)-(\tau^\star+t)(1-\tau^\star)(\theta_F-\theta_G)=0
\end{align*}
as $\rho_{FG}<(\theta_F-\theta_G)/(2\tau^\star)$. We conclude that $\Psi_1(t)<\Psi_1(\tau^\star)$ and $\Psi_1(t)> -\Psi_1(\tau^\star)$ and thus $|\Psi_1(t)|<|\Psi_1(\tau^\star)|$ for $t<\tau^\star$. Now we proceed with the other case $t>\tau^\star$:
\begin{align*}
\Psi_1(\tau^\star)-\Psi_1(t)=&(1-\tau^\star)\tau^\star(\theta_F-\theta_G)+2\frac{\tau^\star-\tau^\star}{\tau^\star}\tau^\star\rho_{FG}-2\tau^\star\tau^\star(1-\tau^\star)\rho_{FG}\\
&-(1-t)\tau^\star(\theta_F-\theta_G)-2\frac{t-\tau^\star}{t}\tau^\star\rho_{FG}+2t\tau^\star(1-\tau^\star)\rho_{FG}\displaybreak[0]\\
=&(t-\tau^\star)\tau^\star(\theta_F-\theta_G)-2\frac{t-\tau^\star}{t}\tau^\star\rho_{FG}+2(t-\tau^\star)\tau^\star(1-\tau^\star)\rho_{FG}\displaybreak[0]\\
=&(t-\tau^\star)\tau^\star(\theta_F-\theta_G)-2\Big(\frac{1}{t}+\tau^\star-1\Big)(t-\tau^\star)\tau^\star\rho_{FG}\\
>&(t-\tau^\star)\tau^\star(\theta_F-\theta_G)-(t-\tau^\star)\tau^\star(\theta_F-\theta_G)=0
\end{align*}
as $\rho_{FG}<\frac{\theta_F-\theta_G}{2(\tau^\star-1+1/\tau^\star)}$. Finally
\begin{align*}
\Psi_1(t)=&(1-t)\tau^\star(\theta_F-\theta_G)+2\frac{t-\tau^\star}{t}\tau^\star\rho_{FG}-2t\tau^\star(1-\tau^\star)\rho_{FG}\\
>&(1-t)\tau^\star(\theta_F-\theta_G)-\frac{t-\tau^\star}{t}\tau^\star(\theta_F-\theta_G)+t(1-\tau^\star)\tau^\star(\theta_F-\theta_G)\\
=&\tau^\star(\theta_F-\theta_G)\left(\frac{\tau^\star}{t}-\tau^\star t\right)\geq0
\end{align*}
as $\rho_{FG}>-(\theta_F-\theta_G)/2$, because $(t-\tau^\star)/t\leq t(1-\tau^\star)$. We concluce that $0\leq \Psi_1(t)<\Psi_1(\tau^\star)$ and thus $|\Psi_1(t)|<|\Psi_1(\tau^\star)|$ for $t>\tau^\star$. The case $\theta_F<\theta_G$ is similar and hence omitted.

Now we treat part (2) of the proposition. As before, we only treat the case $\theta_F>\theta_G$. By our assumption, $\rho_{FG}\geq-\frac{1}{2}\frac{1-\tau^\star}{\tau^\star}(\theta_F-\theta_G)$. For $t\in[0,\tau^\star)$, it follows that
\begin{align*}
\Psi_2(\tau^\star)-\Psi_2(t)=&\tau^\star(1-\tau^\star)(\theta_F-\theta_G)-2\frac{\tau^\star(\tau^\star-\tau^\star)}{1-\tau^\star}(1-\tau^\star)\rho_{FG}\\
&-t(1-\tau^\star)(\theta_F-\theta_G)+2\frac{t(\tau^\star-t)}{1-t}(1-\tau^\star)\rho_{FG}\displaybreak[0]\\
=&(\tau^\star-t)(1-\tau^\star)(\theta_F-\theta_G)+2\frac{\tau^\star-t}{1-t}t(1-\tau^\star)\rho_{FG}\displaybreak[0]\\
\geq&(\tau^\star-t)(1-\tau^\star)(\theta_F-\theta_G)-(\tau^\star-t)(1-\tau^\star)\frac{t}{1-t}\frac{1-\tau^\star}{\tau^\star}(\theta_F-\theta_G)\\
>&(\tau^\star-t)(1-\tau^\star)(\theta_F-\theta_G)-(\tau^\star-t)(1-\tau^\star)(\theta_F-\theta_G)=0.
\end{align*}
So $\Psi_2(t)<\Psi_2(\tau^\star)$. On the other hand, $\rho_{FG}<\frac{1}{2}(\theta_F-\theta_G)$ and consequently
\begin{align*}
\Psi_2(\tau^\star)+\Psi_2(t)=&\tau^\star(1-\tau^\star)(\theta_F-\theta_G)-2\frac{\tau^\star(\tau^\star-\tau^\star)}{1-\tau^\star}(1-\tau^\star)\rho_{FG}\\
&+t(1-\tau^\star)(\theta_F-\theta_G)-2\frac{t(\tau^\star-t)}{1-t}(1-\tau^\star)\rho_{FG}\displaybreak[0]\\
=&(\tau^\star+t)(1-\tau^\star)(\theta_F-\theta_G)-2\frac{\tau^\star-t}{1-t}t(1-\tau^\star)\rho_{FG}\displaybreak[0]\\
\geq&(\tau^\star+t)(1-\tau^\star)(\theta_F-\theta_G)-\frac{\tau^\star-t}{1-t}t(1-\tau^\star)(\theta_F-\theta_G)\\
>&(\tau^\star+t)(1-\tau^\star)(\theta_F-\theta_G)-(\tau^\star+t)(1-\tau^\star)(\theta_F-\theta_G)=0.
\end{align*}
So we have that $\Psi_2(t)>-\Psi_2(\tau^\star)$ and thus $|\Psi_2(t)|<|\Psi_2(\tau^\star)|$ for $t<\tau^\star$. For $t>\tau^\star$, note that by our assumption $\rho_{FG}\leq\frac{1}{2}\frac{\tau^\star}{1-\tau^\star}(\theta_F-\theta_G)$
\begin{align*}
\Psi_2(\tau^\star)-\Psi_2(t)=&(1-\tau^\star)\tau^\star(\theta_F-\theta_G)+2\frac{(\tau^\star-\tau^\star)}{\tau^\star}\tau^\star\rho_{FG}-2(\tau^\star-\tau^\star)\tau^\star\rho_{FG}\\
&-(1-t)\tau^\star(\theta_F-\theta_G)-2\frac{(t-\tau^\star)}{t}\tau^\star\rho_{FG}+2(t-\tau^\star)\tau^\star\rho_{FG}\displaybreak[0]\\
=&(t-\tau^\star)\tau^\star(\theta_F-\theta_G)-2\frac{1-t}{t}(t-\tau^\star)\tau^\star\rho_{FG}\displaybreak[0]\\
\geq&(t-\tau^\star)\tau^\star(\theta_F-\theta_G)-(t-\tau^\star)\tau^\star\frac{1-t}{t}\frac{\tau^\star}{1-\tau^\star}(\theta_F-\theta_G)\\
>&(t-\tau^\star)\tau^\star(\theta_F-\theta_G)-(t-\tau^\star)\tau^\star(\theta_F-\theta_G)>0.
\end{align*}
Finally, $\Psi_2(t)>-\Psi_2(\tau^\star)$ for $t<\tau^\star$ can be shown in the same way, which completes the proof.
\end{proof}


\subsection{Proofs of Section \ref{sec:asym.dist} (Asymptotic distributions under fixed alternatives)}


Theorem \ref{th:fixed} shows the convergence of the processes $\left( \frac{1}{n} D_n^F([nt])\right)_{0\le t\le 1}$ and $\left( \frac{1}{n} D_n^L([nt])\right)_{0\le t\le 1}$ to the deterministic limit processes $\Psi_1$ and $\Psi_2$, respectively. In this section we study the asymptotic distribution of the processes. In Proposition \ref{prop:procfixed} below we show that
\[
  \sqrt{n}  \big( \frac{1}{n} D_n^F([nt]) - \Psi_1(t) \big)_{0\leq t\leq 1} 
	\ \mbox{ and } \ 
	\sqrt{n}  \big( \frac{1}{n} D_n^L([nt]) - \Psi_2(t) \big)_{0\leq t\leq 1} 
\]
both converge in distribution to centered Gaussian limit processes, which lays the foundation for Theorem \ref{th:limitfixed}.
We consider the fixed alternative, where the random variables $X_{in},\, 1\leq i\leq [n\tau^\ast]$, have distribution $F$,
and the random variables $X_i=X_{in},\, [n\tau^\ast]+1\leq i\leq n$, have distribution $G$. 
More specifically, we assume that $(X_i,Y_i)$ is an i.i.d. process, where the $X_i$'s have distribution $F$,  the $Y_i$'s have distribution $G$, and such that 
\[
X_{in}= \begin{cases}
     X_i & \mbox{ for } 1\leq i \leq k^\ast, \\[1mm]
     Y_i & \mbox{ for } k^\ast +1 \leq i\leq n,
  \end{cases}
\]
where $k^\ast$ is short for $[n\tau^\ast]$.
The main tool in the analysis of the asymptotic distribution of our test statistics will be the Hoeffding decomposition of the kernel $h(x,y)$. 
Under the fixed alternative, we have to deal with three different Hoeffding decompositions of $h(X_{in},X_{jn})$, depending on whether the random variables $X_{in}$ and $X_{jn}$ have distribution $F$ or $G$, respectively.  We define the functions 
\[
	h_F(x) =\int h(x,y)\, dF(y), 
	\qquad
h_G(x)  = \int h(x,y)\,  dG(y)
\]
and obtain, 
e.g., for $1\leq i\leq k_n^\ast$, $k_n^\ast+1\leq j\leq n$, 
\[
  h(X_i,X_j)=\theta_{FG} + (h_G(X_i) -\theta_{FG})  + (h_F(X_j)-\theta_{FG}) +h_{2,F,G}(X_i,X_j), 
 \]
 The other two Hoeffding decompositions, valid for $1\leq i<  j\leq k_n^\ast$, and $k_n+1\leq i<j \leq n$, respectively, are given by 
 \begin{align*}
 h(X_i,X_j)&= \theta_F +(h_F(X_i)-\theta_F)+(h_F(X_j)-\theta_F) + h_{F,F}(X_i,X_j)\\
  h(X_i,X_j)&= \theta_G +(h_G(X_i)-\theta_G)+(h_G(X_j)-\theta_G) + h_{G,G}(X_i,X_j).
 \end{align*}
 In each of these decompositions, the 2nd order terms are defined by the above identities. By definition, the first order terms have mean zero, and the 2nd order terms are degenerate in the sense that, e.g., $E(h_{F,G}(X_i,X_j)|X_i)=0$ for $1\leq i\leq k_n^\ast$, $k_n^\ast+1\leq j\leq n$.

\begin{proposition}
\label{prop:procfixed}
Under the condition of Lemma \ref{lem:fixed:independent}, we have
\begin{align}
 \frac{1}{\sqrt{n}} \big(D_n^F([nt]) - n\Psi_1(t) \big)_{0\leq t\leq 1}  & \to (G_1(t))_{0\leq t\leq 1},\\
  \frac{1}{\sqrt{n}} \big(D_n^F([nt]) - n\Psi_2(t) \big)_{0\leq t\leq 1}  & \to (G_2(t))_{0\leq t\leq 1}
\end{align}
in distribution in the Skorokhod space $D[0,1]$, 
where $G_1$ and $G_2$ are mean-zero Gaussian processes with the following representations
\begin{align}
G_1(t) &=  \left\{  
\begin{array}{ll} 
 2(1-t\, \tau^\ast) W_F^{(1)}(t)  - 2t\, \tau^\ast (W_F^{(1)} (\tau^\ast) -W_F^{(1)}(t)) -2t\, \tau^\ast (W_F^{(2)} (1)-W_F^{(2)}(\tau^\ast))  &\\
 -2t(1-\tau^\ast) W_G^{(1)}(\tau^\ast) -2t(1-\tau^\ast)(W_G^{(2)}(1)-W_G^{(2)}(\tau^\ast)), & t\leq \tau^\ast \\[2mm]
2 \tau^\ast(\frac{1}{t}-t) W_F^{(1)}(\tau^\ast) +2 \big(\frac{t-\tau^\ast}{t} -t(1-\tau^\ast)  \big) W_G^{(1)}(\tau^\ast) &
 \\
  + 2\tau^\ast(\frac{1}{t}-t) (W_F^{(2)} (t) -W_F^{(2)}(\tau^\ast))  -2t\tau^\ast (W_F^{(2)}(1) -W_F^{(2)}(t)) &\\
  + 2\big( \frac{t-\tau^\ast}{t} -t(1-\tau^\ast)  \big) (W_G^{(2)}(t)-W_G^{(2)}(\tau^\ast)) -2t(1-\tau^\ast) (W_G^{(2)}(1)-W_G^{(2)}(t)) 
   , &  t\geq \tau^\ast  
\end{array}
\right. \\
G_2(t)&=  \left\{  
\begin{array}{ll} 
 2(1-t) W_F^{(1)}(t) -2t\frac{\tau^\ast -t}{1-t}  (W_F^{(1)}(\tau^\ast)-W_F^{(1)}(t)) -2t\frac{\tau^\ast-t}{1-t} (W_F^{(2)}(1)-W_F^{(2)}(\tau^\ast))& \\
  - 2\frac{t(1-\tau^\ast)}{1-t} (W_G^{(1)}(\tau^\ast)-W_G^{(1)}(t)) -2 t \frac{1-\tau^\ast}{1-t} (W_G^{(2)}(1)-W_G^{(2)}(\tau^\ast)),
&  t\leq \tau^\ast  \\[2mm]
  2\frac{(1-t)\tau^\ast}{t} W_F^{(1)}(\tau^\ast) +2 \frac{(1-t)(t-\tau^\ast) }{t} W_G^{(1)}(\tau^\ast) +2\frac{(1-t)\tau^\ast}{t} (W_F^{(2)}(t)-W_F^{(2)}(\tau^\ast) ) \\
  +2\frac{(t-\tau^\ast)(1-t)}{t} (W_G^{(2)}(t) -W_G^{(2)}(\tau^\ast))  -2t (W_G^{(2)}(1)-W_G^{(2)}(t)),  &  t\geq \tau^\ast,
\end{array}
\right. 
\end{align}
and where $(W_F^{(1)}(t), W_G^{(1)}(t),W_F^{(2)}(t), W_G^{(2)}(t) ) $ is a four-dimensional Brownian motion with mean zero and the same covariance structure as $(h_F(X),h_G(X),h_F(Y),h_G(Y))$.
\end{proposition}

%
%
%
%

\begin{proof}
We will first analyze $U_{1:k}$ with the help of the Hoeffding decomposition. Recall $k_n^\ast =[n\tau^\ast]$. For $k\leq k_n^\ast$ we obtain
\begin{align*}
U_{1:k} &= \frac{1}{\binom{k}{2}} \sum_{1\leq i<j\leq k} h(X_{in},X_{jn}) \\
&= \frac{2}{k(k-1)} \sum_{1\leq i<j \leq k} \big\{\theta_F+(h_F(X_i)-\theta_F) +(h_F(X_j) -\theta_F) +h_{FF}(X_i,X_j) \big\} \\
&=\theta_F +\frac{2}{k} \sum_{i=1}^k (h_F(X_i)-\theta_F) +\frac{2}{k(k-1)} \sum_{1\leq i<j\leq k} h_{FF}(X_i,X_j).
\end{align*}
For $k_n^\ast \leq k \leq n$ we obtain the decomposition
\begin{align*}
 U_{1:k} &= \frac{1}{\binom{k}{2}} \sum_{1\leq i<j\leq k} h(X_{in},X_{jn}) \\[1mm]
 &=\frac{2}{k(k-1)} \sum_{1\leq i<j \leq k_n^\ast} 
 \big\{  \theta_F+(h_F(X_i)-\theta_F) +(h_F(X_j) -\theta_F) +h_{FF}(X_i,X_j)   \big\} \\
 & \quad + \frac{2}{k(k-1)} \sum_{i=1}^{k_n^\ast} \sum_{j=k_n^\ast+1}^k  \big\{ \theta_{FG} +(h_G(X_i)-\theta_{FG}) 
 +(h_F(Y_j) - \theta_{FG}) +h_{FG}(X_i,Y_j)  \big\} \\
 & \quad + \frac{2}{k(k-1)} \sum_{k_n^\ast + 1 \leq i<j \leq k} \big\{ \theta_G +(h_G(Y_i) -\theta_G) +(h_G(Y_j)-\theta_G) +
 h_{GG}(Y_i,Y_j)  \big\} \\[1mm]
 &= \frac{k_n^\ast (k_n^\ast-1)}{k(k-1)} \theta_F + 2\frac{k_n^\ast (k-k_n^\ast)}{k(k-1)} \theta_{FG} + 
 \frac{(k-k_n^\ast)(k-k_n^\ast -1)}{k(k-1)} \theta_G\\
 &\quad + \frac{2 (k_n^\ast -1)}{k(k-1)}  \sum_{i=1}^{k_n^\ast} (h_F(X_i)-\theta_F) 
 + \frac{2 (k-k_n^\ast)}{k(k-1)}  \sum_{i=1}^{k_n^\ast} (h_G(X_i) -\theta_{FG} ) \\
 &\quad + \frac{2 k_n^\ast }{k(k-1)} \sum_{j=k_n^\ast+1}^k (h_F(Y_j) -\theta_{FG})
 + \frac{2(k-k_n^\ast-1)}{k(k-1)} \sum_{j=k_n^\ast +1}^k (h_G(Y_j)-\theta_G) \\
 &\quad + \frac{2}{k(k-1)} 
 \big\{  
   \sum_{1\leq i<j\leq k_n^\ast} h_{FF}(X_i,X_j) +\sum_{i=1}^{k_n^\ast} \sum_{j=k_n^\ast+1}^k h_{FG} (X_i,Y_j) 
   +\sum_{k_n^\ast+1 \leq i<j\leq k} h_{GG}(Y_i,Y_j)
 \big\}.
\end{align*}
Using the identity $\theta_{FG}=\rho_{FG} +\frac{1}{2}\theta_F +\frac{1}{2}\theta_G$, we can rewrite the constant term in the above decomposition as 
\begin{align*}
& \frac{k_n^\ast (k_n^\ast-1)}{k(k-1)} \theta_F + 2\frac{k_n^\ast (k-k_n^\ast)}{k(k-1)} \theta_{FG} + 
 \frac{(k-k_n^\ast)(k-k_n^\ast -1)}{k(k-1)} \theta_G\\[1mm]
 &\quad 
= \Big( \frac{k_n^\ast (k_n^\ast-1)}{k(k-1)} +  \frac{k_n^\ast (k-k_n^\ast)}{k(k-1)}  \Big) \theta_F + 2\frac{k_n^\ast (k-k_n^\ast)}{k(k-1)} \rho_{FG} + 
 \frac{(k-k_n^\ast)(k-k_n^\ast -1)}{k(k-1)} \theta_G\\
 &\qquad +  \Big( \frac{(k-k_n^\ast)(k-k_n^\ast -1)}{k(k-1)} + \frac{k_n^\ast (k-k_n^\ast)}{k(k-1) } \Big) \theta_G \\[1mm]
 & \quad = \frac{k_n^\ast}{k} \theta_F  + 2\frac{k_n^\ast (k-k_n^\ast)}{k(k-1)} \rho_{FG} 
 +\frac{k-k_n^\ast}{k} \theta_G.
\end{align*}
Thus, we obtain for $k\leq k_n^\ast$ the following decomposition of $D_n^F(k)=k(U_{1:k}-U_{1:n})$ into a constant term, a linear term, and a second order term:
\begin{align*}
 k(U_{1:k}-U_{1:n}) &=\frac{k\, (n-k_n^\ast)}{n}  (\theta_F-\theta_G) -  2\frac{k\, k_n^\ast\, (n-k_n^\ast)}{n(n-1)} \rho_{FG} 
 + 2 \sum_{i=1}^k (h_F(X_i)-\theta_F) \\
 &- \frac{2\, k\,  (k_n^\ast -1)  }{n(n-1)}\sum_{i=1}^{k_n^\ast} (h_F(X_i)-\theta_F) 
 - \frac{2\, k \, (n-k_n^\ast)}{n(n-1)}  \sum_{i=1}^{k_n^\ast} (h_G(X_i) -\theta_{FG} ) \\
 &\quad - \frac{2\, k \, k_n^\ast}{n(n-1)} \sum_{j=k_n^\ast+1}^n (h_F(Y_j) -\theta_{FG})
 - \frac{2\, k\, (n-k_n^\ast-1)}{n(n-1)} \sum_{j=k_n^\ast +1}^n (h_G(Y_j)-\theta_G) \\
 & \quad + \frac{2}{k-1} \sum_{1\leq i<j\leq k} h_{FF}(X_i,X_j)  - 
 \frac{2\, k}{n(n-1)} 
   \sum_{1\leq i<j\leq k_n^\ast} h_{FF}(X_i,X_j) 
   \\
   &\quad  -   \frac{2\, k}{n(n-1)} \sum_{i=1}^{k_n^\ast} \sum_{j=k_n^\ast+1}^n h_{FG} (X_i,Y_j) 
    -   \frac{2\, k}{n(n-1)}  \sum_{k_n^\ast+1 \leq i<j\leq n} h_{GG}(Y_i,Y_j).
\end{align*}
The second order terms on the right hand side are $O_P(1)$, see Lemma \ref{lem:fixed:independent}. 
%
%
Hence, we get for $t\leq \tau^\ast$:
\begin{align*}
& \frac{1}{\sqrt{n}} \Big[[nt]\, (U_{1:[nt]}-U_{1:n})-n\, \Psi_1(t)  \Big] \\
 &\; = \frac{1}{\sqrt{n}}  \Big(\frac{[nt](n-[n\tau^\ast]) }{n} - n\, t(1-\tau^\ast) \Big) (\theta_F-\theta_G) -\frac{2}{\sqrt{n}}\Big(\frac{[nt]\, [n\tau^\ast] (n-[n\tau^\ast])}{n(n-1)} -  nt\tau^\ast(1-\tau^\ast)\Big) \rho_{FG} \\
& \quad+ \frac{2}{\sqrt{n}} \sum_{i=1}^{[nt]} (h_F(X_i)-\theta_F)
 -2\frac{[nt]([n\tau^\ast]-1)}{n(n-1)} \frac{1}{\sqrt{n}} \sum_{i=1}^{[n\tau^\ast]} (h_F(X_i)-\theta_F) \\
 &\quad -2\frac{[nt](n-[n\tau^\ast])}{n(n-1)} \frac{1}{\sqrt{n}}\sum_{i=1}^{[n\tau^\ast]} (h_G(X_i)-\theta_{FG})
 -2\frac{[nt][n\tau^\ast] }{n(n-1)} \sum_{j=[n\tau^\ast]+1}^n (h_F(Y_j)-\theta_{FG})\\
 &\quad  -  \frac{2\, [nt]\, (n-[n\, \tau^\ast]-1)}{n(n-1)} \sum_{j=[n\, \tau^\ast] +1}^n (h_G(Y_j)-\theta_G) +O_P(1/\sqrt{n})\\
 &\;= 2(1-t\, \tau^\ast) \frac{1}{\sqrt{n}} \sum_{i=1}^{[nt]} (h_F(X_i)-\theta_F) -2\, t\, \tau^\ast \frac{1}{\sqrt{n}} \sum_{i=[nt]}^{[n\tau^\ast]} (h_F(X_i)-\theta_F) \\
 &\quad - 2\, t (1-\tau^\ast) \frac{1}{\sqrt{n}} (\sum_{i=1}^{[n\tau^\ast]} h_G(X_i) -\theta_{FG}) - 2\,t\,\tau^\ast
 \frac{1}{\sqrt{n}} \sum_{i=[n\tau^\ast]}^n (h_F(Y_i)-\theta_{FG}) \\
 &\quad - 2\, t(1-\tau^\ast)\frac{1}{\sqrt{n}} \sum_{i=[n\, \tau^\ast]+1}^n (h_G(Y_i)-\theta_G) +o_P(1),
 \end{align*}
To the sums on the right hand side  we apply the functional central limit theorem for the $\R^4$-valued partial sum process
 \[
  \frac{1}{\sqrt{n}} \sum_{i=1}^{[nt]} 
  \big( h_F(X_i)-\theta_F,h_G(X_i)-\theta_{FG},h_F(Y_i)-\theta_{FG}, h_G(Y_i)-\theta_G  \big),
 \]
 whose limit process we denote by $(W_F^{(1)}(t), W_G^{(1)}(t), W_F^{(2)}(t), W_G^{(2)}(t))_{0\leq t\leq 1}$.  The continuous mapping theorem yields convergence of the above process to the limit stated in the formulation of the theorem, for $t\leq \tau^\ast$. 
 \\[1mm]
 For $k\geq k_n^\ast$, we obtain
 \begin{align*}
&  k(U_{1:k}-U_{1:n}) = k_n^\ast\frac{n-k}{n}(\theta_F-\theta_G) + 2 \Big[ \frac{k_n^\ast(k-k_n^\ast)}{k-1} - k\frac{k_n^\ast(n-k_n^\ast)}{n(n-1)}  \Big] \rho_{FG} \\
 &\; + 2\Big[\frac{k_n^\ast-1}{k-1} -\frac{k(k_n^\ast-1)}{n(n-1)}   \Big] \sum_{i=1}^{k_n^\ast} (h_F(X_i) -\theta_F) 
 + 2\Big[ \frac{k-k_n^\ast}{k-1} - \frac{k(n-k_n^\ast)}{n(n-1)} \Big] \sum_{i=1}^{k_n^\ast} (h_G(X_i)-\theta_{FG}) \\
 &\; +2\Big[ \frac{k_n^\ast}{k-1} -\frac{k\, k_n^\ast}{n(n-1)}  \Big] \sum_{i=k_n^\ast+1}^k (h_F(Y_i)-\theta_{FG} )
 -2 \frac{k\, k_n^\ast}{n(n-1)} \sum_{i=k+1}^n (h_F(Y_i)-\theta_{FG}) \\
 & \; +2 \Big[ \frac{k-k_n^\ast -1}{k-1} -\frac{k(n-k_n^\ast -1)}{n(n-1)}  \Big] \sum_{i=k_n^\ast+1}^ k (h_G(Y_i)-\theta_G)
 - 2 \frac{k(n-k_n^\ast-1) }{n(n-1) }  \sum_{i=k+1}^n (h_G(Y_i)-\theta_G) \\
 &\; +O_P(1),
 \end{align*}
 and thus
 \begin{align*}
& \frac{1}{\sqrt{n}} \Big( [nt] (U_{1:[nt]}   -U_{1:n}) -n\, \Psi_1(t)  \Big) \\
&\; =  \frac{1}{\sqrt{n}} \Big( [n\tau^\ast] \frac{n-[nt]}{n}(\theta_F-\theta_G) + 2 \Big[ \frac{[n\tau^\ast]([nt]-[n \tau^\ast])}{[nt]-1} - [nt]\frac{[n\tau^\ast](n-[n\tau^\ast])}{n(n-1)}  \Big] \rho_{FG} \\
&\quad - n (1-t)\tau^\ast (\theta_F-\theta_G) - 2\Big[\frac{t-\tau^\ast}{t} \tau^\ast -t \tau^\ast (1-\tau^\ast) \Big] \rho_{FG} \Big) \\
&\quad + 2\Big[\frac{[n\tau^\ast]-1}{[nt]-1} -\frac{[nt]([n\tau^\ast]-1)}{n(n-1)}  \Big] \frac{1}{\sqrt{n}} \sum_{i=1}^{[n\tau^\ast]} (h_F(X_i) -\theta_F)  \\
&\quad + 2\Big[ \frac{[nt]-[n\tau^\ast]}{[nt]-1} - \frac{[nt](n-n\tau^\ast])}{n(n-1)}  \Big]  \frac{1}{\sqrt{n}} \sum_{i=1}^{[n\tau^\ast]} (h_G(X_i)-\theta_{FG})  \\
 &\quad +2\Big[ \frac{[n\tau^\ast]}{[nt]-1} -\frac{[nt]\, [n\tau^\ast]}{n(n-1)}  \Big] \frac{1}{\sqrt{n}} \sum_{i=[n\tau^\ast]+1}^{[nt]} (h_F(Y_i)-\theta_{FG} ) 
 -2 \frac{[nt]\, [n\tau^\ast]}{n(n-1)} \frac{1}{\sqrt{n}} \sum_{i=[nt]+1}^n (h_F(Y_i)-\theta_{FG}) \\
 & \quad +2 \Big[ \frac{[nt]-[n\tau^\ast] -1}{[nt]-1} -\frac{[nt](n-[n\tau^\ast] -1)}{n(n-1)}  \Big] \frac{1}{\sqrt{n}} \sum_{i=[n\tau^\ast]+1}^{[nt]} (h_G(Y_i)-\theta_G) \\
&\quad  - 2 \frac{[nt](n-[n\tau^\ast]-1) }{n(n-1) } \frac{1}{\sqrt{n}}  \sum_{i=[nt]+1}^n (h_G(Y_i)-\theta_G) \quad  +O_P(\frac{1}{\sqrt{n}} ).
 \end{align*} 
The same argument as above yields convergence of this process to the limit stated in the formulation of the theorem, for $t\geq \tau^\ast$. As the limit process is continuous in $t=\tau^\ast$, we finally obtain convergence of the process 
$ \frac{1}{\sqrt{n}} ( [nt] (U_{1:[nt]}   -U_{1:n}) -n\, \Psi_1(t) )_{0\leq t\leq 1}$. 
\\[1mm]
In order to determine the asymptotic distribution of $D_n^L(k)=\frac{k(n-k)}{n}(U_{1:k}-U_{k+1:n})$, we need to calculate the U-statistic $U_{k+1:n}$. For $k\leq k_n^\ast$, we obtain
\begin{align*}
 U_{k+1:n} &= \frac{2}{(n-k)(n-k-1) } \sum_{k+1\leq i<j\leq n} h(X_{in},X_{jn}) \\
 & =  \frac{2}{(n-k)(n-k-1) } \sum_{k+1\leq i<j\leq k_n^\ast} \big[\theta_F+(h_F(X_i)-\theta_F)+(h_F(X_j)-\theta_F)+h_{FF}(X_i,X_j)    \big] \\
 & \quad +  \frac{2}{(n-k)(n-k-1) } \sum_{i=k+1}^{k_n^\ast} \sum_{j=k_n^\ast}^n \big[\theta_{FG} +(h_G(X_i) -\theta_{FG}) +(h_F(Y_j)-\theta_{FG}) +h_{FG}(X_i,Y_j)  \big] \\
  & \quad +  \frac{2}{(n-k)(n-k-1) } \sum_{k_n^\ast +1 \leq i<j \leq n}
  \big[\theta_G + (h_G(Y_i) -\theta_G) +(h_G(Y_j)-\theta_G) +h_{GG}(Y_i,Y_j) \big]\\
& = \frac{(k_n^\ast -k) (k_n^\ast -k-1)}{(n-k)(n-k-1)}\theta_F +2 \frac{(k_n^\ast -k)(n-k_n^\ast)}{(n-k)(n-k-1)} \theta_{FG} 
+ \frac{(n-k_n^\ast)(n-k_n^\ast-1)}{(n-k)(n-k-1)} \\
&\; + 2\frac{k_n^\ast-k-1}{(n-k)(n-k-1)} \sum_{i=k+1}^{k_n^\ast} (h_F(X_i)-\theta_F) \\
&\;+ 2\frac{n-k_n^\ast}{(n-k)(n-k-1)} \sum_{i=k+1}^{k_n^\ast} (h_G(X_i)-\theta_{FG}) \\
&\; + 2\frac{k_n^\ast -k}{(n-k)(n-k-1)} \sum_{i=k_n^\ast +1}^n (h_F(Y_i) -\theta_{FG}) \\
&\; + 2\frac{n-k_n^\ast-1}{(n-k)(n-k-1)} \sum_{i=k_n^\ast+1}^n (h_G(Y_i)-\theta_G)\\
&\; + \frac{2}{n-k)(n-k-1)} \Big[ \sum_{k+1 \leq i<j\leq k_n^\ast} h_{FF}(X_i,X_j) +\sum_{i=k+1}^{k_n^\ast} \sum_{j=k_n^\ast+1}^n  h_{FG}(X_i,Y_j)  \\
&\qquad +\sum_{k_n^\ast+1 \leq i<j\leq n} h_{GG}(Y_i,Y_j) \Big],
\end{align*}
while for $k\geq k_n^\ast$ we have
\[
 U_{k+1:n} 
 =\theta_G+\frac{2}{n-k} \sum_{i=k+1}^n (h_G(Y_i) -\theta_G) +\frac{2}{(n-k)(n-k-1)} 
 \sum_{k+1\leq i<j \leq n} h_{GG} (Y_i,Y_j).
\]

Thus, we obtain for $k\leq k_n^\ast$ 
\begin{align*}
& \frac{k(n-k)}{n} (U_{1:k}-U_{k+1:n}) \\
&\; = \frac{k(n-k_n^\ast)}{n} (\theta_F-\theta_G) -2 \frac{k(k_n^\ast-k)(n-k_n^\ast)}{n(n-k-1)} \rho_{FG} +2\frac{(n-k)}{n} \sum_{i=1}^k (h_F(X_i)-\theta_F) \\
& \quad - 2\frac{k(k_n^\ast-k-1)}{n(n-k-1)} \sum_{i=k+1}^{k_n^\ast} (h_F(X_i)-\theta_F) 
-2\frac{k(n-k_n^\ast)}{n(n-k-1)} \sum_{i=k+1}^{k_n^\ast} (h_G(X_i)-\theta_{FG}) \\
&\quad -2 \frac{k(k_n^\ast-k)}{n(n-k-1)} \sum_{i=k_n^\ast+1}^n (h_F(Y_i)-\theta_{FG}) 
-2\frac{k(n-k_n^\ast-1)}{n)n-k-1)} \sum_{i=k_n^\ast+1}^n (h_G(Y_i) -\theta_G) +O_P(1),
\end{align*}
and hence
\begin{align*}
 & \frac{1}{\sqrt{n}}\Big[ \frac{[nt](n-[nt]) }{n } \Big( U_{1:[nt]} -U_{[nt]+1:n} \Big)   -n\Psi_2(t)\Big] \\
 & \; = \frac{1}{\sqrt{n}} \Big[ \frac{[nt](n-[nt])}{n} -nt(1-\tau^\ast)  \Big](\theta_F-\theta_G) \\
 &\quad -\frac{2}{\sqrt{n}} \Big[\frac{[nt]([n\tau^\ast]-[nt])(n-[n\tau^\ast])}{ n (n-[nt]-1)} - \frac{t(\tau^\ast-t)}{1-t }(1-\tau^\ast)   \Big] \rho_{FG} \\
 &\quad+2\frac{n-[nt]}{n}\frac{1}{\sqrt{n}} \sum_{i=1}^{[nt]} (h_F(X_i)-\theta_F) \\
 &\quad -2\frac{[nt]([n\tau^\ast]-[nt] -1)}{n(n-[nt]-1)} \frac{1}{\sqrt{n}} \sum_{i=[nt]+1}^{[n\tau^\ast]} (h_F(X_i)-\theta_F) \\
 &\quad -2 \frac{[nt](n-[n\tau^\ast])}{n(n-[nt]-1)}\frac{1}{\sqrt{n}} \sum_{i=[nt]+1}^{[n\tau^\ast]} (h_G(X_i)-\theta_{FG}) \\
 &\quad - 2\frac{[nt]([n\tau^\ast]-[nt])}{n(n-[nt]-1)} \frac{1}{\sqrt{n}} \sum_{i=[n\tau^\ast]+1} (h_F(Y_i)-\theta_{FG}) \\
 &\quad -2\frac{[nt](n-[n\tau^\ast]-1)}{n(n-[nt]-1)} \frac{1}{\sqrt{n}} \sum_{i=[n\tau^\ast] +1}^n (h_G(Y_i)-\theta_G)
\end{align*}
The same argument as above yields convergence of this process to the limit stated in the formulation of the theorem, for $t\leq \tau^\ast$. 
\\[1mm]
For $k\geq k_n^\ast+1$, we obtain
\begin{align*}
& \frac{k(n-k)}{n} (U_{1:k}-U_{k+1:n}) \\
&\; =  \frac{k_n^\ast (n-k)}{n } (\theta_F-\theta_G) + 2\frac{k_n^\ast (k-k_n^\ast) (n-k) }{n(k-1)} \rho_{FG}
  + \frac{2(n-k) (k_n^\ast -1)}{n(k-1)} \sum_{i=1}^{k_n^\ast} (h_F(X_i)-\theta_F) \\
 &\quad + \frac{2(n-k) (k-k_n^\ast) }{n(k-1)} \sum_{i=1}^{k_n^\ast} (h_G(X_i) -\theta_{FG} ) 
  + \frac{2(n-k) k_n^\ast }{n(k-1)} \sum_{j=k_n^\ast+1}^k (h_F(Y_j) -\theta_{FG})\\
 &\quad+ \frac{2(k-k_n^\ast-1)(n-k)}{n(k-1)} \sum_{j=k_n^\ast +1}^k (h_G(Y_j)-\theta_G) 
 - \frac{2k}{n}  \sum_{i=k+1}^n (h_G(Y_i)-\theta_G) +O_P(1),
\end{align*}
and hence for $t\geq \tau^\ast$
\begin{align*}
& \frac{1}{\sqrt{n}}\Big[ \frac{[nt](n-[nt])}{n}  \big(U_{1:[nt]} -U_{[nt]+1:n} \big) -n\Psi_2(t)  \Big] \\
 &\; = \frac{1}{\sqrt{n}} \Big[ \frac{[n\tau^\ast] (n-[nt])}{n }   - n (1-t)\tau^\ast \Big] (\theta_F-\theta_G) \\
 & + 2 \Big[ \frac{[n\tau^\ast] ([nt]-[n\tau^\ast) (n-[nt]) }{n([nt]-1)}  - n(t-\tau^\ast)\tau^\ast \big(\frac{1}{t}-1\big) 
 \Big] \rho_{FG}   \\
 &\quad 
  + \frac{2(n-[nt]) ([n\tau^\ast] -1)}{n([nt]-1)}  \frac{1}{\sqrt{n}} \sum_{i=1}^{[n\tau^\ast]} (h_F(X_i)-\theta_F) \\
 &\quad + \frac{2(n-[nt]) ([nt]-[n\tau^\ast]) }{n([nt]-1)}  \frac{1}{\sqrt{n}}\sum_{i=1}^{[n\tau^\ast]} (h_G(X_i) -\theta_{FG} ) \\
 &\quad  + \frac{2(n-[nt]) [n\tau^\ast] }{n([nt]-1)}  \frac{1}{\sqrt{n}} \sum_{j=[n\tau^\ast]+1}^{[nt]} (h_F(Y_j) -\theta_{FG})\\
 &\quad+ \frac{2([nt]-[n\tau^\ast]-1)(n-[nt])}{n([nt]-1)}  \frac{1}{\sqrt{n}}\sum_{j=[n\tau^\ast] +1}^{[nt]} (h_G(Y_j)-\theta_G) \\
  & \quad - \frac{2[nt]}{n}   \frac{1}{\sqrt{n}} \sum_{i=[nt]+1}^n (h_G(Y_i)-\theta_G) +O_P(1/\sqrt{n}),
\end{align*} 
The same argument as above yields convergence of this process to the limit stated in the formulation of the theorem, for $t\geq \tau^\ast$. As the limit process is continuous in $t=\tau^\ast$, we finally obtain convergence of the process 
$ \frac{1}{\sqrt{n}} ( [nt] (U_{1:[nt]}   -U_{1:n}) -n\, \Psi_1(t) )_{0\leq t\leq 1}$. 
This completes the proof of Proposition \ref{prop:procfixed}.
\end{proof}


\begin{proof}[Proof of Theorem \ref{th:limitfixed}] 
As we are in the independent-sequences case, recall that the condition of Lemma \ref{lem:fixed:independent} implies Assumption \ref{ass:fixed}.
Since the proofs for the \FF\ and the \FL\ approach are essentially the same, we only give the details for one of them. 
Without loss of generality, we can assume that $\Psi_1(\tau^\star)>0$. By the continuity of $\Psi_1$, this will also hold in a neighborhood of $\tau^\star$. Furthermore, by Theorem \ref{th:fixed}, $D_n^F([nt])$ will also be positive with probability going to 1. First note that
\begin{multline*}
\big(T_1(h)-\sqrt{n}|\Psi_1(\tau^\star)|\big)=\frac{1}{\sqrt{n}}\Big(\sup_{t\in[0,1]}D_n^F([nt])-n\Psi_1(\tau^\star)\Big)\\
\geq \frac{1}{\sqrt{n}}\big(D_n^F([n\tau^\star])-n\Psi_1(\tau^\star)\big)
\end{multline*}
By Proposition \ref{prop:procfixed}, the right hand side converges in distribution to $G_1(\tau^\star)$, which is a centered normal random variable. A short calculation gives the variance as stated in Theorem \ref{th:limitfixed}. We still have to show that there also exists an upper bound which converges to the same limit distribution. For this note that 
\begin{equation*}
T_1(h)=\frac{1}{\sqrt{n}}D_n^F([n\hat{\tau}_{1,n}])
\end{equation*}
because $\hat{\tau}_{1,n}=\argmax_t |D_n^F([nt])|$, see (\ref{eq:location_estimator}), so
\begin{multline*}
\frac{1}{\sqrt{n}}\big(T_1(h)-n|\Psi_1(\tau^\star)|\big)=\frac{1}{\sqrt{n}}\Big(D_n^F([n\hat{\tau}_{1,n}])-n\Psi_1(\tau^\star)\Big)\\
\leq \frac{1}{\sqrt{n}}\Big(D_n^F([n\hat{\tau}_{1,n}])-n\Psi_1(\hat{\tau}_{1,n})\Big)
\end{multline*}
as, by Proposition \ref{prop:max}, $|\Psi_1(t)|$  has its maximal value at $t=\tau^\star$. We know from  the argmax theorem that $\hat{\tau}_{1,n}\rightarrow \tau^\star$ in probability. So for any $\epsilon>0$ with probability going to 1
\begin{equation*}
\Big(D_n^F([n\hat{\tau}_{1,n}])-\sqrt{n}\Psi_1(\hat{\tau}_{1,n})\Big)\leq \sup_{|t-\tau^\star|\leq \epsilon}\frac{1}{\sqrt{n}}\Big(D_n^F([nt])-n\Psi_1(t)\Big).
\end{equation*}
By the continuous mapping theorem, we have the convergence
\begin{equation*}
\sup_{|t-\tau^\star|\leq \epsilon}\frac{1}{\sqrt{n}}\Big(D_n^F([nt])-n|\Psi_1(t)\Big)\rightarrow \sup_{|t-\tau^\star|\leq \epsilon}G_1(t)
\end{equation*}
in distribution. By the continuity of the limit process, we have $\sup_{|t-\tau^\star|\leq \epsilon}G_1(t)\rightarrow G_1(\tau^\star)$ for $\epsilon\rightarrow 0$, which completes the proof of Theorem \ref{th:limitfixed}.

\end{proof}


\section{Proofs and further results for Section \ref{sec:examples}}
\label{app:examples}

\begin{proof}[Proof of Proposition \ref{prop:gmd}]
We will show 
\be \label{eq:rhoFG:nonnegative}
  \rho_{FG} = - \frac{1}{2} \int\int |x-y| \{f(x)-g(x)\}\{f(y)-g(y)\} dx dy \ge 0
\ee
using properties of characteristic functions. Let $a > 0$ be a parameter. The density
\[
	\phi_a(x) = \frac{1}{\pi} \frac{1-\cos(a x)}{a x^2}, \qquad x \in \R,
\]
has the characteristic function
\[
	\chi_a(t) = \left( 1 - \frac{|t|}{a} \right) \Varind{[-a,a]}(t), \qquad t \in \R,
\]
see \citet[][p.~503]{feller:1971} or \citet[][p.~271]{chow:teicher:1978} Hence, by Bochner's theorem \citep[][p.~623]{feller:1971}, the function $\chi_a$ is positive definite. 
This implies 
\be \label{eq:gmd1}
	\iint\left( 1 - \frac{|x - y|}{a} \right) \Varind{[-a,a]}(x-y) k(x) k(y) dx dy \ge 0, 
\ee
where $k(x)$ is short for $(f(x) - g(x))\Varind{\{|x|\leq b\}}$ with $b=a/2$. Note that $|x-y|\leq a$ if $|x|,|y|\leq b$. Hence
\begin{multline*}
0 \leq a\iint\left( 1 - \frac{|x - y|}{a} \right) \Varind{[-a,a]}(x-y) k(x) k(y) dx dy \\
=a\int_{-b}^b \!\int_{-b}^b (f(x) - g(x))(f(y) - g(y)) dx dy  \ - \  \int_{-b}^b\! \int_{-b}^b  |x - y| (f(x) - g(x))(f(y) - g(y)) dx dy
\end{multline*}
for all $a > 0$. The finite first moments of $F$ and $G$ imply that the first term converges to 0 as $a \to \infty$:
\begin{multline*} 
  a\int_{-b}^b \!\int_{-b}^b (f(x) - g(x))(f(y) - g(y)) dx dy 
 = 2 b \big(P_F(|X|\leq b)-P_G(|X|\leq b)\big)^2\\
 = 2 b \big(P_F(|X|>b)-P_G(|X|> b)\big)^2 \leq \frac{2}{b} \big(E_F[|X|]+E_G[|X|]\big)^2 \rightarrow  0
\end{multline*}
by Markov inequality as $b=a/2\rightarrow\infty$. By virtue of the dominated converge theorem (again using $\int_\R | x k(x)| < \infty$), 
we conclude 
\begin{multline*}\rho_{F,G}=-\frac{1}{2} \int_\R\int_\R |x-y| (f(x)-g(x))(f(y)-g(y)) dx dy\\
=\lim_{b\rightarrow\infty} \Big(- \frac{1}{2}   \int_{-b}^b\! \int_{-b}^b  |x - y| \{f(x) - g(x)\} \{f(y) - g(y)\} \Big) \geq 0
\end{multline*}
This completes the proof of Proposition \ref{prop:gmd}. 
\end{proof}

%
%

\bigskip
\begin{proof}[Proof of Proposition \ref{prop:gmd:location}]
Consider the function
\begin{equation*}
h_1(y)=\int |x-y|dF(x).
\end{equation*}
It has a minimum in $y=0$ and is convex because for all $x$, the mapping $y\mapsto |x-y|$ is convex. We conclude that $h_1(y)$ is non-increasing for $y<0$ and non-decreasing for $y>0$. Let $c>1$, then
\begin{multline*}
\theta_F=\iint |x-y|dF(x)dF(y)=\int h_1(y)dF(y)\leq \int h_1(cy)dF(y)\\
=\int h_1(y)dG(y)=\iint |x-y|dF(x)dG(y)=\theta_{FG}
\end{multline*}
and with the same reasoning $\theta_{FG} \leq \theta_G$. The other case $c<1$ follows by changing the roles of $F$ and $G$, and we always have $\theta_{FG} \in [\min\{ \theta_F, \theta_G \}, \max\{ \theta_F,\theta_G \}]$.
\end{proof}

%
%

\bigskip
\begin{proof}[Proof of Proposition \ref{prop:var}]
Recall $h(x,y) = (x - y)^2/2$. Let $X \sim F$ and $Y \sim G$ be independent and let $\mu = E(X)$ and $\nu = E(Y)$. Then 
$\theta_F = \Var(X)$, $\theta_G = \Var(Y)$, and
\[
	\theta_{FG} = \frac{1}{2} E (X-Y)^2
	 = \frac{1}{2} E \Big\{ (X-\mu) - (Y-\nu) + (\mu-\nu) \Big\}^2
\]\[
	= \frac{1}{2} E [ (X-\mu) - (Y-\nu) ]^2 +  E [(X-\mu) - (Y-\nu)](\mu-\nu)  + \frac{1}{2} (\mu-\nu)^2 
\]\[
	= \frac{1}{2} [ \Var(X) + \Var(Y) ]^2 + \frac{1}{2} (\mu-\nu)^2 
\]
and hence
\[
	\rho_{FG} = \theta_{FG} - \frac{1}{2}(\theta_F + \theta_G) 
	= \frac{1}{2} (\mu-\nu)^2,
\]
which completes the proof of Proposition \ref{prop:var}.
\end{proof}

The analogous formula for the sample covariance, as given in Section \ref{sec:var} of the paper, is obtained likewise: 
Let $(X,Y)$ and $(X',Y')$  be independent random vectors with bivariate distributions $F$ and $G$, respectively, and consider the kernel
$h((x,y),(x',y')) = (x-x')(y-y')/2$. 
Defining $\mu_F = E(X)$, $\nu_F = E(Y)$, $\mu_G = E(X')$, and $\nu_G = E(Y')$, we have
$\theta_F = \Cov(X,Y)$, $\theta_G = \Cov(X',Y')$, and
\begin{align*}
\theta_{FG}&=E h \{ (X,Y),(X^\prime,Y^\prime)\} =  \frac{1}{2} E\{(X-X^\prime)(Y-Y^\prime)  \}\\
 &=\frac{1}{2} E\big\{\big((X-\mu_F)-(X^\prime-\mu_G) + (\mu_F-\mu_G)\big)\big((Y-\nu_F)-(Y^\prime-\nu_G)+(\nu_F-\nu_G)\big) \big \}\\
 &=\frac{1}{2} \big(\Cov(X,Y)+\Cov(X^\prime,Y^\prime) +(\mu_F-\mu_G)(\nu_F-\nu_G)\big), 
\end{align*}
as all the mixed terms vanish by independence of $(X,Y)$ and $(X^\prime,Y^\prime)$. Thus, the eccentricity  is given by
\[
\rho_{FG}=\theta_{FG}-\frac{1}{2} (\theta_F+\theta_G) =\frac{1}{2} (E_FX-E_GX)(E_FY-E_GY).
\]


%
%

\bigskip
The following example illustrates Corollary \ref{cor:consistency} for Gini's mean difference. A change in location only may lead to a non-zero eccentricity $\rho_{FG}$ and hence both tests will detect this change with asymptotic probability 1. 

\begin{example}\label{ex:gmd:change-in-mean} 
Let $h(x,y)=|x-y|$, so the corresponding $U$-statistic is Gini's mean difference. Let $F$ be the distribution function for the uniform distribution on $[0,1]$ and $G$ be the uniform distribution function on $[1,2]$. Then
\begin{equation*}
\theta_F=\theta_G=\int_0^1\int_0^1|x-y|dxdy=\frac{1}{3}
\end{equation*}
but
\begin{equation*}
\theta_{FG}=\int_0^1(1-x)dx+\int_1^2(y-1)dy=1
\end{equation*}
and $\rho_{FG}=2/3$. The reasoning can be extended to a suitable class of distributions, e.g. $F$ having a symmetric density $f$ whose support is a connected set, and $G$ is a shifted version of $F$.
\end{example}

%
%

\bigskip
The next example illustrates Proposition \ref{prop:max} for Gini's mean difference. It describes a situation where both tests are consistent, i.e., they detect the change with asymptotic probability 1, but both associated argmax estimators of the change-point are not consistent.
\begin{example} 
\label{ex:gmd:location}
Let $h(x,y)=|x-y|$, let $F$ be the distribution function for the uniform distribution on $[0,1]$ and $G$ be the distribution function for the uniform distribution on $[1,3]$. Some short calculations give $\theta_F=1/3$, $\theta_G=2/3$, $\theta_{FG}=3/2$ and $\rho_{FG}=1$. For $\tau^\star=1/2$, we obtain
\begin{equation*}
\Psi_2(t) =\begin{cases}t\frac{1}{2}\frac{1}{3}-2\frac{t(\frac{1}{2}-t)}{1-t}\frac{1}{2}=\frac{5t^2-2t}{6(1-t)} &\text{ for } t<\tau^\star, \\
(1-t)\frac{1}{2}\frac{1}{3}+2\frac{(t-\frac{1}{2})}{t}\frac{1}{2}-2(t-\frac{1}{2})\frac{1}{2}=\frac{-7t^2+10t-3}{6t} &\text{ for } t\geq\tau^\star. \end{cases}
\end{equation*}
We have $\Psi_2(1/2)=1/12$, but the maximum of the function is at $t=\sqrt{3/7}$ with $\Psi_2(\sqrt{3/7})=5/3-\sqrt{7/3}>1/12$.
So we find the change-point location estimator 
$\hat{\tau}_{2,n}=\operatorname*{argmax}_{t\in[0,1]}\left|D_n^L([nt])\right|$ is inconsistent in this case. 
The reasoning for $\hat{\tau}_{1,n}$ runs analogously. 
\end{example}

%
%

\bigskip
\begin{proof}[Proof of Proposition \ref{prop:kendall:symmetry}] Recall that $
h((x_1,y_1),(x_2,y_2)) = \Ind{(x_2-x_1)(y_2-y_1) > 0}-\Ind{(x_2-x_1)(y_2-y_1) < 0},$ hence
\begin{align}
h\{(x_1,-y_1),(x_2,-y_2)\} &= - h\{(x_1,y_1),(x_2,y_2)\}\label{eq:antisym}, \\
h\{(x_1,y_1),(x_2,y_2)\} &= h\{(x_2,y_2),(x_1,y_1)\}\label{eq:kendalsym}.  
\end{align}
We show in the following: (I) $\theta_F = - \theta_G$ and (II) $\rho_{FG} = 0$.  Let $(X_1,Y_1)$, $(X_2,Y_2)$ be two independent random vectors with distribution $F$, so both $(X_1,-Y_1)$ and $(X_2,-Y_2)$ have distribution $G$. \\ 
Part (I): By \eqref{eq:antisym}
\begin{equation*}
\theta_G=E\left[h\{(X_1,-Y_1),(X_2,-Y_2)\}\right]\\
=-E\left[h\{(X_1,Y_1),(X_2,Y_2)\}\right]=-\theta_F
\end{equation*}
Part (II): We use \eqref{eq:antisym} again
\begin{equation*}
\theta_{FG}=E\left[h\{(X_1,Y_1),(X_2,-Y_2)\}\right]=-E\left[h\{(X_1,-Y_1),(X_2,+Y_2)\}\right].
\end{equation*}
Now \eqref{eq:kendalsym} leads to
\begin{equation*}
\theta_{FG}=-E\left[h\{(X_1,-Y_1),(X_2,+Y_2)\}\right]=-E\left[h\{(X_2,Y_2),(X_1,-Y_1)\}\right]=-\theta_{FG}
\end{equation*}
and we can conclude that $\theta_{FG}=0$. Furthermore by (I), $\theta_F+\theta_G=0$, so  $\rho_{FG}=\theta_{FG}-\frac{1}{2}(\theta_F+\theta_G)=0$.
\end{proof}

\section{Details on long-run variance estimation} 
\label{app:data}

The p-values reported in Section \ref{sec:data} are obtained by comparing the studentized test statistics to the Kolmogorov distribution, i.e., the distribution of $\sup_{0 \le t \le 1} | B(t)|$, where $B$ is a Brownian bridge. This means, we divide both test statistics $T_1(h)$ and $T_2(h)$, (see (\ref{eq:cusum:ustat1})) by the following estimate of the long-run variance $\sigma^2_h$ (see (\ref{eq:sigma2})):
\[
	\hat{\sigma}_{h}^2 = \, 4 \! \sum_{k=-(n-1)}^{n-1} W\left(\frac{|k|}{b_n}\right)
	\frac{1}{n}\sum_{i=1}^{n-|k|} \hat{h}_1(X_i) \hat{h}_1(X_{i+|k|}), 
\]
where $\hat{h}_1(x) = n^{-1}\sum_{i=1}^{n} h(x,X_i) - U_n$, furthermore $W$ is a suitable kernel function, downweighting higher-lag covariance terms, and $b_n$ a sample-size dependent bandwidth parameter. We follow here \citet{dehling:vogel:wendler:wied:2017}, who adopt the kernel estimation by \citet{dejong:2000} for U-statistics and give general conditions on $W$ and $b_n$. 

Throughout the current paper, we use the Bartlett (or triangular) kernel, i.e., 
\[
	W(x) = (1 - |x|)\Ind{|x| \le 1}.
\]

\begin{figure}[t]%
\centering
\includegraphics[width=0.9\columnwidth]{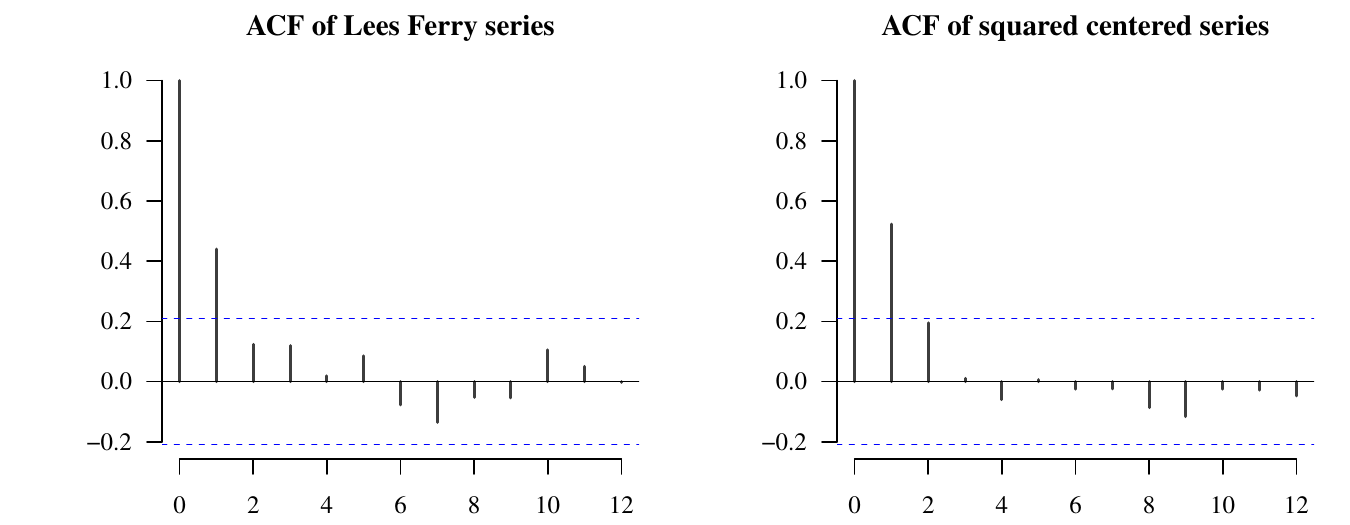}%
\caption{Auto-correlation function of the Lees Ferry times series (left) and the squared mean-centered Lees Ferry time series (right)}%
\label{fig:app:1}%
\end{figure}

\begin{figure}[t]%
\centering
\includegraphics[width=0.9\columnwidth]{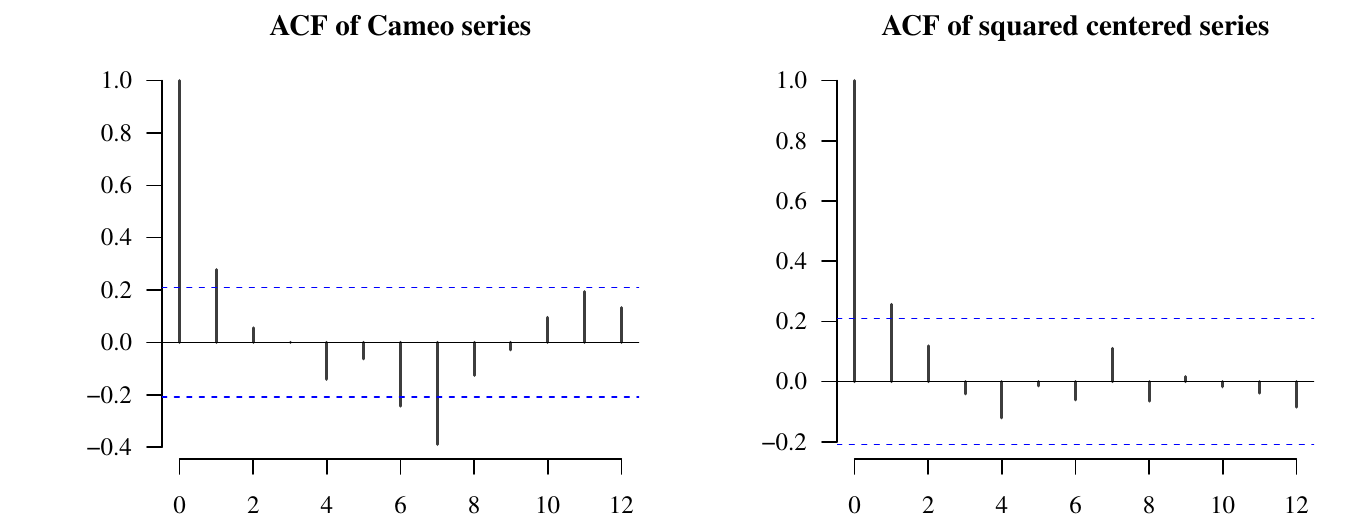}%
\caption{Auto-correlation function of the Cameo times series (left) and the squared mean-centered Cameo time series (right)}%
\label{fig:app:2}%
\end{figure}

The optimal bandwidth selection problem for long-run variance estimation has itself produced a substantial body of literature, see, e.g., \citet{belotti:2023} and the references therein. 
Our choice of $b_n = 2$ in the data example is motivated by a simple hands-on criterion: Figures \ref{fig:app:1} and \ref{fig:app:2} show the auto-correlation function of both time series along with the auto-correlation function of the squared mean-centered time series. This gives a rough indication that serial dependencies of lag 2 or higher are negligible for this series.


\bibliographystyle{abbrvnat} 
\bibliography{DVW-References}

\end{document}